\documentclass[3p]{elsarticle}
\pdfoutput=1
\usepackage[utf8]{inputenc}
\usepackage[T1]{fontenc}
\usepackage{amsmath,amsfonts,amsthm,amssymb}
\usepackage[usenames]{color}
\usepackage[english]{babel}
\usepackage[overload]{empheq}
\usepackage{mathrsfs}
\usepackage[normalem]{ulem}
\usepackage[font=footnotesize]{subcaption}
\usepackage{multirow}
\usepackage{enumitem}

\newlist{Steps}{enumerate}{1}
\setlist[Steps]{label=Step~\arabic*:}

\newcommand{\R}{\mathbb{R}}

\newcommand{\N}{\mathbb{N}}
\newcommand{\Z}{\mathbb{Z}}

\newtheorem{theorem}{Theorem}[section]
\newtheorem{lemma}[theorem]{Lemma}

\newtheorem{remark}[theorem]{Remark}
\numberwithin{equation}{section}

\makeatletter
\def\ps@pprintTitle{%
 \let\@oddhead\@empty
 \let\@evenhead\@empty
 \def\@oddfoot{}%
 \let\@evenfoot\@oddfoot}
\makeatother

\begin{document}
\begin{frontmatter}
\date{\today}
\author[umn]{P.~Cazeaux}\ead{pcazeaux@umn.edu}
\author[epfl]{J.S.~Hesthaven}
\ead{jan.hesthaven@epfl.ch}
\address[umn]{School of Mathematics, University of Minnesota,\\ Minneapolis, MN 55455, USA}
\address[epfl]{Ecole F\'ed\'erale Polytechnique de Lausanne,\\ CH-1015 Lausanne, Switzerland}

\title{Projective multiscale time-integration for electrostatic particle-in-cell methods.}

\begin{abstract}
The simulation of problems in kinetic plasma physics are often challenging due to strongly coupled phenomena across multiple scales. In this work, we propose a wavelet-based coarse-grained numerical scheme, based on the framework of Equation-Free Projective Integration, for a kinetic plasma system modeled by the Vlasov-Poisson equations. A kinetic particle-in-cell (PIC) code is used to simulate the meso scale dynamics for short time intervals. This allows the extrapolation over long time-steps of the behavior of a coarse wavelet-based discretization of the system. To validate the approach and the underlying concepts, we perform two 1D1V numerical experiments: nonlinear propagation and steepening of an ion wave, and the expansion of a plasma slab in vacuum. The direct comparisons to resolved PIC simulations show good agreement. We show that the speedup of the projective integration scheme over the full particle scheme scales linearly with the system size, demonstrating efficiency while taking into account fully kinetic, non-Maxwellian effects. This suggests that the approach is potentially interesting for kinetic plasma problems with a large separation of scales, especially in higher dimensions.
\end{abstract}
\end{frontmatter}

\section{Introduction}

The importance and complexity of plasma physics makes its simulation a crucial challenge across various domains of science, from astrophysics (stars, solar wind) to efforts in developing nuclear fusion reactors (tokamaks) to everyday human environment (lighting, industrial processes). 
Plasma phenomena are typically characterized by a complex multiscale character, as physics at scales separated by multiple orders of magnitude in space and time directly influences the global behavior of the plasma. 
Examples of such phenomenons include magnetic reconnection in the solar wind, in which large amounts of magnetic energy are released into the plasma by the breaking and reforming of magnetic field lines at extremely small length scales; or meso-scale plasma turbulence, which is known to have a significant influence on the global transport properties of plasmas in laboratory tokamaks. 
The direct modeling of such problems still defies numerical methods developed by computational scientists.

In this paper we study the application of an original multiscale computational technique, the Equation-Free Projective Integration method (EFPI)~\cite{Gear_2002, Kevrekidis_2003, Gear_2003, Kavousanakis_2007}, and its application to the simulation of plasmas.
A commonly used approach to simulate strongly multiscale systems is to derive and discretize a set of reduced equations for macro-scale phenomenons by analytically averaging the contributions of smaller scales, e.g., using homogenization theory.
This procedure typically relies on strong structural assumptions such as periodicity~\cite{Papanicolau_1978}, or stationary randomness.
In the context of plasmas, reduced fluid models such as magneto-hydrodynamics can sometimes be used depending on the physical setting. The EFPI framework, on the other hand, does not rely on discretizing a closed set of reduced equations.
Rather, full-scale and highly resolved simulations of the system are computed for short period of times, allowing the
recovery of the slow macroscopic behavior which is then projected forward in time with large timesteps. This approach seeks to
account for the effect of meso-scale physics on the macroscopic behavior even when no satisfying upscaling theory is known, assuming only that there is a sufficient separation of scales.
Whether this approach can provide accurate, efficient simulations for kinetic plasma problems remains an open question and a question that we shall consider in this work.

In this paper we are particularly concerned with the simulation of collisionless plasmas for which a kinetic description is required.
In this context, particle methods, in particular Particle-In-Cell methods~\cite{Birdsall_Langdon_1985} have traditionally been
preferred to Eulerian (grid-based) methods as they allow a coarse, but reasonably precise description of phase-space. 
However, standard explicit PIC schemes impose stability conditions to guarantee that space and time steps be sufficiently small to resolve the Debye length and plasma periods which are the typical space and time scales at which the electrostatic force tends to restore local charge neutrality in the plasma. For the majority of applications, these are extremely small when compared 
to the typical scales of the plasma under investigation.
The development of implicit PIC methods, which are free of such constraints, has been the focus point of many recent efforts. 
To avoid the expensive fully implicit resolution of the particle positions and the fields, alternative algorithms have been developed, e.g., the direct implicit method~\cite{Denavit_1981, Cohen_1982, Langdon_1983}
where a two-step predictor-corrector approximation is implemented or semi-implicit methods~\cite{Ricci_2002, Lapenta_2006, Markidis_2010}.
In the implicit moment methods~\cite{Mason_1981, Mason_1983}, moment equations are used to predict the value of the fields.
Alternatives include the hybrid method, see e.g.~\cite{Lipatov_2002}, where a particle model is used for some of the particle species and a fluid model for others; or the so-called \textit{multiscale} PIC methods~\cite{Friedman_1991,Parker_1993} which use different step sizes for different spatial regions, depending on the electric gradient scale in each region.
Recently, fully implicit PIC methods based on Jacobian-free iterative algorithms have been developed~\cite{Chen2011,Markidis2011,Chen2014,Chen2015,Chacon2016} and allow to nonlinear converge particle and fields, resulting in an overall computational speedup.
Another recent development is the asymptotic-preserving methods~\cite{Degond_2006, Degond_2010} that allow to better account for the stiffness due to the quasineutrality constraint.

We shall pursue the development of a different method, utilizing the equation-free projective integration method with an explicit PIC code as the fine-scale solver. The basic idea is to advance a macroscopic, or coarse, representation of the solution by initializing a consistent mesoscopic, or fine-scale, state of the system (lifting operation). This fine description is then advanced for a number of time-steps of a fine-scale solver, here the PIC code, recording the evolution of the system at the coarse level by means of a restriction operator. A linear least squares fit is then computed for each coarse variable to evaluate its temporal derivative. The coarse representation is then projected forward in time using a much larger time step. 

 A first attempt at implementing an equation free procedure in this context was the EFREE procedure developed by Shay {\it et al.}~\cite{Shay_2007} and applied to the nonlinear propagation and steepening of an ion-acoustic wave. In EFREE, both electron and ion distribution functions are assumed to be Maxwellian and the electrons are adiabatic. The coarse description is then simply the ion density, mean velocity and pressure discretized on a coarse grid. Despite neglecting all kinetic effects, the EFREE results are promising as the wave propagation is correctly reproduced during the initial phase. However as the wave steepens, kinetic aspects such as trapping and non-Maxwellian phase space structures appear and the assumptions of the EFREE procedure fails to adequately account for these. An effort to include such effects was proposed in~\cite{Maluckov_2008}, by using a coarse discretization of the full inverse cumulative ion distribution function. However, this approach did not show a strong linear scaling for the speedup with the ion acoustic wavelength as the EFREE procedure achieved. Achieving such a scaling appears necessary to establish the usability of the equation free projection integration paradigm in kinetic plasma simulations.
 
In this paper we propose a new method based on representing the ion distribution function by the use of a multiresolution wavelet analysis. 
More precisely, the distribution function is discretized using a coarse grid in physical space and a wavelet basis in velocity space. 
This allows to fully account for kinetic effects by tracking the evolution of macro-scale phase-space structures. 
The motivation for introducing a multiresolution wavelet analysis, widely used in signal processing, is its multiscale nature which allows the local adaption driven by the smoothness of the function to be approximated~\cite{Mallat_1999}. 
One of their most successful applications has been signal denoising~\cite{Donoho_1996}. Wavelet-based density estimation (WBDE) has recently been applied to the estimation of particle distribution functions obtained from plasma simulations with PIC methods~\cite{Nguyen_2010}.

The noise reduction is typically obtained by thresholding the expansion coefficients in a wavelet basis, based on the fact that this expansion is assumed to be sparse, or more precisely, the signal is assumed to be well-approximated by a small number of large expansion coefficients. Note that the sparsity of the representation is naturally connected to the coarse-graining of the data in phase-space.
In this first step of the wavelet-based projective integration method, we limit attention to a non-adaptive approach, i.e., the wavelet representation is truncated at some pre-defined level of resolution. In the future, an adaptative denoising algorithm could be used to improve on this aspect, allowing for data-driven improvement of the resolution in specific regions of the simulation grid~\cite{Luisier_2010}.

The two fundamental properties, sparsity and denoising, appear particularly well suited for application in the equation-free projective integration framework. 
First, depending on the bin size used for estimation of the ion distribution function, important random fluctuations can appear in the measured coefficients, especially in phase space regions with low particle count. The smoothing operation from fine to coarse grid is typically not sufficient to correct for the strong noise in the simulations~\cite{Shay_2007, Maluckov_2008}. On the other hand, the wavelet representation enables the computation of a much smoother approximation as we will show.

In addition, an important concern is the handling of PIC-generated statistical noise by the projective integration procedure. 
We show that by combining the restriction on a coarse grid in physical space, wavelet thresholding techniques in velocity space, quiet starts and projection along approximated characteristics, our proposed method reduces noise dramatically. 

We test numerically our proposed method on two well-known one-dimensional problems: the propagation and steepening of an ion acoustic wave, as in~\cite{Shay_2007, Maluckov_2008}; and the expansion of a Gaussian plasma in a vacuum, following~\cite{Mora2005}.
Let us note that at the moment, our method is not necessarily more efficient for either of these problems than a traditional implicit or multiscale PIC methods described above.
Rather, its potential lies in the potential of its extension to strong multiscale problems for which more conventional methods fail. 
However the method must first be developed and evaluated in simple settings, for which trusted results exist.
Our numerical experiments confirm that reference solutions obtained by brute-force high resolution PIC simulations are well reproduced by the proposed wavelet-based method. By increasing resolution, the ion pressure and phase space features converge correctly unlike previous attempts~\cite{Shay_2007}.

While we have restricted here our application to one-dimensional examples, the method is in principle not restricted to such cases. However, adaptativity in phase space seems to be an essential future development to beat the curse of dimensionality. Bypassing the use of a fine grid discretization for sampling and loading particles could also prove necessary, although this can be done in parallel for each spatial grid cell. Note that when high-frequency structure start to dominate the physics, separation of scales does not apply. This could be detected by the adaptative filtering, at which point one should revert to traditional approach, e.g. pure particle-in-cell with no upscaling.

The paper is organized as follows. In Section~\ref{sec:presentation}, we present the two-fluid Vlasov-Poisson model and the classical PIC scheme. In Section~\ref{sec:EFPI}, we recall the equation-free projective integration (EFPI) framework and the original EFREE method developed in~\cite{Shay_2007}. From this approach, we develop our proposed wavelet-based method. Section~\ref{sec:numresults} is devoted to the presentation and discussion of the numerical results and a comparison between the brute-force explicit PIC and our proposed schemes is provided. These results show that the scheme is stable and able to deal with multiscale kinetic problems where the stiffness comes from the fast electron scale.

\section{Presentation of the problem}\label{sec:presentation}
\subsection{The Vlasov-Poisson equations}

The theoretical modeling of a two-species plasma, comprising electrons and one ion species, usually begins with the two-fluid Vlasov--Poisson system. Each species of ions and electrons is described by a distribution function, respectively $f^\mathrm{i}(x,v,t)$ and $f^\mathrm{e}(x,v,t)$. The superscripts $\mathrm{i}$ and $\mathrm{e}$ will refer to the ions and electrons, respectively, throughout the whole paper. The position and velocity variables $x$, $v$ are such that $(x,v) \in \Omega \times \R^d$ with $\Omega \subset \R^d$, $d \leq 3$, and $t\geq 0$ is the time variable. We restrict our attention to the one dimensional case, with $d=1$, $\Omega = [0,L]$ where $L > 0$ is the size of the system, and periodic boundary conditions are imposed. However, the majority of the developments proposed here can be generalized to the full six dimensional Vlasov--Poisson problem.

The Vlasov-Poisson equations is expressed as
\begin{equation}\label{eq:vlasov}
\left \{ 
\begin{aligned}
\partial_t f^\mathrm{i} + v \cdot \nabla_x f^\mathrm{i} - \frac{e}{m_\mathrm{i}} (\nabla_x \phi) \cdot \nabla_v f^\mathrm{i} = 0, \\
\partial_t f^\mathrm{e} + v \cdot \nabla_x f^\mathrm{e} + \frac{e}{m_\mathrm{e}} (\nabla_x \phi) \cdot \nabla_v f^\mathrm{e} = 0, \\
\end{aligned}
\right.
\end{equation}
where we denote $e>0$ as the positive elementary charge, by $m_\mathrm{i}$ and $m_\mathrm{e}$, respectively, the ion and electron masses, and by $\phi$ the electric potential. The electrostatic potential $\phi$ is recovered by 
\begin{equation}\label{eq:poisson1}
-\Delta \phi = \frac{e}{\varepsilon_0} (n^\mathrm{i} - n^\mathrm{e})
\end{equation}
where $\varepsilon_0$ is the vacuum permittivity, and $n^\mathrm{i}$, $n^\mathrm{e}$ are the ion and electron densities, respectively, given as
\begin{equation*}
n^\mathrm{i}(x,t) = \int_{\R^d} f^\mathrm{i}(x,v,t) \mathrm{d}v, \qquad n^\mathrm{e}(x,t) = \int_{\R^d} f^\mathrm{e}(x,v,t) \mathrm{d}v.
\end{equation*}

Two very important physical scales for this model are the Debye length $\lambda$ and the electron plasma frequency $\omega_p$, given as:
\begin{equation}\label{def:debye}
\lambda = \sqrt{\frac{\varepsilon_0 k_B T^\mathrm{e}_0}{e^2 n_0}} \qquad \omega_p = \sqrt{ \frac{n_0 e^2}{\varepsilon_0 m_\mathrm{e}}},
\end{equation}
where $n_0$ is the plasma density scale, $T_0^\mathrm{e}$ is the electron temperature scale and $k_B$ is the Boltzmann constant. In problems of interest, both the Debye length and the electron plasma period become very small compared to the macroscopic scales. 

\subsection{Classical particle-in-cell scheme} \label{sec:PIC}
In order to numerically solve the Vlasov-Poisson equations, the explicit particle-in-cell (PIC) method, discussed at length in~\cite{Birdsall_Langdon_1985}, is widely used due to its strong parallel scaling and relatively low computational cost compared to traditional Eulerian methods. This is particularly true for problems in high dimensions.
The basic idea of the PIC method is to discretize the phase-space distribution function $f^{i,e}$ with weighted macro-particles:
\begin{equation}\label{def:empiricaldistribution}
f^{i,e}(t,x,v) \approx f_{N_p}^{i,e} = \sum_{j=1}^{N_p} \omega^{i,e}_j \delta(v-V^{i,e}_j) \psi(h^{-1}(x-X^{i,e}_j)),
\end{equation} 
where $N_p$ is the number of particles, $\psi$ is a shape function, $X^{i,e}_j(t)$, $V^{i,e}_j(t)$ is the location of the $j$-th particle in phase space and $\omega^{i,e}_j$ is a weight which is defined at initialization. One should not consider each particle as a physical particle but rather as particle clouds. Each of these particles follows Newton's equations:
\begin{equation}
\begin{aligned}
\dot{X^\mathrm{i}_j}(t) &= V^\mathrm{i}_j(t), &\dot{V^\mathrm{i}_j}(t) &= - \frac{e}{m_\mathrm{i}} \nabla_x \phi_{{N_p},h} \left ( X^\mathrm{i}_j(t), t \right ), \\
\dot{X^\mathrm{e}_j}(t) &= V^\mathrm{i}_j(t), &\dot{V^\mathrm{e}_j}(t) &= \frac{e}{m_\mathrm{e}} \nabla_x \phi_{{N_p},h} \left ( X^\mathrm{e}_j(t), t \right ).
\end{aligned} 
\end{equation}
These equations are then discretized in time. The potential $\phi_{{N_p},h}$, approximating the electrostatic potential $\phi$, is determined at the particle locations by the following procedure, repeated at each time step. Given a fixed spatial grid of space step $h$, the nodal values of the density of each species is computed by projecting every particle to the grid nodal locations with appropriate weights, given by the shape of the macroparticle~\cite{Birdsall_Langdon_1985}. The total charge density is then used to determine the nodal values of the electrostatic potential on the mesh by solving the Poisson equation~\eqref{eq:poisson1} numerically, using either finite differences or the fast Fourier transform. Field values at each particle location are then determined by interpolation.

\subsection{Advantages and limitations}
Due to their efficiency and simplicity, PIC-based numerical methods are a dominating tool used in the modeling of kinetic plasma physics phenomena. 
 While the mathematical literature on PIC methods is limited, some a priori convergence results have been obtained, see e.g.~\cite{Cottet_1984, Ganguly_1989, Wollman_1996}. 
Typically, the error scales with the number of particles, as~$1/\sqrt{{N_p}}$ as would be expected. 
The independence of this rate with respect to dimension is a great advantage of the method, since Eulerian grid-based methods become extremely costly in dimension 4 or higher. 
Nevertheless, it is well-known that to ensure stability, explicit PIC codes must resolve the local Debye length and the electron plasma period everywhere in the computational domain~\cite{Birdsall_Langdon_1985}.

The main drawback of PIC methods is the well-known statistical sampling noise due to the low number of numerical particles in each cell~\cite{Nevins_2005} - a number that by design is orders of magnitudes smaller than the true number of particles in the plasma. This so-called \textit{shot noise} has a significant impact on the use of PIC solvers as a fine-scale solver in the projective integration framework~\cite{Shay_2007, Maluckov_2008}.
Various noise-reduction techniques have been developed to reduce its impact on computations. We will in particular make use of the classical quiet start procedure~\citep{Birdsall_Langdon_1985,Wollman_1996}, which helps to ensure that the initial noise is small. Instead of sampling randomly, similar to that of a Monte Carlo method, or on a regular grid in phase space as studied in~\cite{Cottet_1984}, a given number of particles are initialized in each cell: the (conditional) inverse cumulative distribution function of the distribution is used to assign the particle velocities, and the particle locations inside the cell are scrambled using a pseudo-random sequence to reduce correlations.

\section{Projective integration scheme}\label{sec:EFPI}
In this section, we present the equation-free projective integration (EPFI) framework and we develop a new method for the solution of plasma problems. The major issues in the design of such a multiscale scheme, based on the PIC method, are the choice of a set of appropriate macroscopic variables, and the handling of the PIC-generated statistical noise.

\subsection{Presentation of EFPI schemes}
The simulation of strongly coupled multiscale systems remains a major challenge due to the coupling of physics,
often known only at the meso-scale, and the dynamics of interest which may emerge at the macroscale. 
Often, direct computation of the solution of the original problem is out of reach due to the necessity for the discretization to resolve the smallest or fastest scales in the system. 
Traditional approaches for the simulation of systems that exhibit a significant separation of scales usually involve two stages. 
First, a set of reduced equations is derived to describe the system at the coarse or macroscopic scale. 
These reduced equations are solved numerically and their solutions analyzed. 
For plasmas, such equations exist in the form of fluid closures which reduce the dimensionality of the problem, such as the equations of magneto-hydrodynamics. However, these rely on assumptions which are not always satisfied and leave out important mesoscale physics~\cite{Krall_1986}.

When such closed-form equations are not available, the equation-free approach, originally proposed by Kevrekidis {\it et al.}~\cite{Gear_2002, Kevrekidis_2003, Gear_2003, Kavousanakis_2007}, sidesteps the need for an explicit reduced model. Instead, the approach relies on short bursts of fine-scale simulations to determine the time derivatives of a coarse representation of the system. 
Even though macro-scale level equations may not be necessary, it is worth stressing that an efficient application of the equation-free procedure requires the system to be parameterizable by a set of coarse variables for which the evolution evolves on a slow manifold~\cite{Kevrekidis_2003}. 
However, if the dynamic behavior of the system can be coarse-grained by a representation that is smooth at macroscopic scales in time and space, then it is conceivable that the full system need only be simulated for a small number of fine-scale time steps to advance the coarse variables over a much larger time interval.

The success of the method depends critically on the correct choice of the coarse description of the system, as we discuss below in more detail. 
Low dimensionality is clearly important as the primary purpose is computational speedup. 
A {\it restricting} operator, allowing the transition between the fine to the coarse representation of the data, as well as a {\it lifting} operator to go from coarse to fine scale, must be determined. Of course, the use of a coarse representation implies that some fine scale information is lost through the restricting process. This data must be reconstructed when lifting to initialize the fine-scale solver, typically through interpolation techniques and physical knowledge of the system.

 \begin{remark}
 Ultimately, a goal should be to develop error estimators to validate the fidelity of the coarse representation during the time integration. This is in particular the case in the context of plasma simulations, since one may not be able to assume that the relaxation properties on which the EFPI framework is based~\cite{Kevrekidis_2003, Gear_2003} hold at all times. For example, quasineutrality or Boltzmann equilibria can be violated as shocks develop or kinetic effects become dominant. Such techniques 
would allow the adaption of the macroscopic projection timestep as well as the coarse resolution, or revert to full-scale simulations when, for example, short transitions are detected during which no scale separation can be exploited.
\end{remark}

\begin{figure}[ht]
\centering
\includegraphics[width=.8\textwidth]{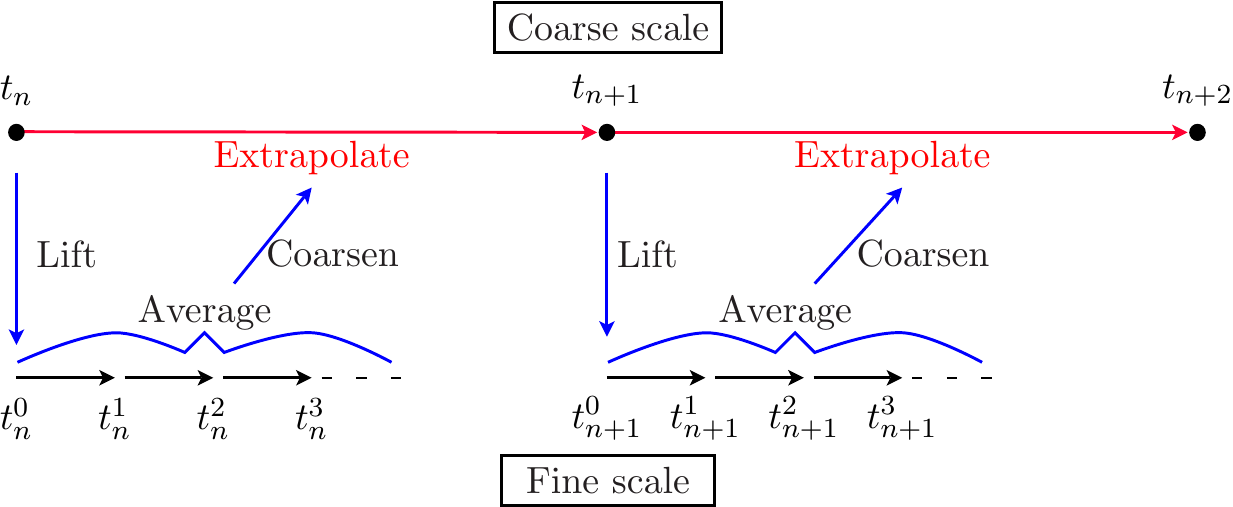}
\caption{Equation free projective integration cycle}\label{fig:efpicycle}
\end{figure}

The basic forward Euler projective integration cycle is presented in Figure~\ref{fig:efpicycle}. 
This cycle comprises three main parts:

\paragraph{Initialization} At the beginning of the step, only the coarse representation of the system is available. This data is used to recover, by the lifting operator, a consistent fine state, which is used to initialize the fine-scale time stepper.

\paragraph{Micro-scale evolution} After its initialization, the fine-scale time solver is advanced forward $N_e$ steps. After each step, the corresponding coarse variables are obtained by restricting the data.

\paragraph{Extrapolation} At the completion of the fine-scale steps, the evolution of the coarse variables is known at $N_e+1$ instants corresponding to the fine-scale time steps. 
For each of the degrees of freedom, the best fit for the rate of change of the corresponding variable $\psi$ is computed e.g. by a least squares algorithm. The macroscopic variables are
projected forward by the macroscopic timestep $\Delta t$ by the forward Euler formula
\begin{equation}\label{eq:extrapolation}
f(t+\Delta t) = f(t) + \Delta t \dot{f}(t).
\end{equation}
\begin{remark}
We note that it is possible to adapt this procedure to increase the accuracy of the time integration, i.e., we could also consider a trapezoidal leapfrog method as described in~\cite{Shay_2007}.
\end{remark}

\subsection{Choice of macroscopic variables: EFREE approach}\label{sec:efree}

The core of the method is the distinction between a coarse-scale (macroscopic) and a fine-scale (mesoscopic) model representations.
 In this study as in~\cite{Shay_2007, Maluckov_2008}, the fine-scale representation consists in a fully kinetic description of the plasma,
 resolving the smallest relevant scales of space and time and using an explicit PIC scheme. A coarse representation must then be carefully chosen to approximate the system. 
 
 We first discuss the original so-called \textit{EFREE} approach proposed in~\cite{Shay_2007}, to model ion-acoustic wave propagation and steepening. In EFREE, both ion and electron distribution functions are assumed to be shifted Maxwellian distributions, the electrons are assumed to be adiabatic, and the plasma is quasineutral. Thanks to these physical assumptions, the system is parameterized by the ion and electron densities $n^\mathrm{i,e}$, mean velocities $V^\mathrm{i,e}$ and temperatures $T^\mathrm{i,e}$ (or equivalently the pressures $P^\mathrm{i,e} = n^\mathrm{i,e} T^\mathrm{i,e}$). 

 The ion variables are chosen as \textit{active} variables, i.e. directly integrated in time. 
 Electron variables are \textit{passive} variables, determined from the ion variables. 
 Due to adiabaticity, the electron temperature is assumed to be constant at all times:
 \begin{equation}\label{eq:adiabaticity}
 T^\mathrm{e} = T^\mathrm{e}_0.
 \end{equation}
The electron density is a function of the electrostatic potential $\phi$:
 \begin{equation}\label{eq:boltzmann}
 n^\mathrm{e} = n_0 \mathrm{exp}\left ( \frac{e \phi}{k_B T^\mathrm{e}_0} \right ),
 \end{equation}
 where $n_0$ is a renormalization constant. Assuming that $n^\mathrm{i} \approx n^\mathrm{e}$ by quasineutrality, one deduces from~\eqref{eq:boltzmann} the potential $\phi$ as a function of $n^\mathrm{i}$. Using~\eqref{eq:poisson1} the electron density is then given as
 \begin{equation} \left \{
\begin{aligned}
\phi &= \frac{k_B T_0^\mathrm{e}}{e} \mathrm{log} \left ( \frac{n^\mathrm{i}}{n_0 } \right ), \\
 n^\mathrm{e} & = n^\mathrm{i} - \frac{\varepsilon_0}{e}\Delta \phi, 
 \end{aligned} \right.
 \end{equation}
 i.e.
 \begin{equation}\label{eq:explicitelectron}
 n^\mathrm{e} = n^\mathrm{i} - \lambda^2 \frac{\partial}{\partial x} \left ( \frac{n_0}{n^\mathrm{i}} \frac{\partial n^\mathrm{i}}{\partial x} \right ).
 \end{equation}
 Finally the electron mean velocity is recovered by
 \begin{equation}\label{eq:electronvelocity}
V^\mathrm{e} = V^\mathrm{i}.
 \end{equation}
The spatial discretization of the \textit{active} variables (in EFREE, $n^\mathrm{i}$, $V^\mathrm{i}$ and $P^\mathrm{i}$) uses a two-scale grid structure. 
The coarse representation of the data is defined as the values of the ion variables at nodes of a coarse grid, which resolves the macroscopic phenomenon.
The data is then lifted to the fine grid which corresponds to the mesh employed for the PIC algorithm. 
This grid typically has a much finer resolution as it must resolve the Debye length $\lambda$. 
The \textit{passive} electron variables are then determined using Eqs.~(\ref{eq:adiabaticity}--\ref{eq:electronvelocity}), and the particle positions are initialized by the PIC code.
 During the fine-scale evolution step of the EFPI procedure, the variables are projected on the fine-scale grid using the usual PIC weighting scheme. 
 
Increasing and decreasing the resolution by moving the data between the coarse and fine grids is realized in two steps:

\paragraph{Lifting operator}
Let $n_c$ be the number of nodes in the coarse grid, and $n_f$ the number of nodes in the fine grid. For simplicity, both numbers are assumed to be powers of $2$, and the interpolation procedure is done in successive factors of $2$, repeated over $l$ levels such that $n_f = 2^l n_c$. Let $f_k = f^{(0)}_k$, $k = 1, \dots, n_c$ be the values of a variable $f$ at the coarse grid points. The first level of interpolation is:
\begin{equation}\label{def:interpolation}
f_{2k}^{(1)} = f_k^{(0)}, \qquad f_{2k-1}^{(1)} = 0.5 \left ( f_{k-1}^{(0)} + f_k^{(0)} \right ), \qquad \forall k = 1, \dots, n_c,
\end{equation}
where we assume periodic boundary conditions: $f_0^{(0)} = f_{n_c}^{(0)}$. This step is repeated $l$ times to obtain the fine-scale data $f^{(l)}_k$, $k = 1, \dots, n_f$.
\paragraph{Restriction}
To restrict the data obtained by projection from the particle positions on the fine grid, a linear smoothing scheme is employed. Knowing the values $f^{(l)}_k$, $k = 1, \dots, n_f = 2^l n_c$ of the variable $f$ at the fine grid points, the first level of restriction gives the values
\begin{equation}\label{def:restriction}
f_{k}^{(l-1)} = 0.25 \left (f_{2k-1}^{(l)} + 2 f_{2k}^{(l)} + f_{2k+1}^{(l)}\right ), \qquad \forall k = 1, \dots, 2^{l-1} n_c,
\end{equation}
also assuming periodic boundary conditions $f_{n_f+1}^{(l)} = f_1^{(l)} $. This operation is repeated until the coarse level is reached.

\begin{remark}
Restriction of the data to a coarse grid is essential to realize a speedup through EFPI integration~\cite{Shay_2007}. 
Here, it also plays a crucial role by smoothing the statistical fluctuations of the variables inherent in the PIC scheme for the fine-scale kinetic representation of the data. 
As we will illustrate below, random fluctuations in the computed quantities appear to be substantially less pronounced when the
EFPI-accelerated schemes is used as compared to a traditional PIC schemes. 
However, it should be noted that levels of fluctuations in the net charge $n^\mathrm{i} - n^\mathrm{e}$, the electric field, or the potential remain too important to allow their use as active variables.
\end{remark}
We refer to~\cite{Shay_2007} for more details on this EFREE procedure. 
The numerical results presented in that work show good agreement with fully resolved PIC simulations for the propagation of an ion acoustic wave. In particular, the propagation speed of the wave is very well reproduced and the
numerical speedup obtained with the EFREE procedure is promising, particularly since a macroscopic CFL condition, related only to the ion sound speed $c_s$ and the coarse grid step $\Delta x_c$, appears as the limiting factor for the coarse time step $\Delta t$:
\begin{equation}\label{eq:CFL}
2 c_s \Delta t \lesssim \Delta x_c.
\end{equation}
This property enables a linear scaling as one increases the system size, i.e. the wavelength of the ion acoustic wave.

However, problems appear as the wave steepens, a shock forms and the ion distribution exhibits strong non-Maxwellian features due to ion trapping. In particular, large differences in the ion pressure appear relatively fast. As already noted in~\cite{Shay_2007}, it is necessary to generalize the description to handle non-Maxwellian distribution functions to address these phenomena. 

\subsection{A wavelet-based equation-free scheme}\label{sec:wefree}
To capture the $(x,v)$ phase-space structures emerging during the time evolution of the kinetic PIC experiments, we propose to represent the full ion distribution function $f^\mathrm{i}$, using a coarse-grained wavelet approximation in the velocity space. As in the EFREE method presented in Sec.~\ref{sec:efree}, we assume that the electrons are adiabatic and that their distribution function is well approximated by a shifted Maxwellian. However, we lift the requirements of quasi-neutrality and, most importantly, the assumption of Maxwellian-distributed ions.

\paragraph{Choice of macroscopic variables and parameterization}
The ion distribution is approximated using a two-scale phase space grid structure.
Let $V_{min}$, $V_{max}$ be an appropriately chosen bounds for the ion velocities. 
A fine phase space rectangular grid is then constructed using $n_{x,\mathrm{f}}$ points on the position axis and $n_{v,\mathrm{f}} = 2^{J_\mathrm{f}}$ points on the velocity axis. Let $\delta x = L/n_{x,\mathrm{f}}$ and $\delta v = (V_{max} - V_{min})/n_{v,\mathrm{f}}$ be the grid steps in position and velocity respectively. The ion distribution function is then approximated at this fine level as continuous and piecewise-linear in $x$ and piecewise-constant in $v$, i.e.
\begin{equation}\label{eq:waveletapproximation}
	f^\mathrm{i}(x,v) = \sum_{k_x = 1}^{n_{x,\mathrm{f}}} \sum_{k_v = 1}^{n_{v,\mathrm{f}}}  f^\mathrm{i,f}_{k_x, k_v} \ \chi_{k_x}\left ( \frac{x}{\delta x} - k_x \right ) \xi_{k_v}\left ( \frac{v - V_{min}}{\delta v} + \frac{1}{2} - k_v \right ),
\end{equation}
where
\begin{equation}
	\chi_{k}(x) =  \begin{cases}
		1 - \vert x \vert & \text{if }\vert x + j n_{x,\mathrm{f}} \vert \leq 1 \text{ for some } j \in \N;\\
		0 & \text{else,}
	\end{cases} 
	\qquad \text{and} \qquad
	\xi_{k}(v) =  \begin{cases}
		1 & \text{if }\vert v \vert \leq 1/2;\\
		0 & \text{else.}
	\end{cases} 
\end{equation}
The fine discrete representation is then composed of the $n_{x,\mathrm{f}} \times n_{v,\mathrm{f}}$ coefficients,
\begin{equation}\label{def:empiricaldistribution2}
	f^\mathrm{i,f} \equiv \left \{ f^\mathrm{i,f}_{k_x, k_v} \right \}_{\stackrel{1 \leq k_x \leq n_{x,\mathrm{f}}}{ 1 \leq k_v \leq n_{v,\mathrm{f}}}}.
\end{equation}
A coarse representation can then be obtained efficiently as follows:
\begin{itemize}
	\item The data is first restricted to the coarse spatial resolution with $n_{x,\mathrm{c}}$ points using the scheme~\eqref{def:restriction}, leading to an intermediate set of coefficients $f^{\mathrm{i}, (l)} \equiv \left \{ f^{\mathrm{i}, (l)}_{k_x, k_v} \right \}_{\stackrel{1 \leq k_x \leq n_{x,\mathrm{c}}}{ 1 \leq k_v \leq n_{v,\mathrm{f}}}}$ in the notation of Sec.~\ref{sec:efree}.
	\item For a given position index $1 \leq k_x \leq n_{x,\mathrm{c}}$, the corresponding column of values is assimilated to a set of scaling coefficients in a wavelet basis at the scale $J_\mathrm{f}$ (where $n_{v,\mathrm{f}} = 2^{J_\mathrm{f}}$). The (periodized) discrete wavelet transform is then applied to compute the scaling coefficients at a coarse scale $J_\mathrm{c} \leq J_\mathrm{f}$~\cite{Daubechies1993,Nguyen_2010}. All coefficients for scales strictly finer than $J_\mathrm{c}$ are then discarded, leading to retain only $n_{v,\mathrm{c}} = 2^{J_\mathrm{c}}$ coefficients for each position index $k_x$. We obtain thus a coarse representation,
	\begin{equation}
		f^\mathrm{i,c} \equiv \left \{ f^\mathrm{i,c}_{k_x, k_v} \right \}_{\stackrel{1 \leq k_x \leq n_{x,\mathrm{c}}}{ 1 \leq k_v \leq n_{v,\mathrm{c}}}}.
	\end{equation}
\end{itemize}

\noindent In this work, we employ the R-Coiflet wavelet family of order 4, chosen for its regularity and symmetry properties~\cite{Daubechies1993} in addition to having four vanishing moments. For completeness, some notions on wavelets and the wavelet transform are recalled in Appendix~\ref{sec:PrimerWavelets}.

The scaling coefficients $f^\mathrm{i,c}$ are chosen as the \textit{active} variables, leading to a wavelet-based equation-free projective integration method. 
\begin{remark}
Note that non-physical negative values of the distribution function can be obtained due to the extrapolation procedure, or due to the non-positivity of the oscillating wavelet functions as observed in~\cite{Nguyen_2010}. We propose and detail in Appendix~\ref{sec:QuietStart} a procedure designed to correct this defect and allow resampling from the approximated distribution function, while conserving mass and momentum.
\end{remark}

To complete the description of the system, we determine the electron \textit{passive} variables $n^\mathrm{e}$, $V^\mathrm{e}$ and $T^\mathrm{e}$ by first computing the ion density and velocity $n^\mathrm{i}$ and $V^\mathrm{i}$ by integrating numerically~\eqref{eq:waveletapproximation} along the velocity coordinate. Using the assumption of adiabatic electrons~\eqref{eq:boltzmann}, the electrostatic potential satisfies the nonlinear Poisson-Boltzmann equation:
\begin{equation}\label{eq:poisson2}
-\Delta \phi = \frac{e}{\varepsilon_0} \left (n^\mathrm{i} - n_0 \mathrm{exp}\left ( \frac{e \phi}{k_B T^\mathrm{e}_0} \right ) \right ).
\end{equation}
Equivalently, we solve self-consistently for the rescaled potential $\widetilde{\phi} = e \phi / (k_B T_0^\mathrm{e})$ such that
\begin{equation}\label{eq:poisson2rescaled}
- \lambda^2 \Delta \widetilde{\phi} + \mathrm{exp} ( \widetilde{\phi} ) = n^\mathrm{i} / n_0.
\end{equation}
In practice, this equation is discretized on the coarse grid and solved iteratively using a Newton method. Then, the values for the electron variables are recovered from Eqs.~\eqref{eq:adiabaticity}, ~\eqref{eq:boltzmann} and~\eqref{eq:electronvelocity}:
\begin{equation}\label{eq:electronpassive}
n^\mathrm{e} = n_0 \mathrm{exp} ( \widetilde{\phi} ), \qquad V^\mathrm{e} = V^\mathrm{i}, \qquad T^\mathrm{e} = T^\mathrm{e}_0.
\end{equation}

\paragraph{Restriction operator}
The process for restricting the PIC data to the coarse level is in principle quite simple: the histogram of the empirical ion distribution function~\eqref{def:empiricaldistribution2} is constructed by projecting each particle using the PIC weighting scheme in space and a nearest-grid-point scheme for the velocity.
However, to further reduce the error associated with the projective integration, it is known that one can take advantage of an appropriate choice of \textit{co-evolving frame} associated with the solution~\cite{Kavousanakis_2007}. 

We propose here to project the particles by transporting the phase space grid backwards along approximated characteristics from time $t_{n+1}$  to $t_n + p \delta t$, where $t_n = n \Delta t$. An equivalent description is that the solution is approximately advected forward along the characteristics up to time $t_{n+1}$. This original procedure is detailed in Appendix~\ref{sec:MovingFrame}.

In practice, we sample for a few fine time steps the histogram of the approximately advected solution $\widetilde{f}_{(n,n+1)}^\mathrm{i,f}$ at times $t_n + p \delta t$, $0 \leq p \leq N_e$ (see Appendix~\ref{sec:MovingFrame} for details). The coarse representation $\widetilde{f}_{(n,n+1)}^\mathrm{i,c}$ of these coefficients is obtained at each step by the coarsening process described above. Note that by construction,
\[
	\widetilde{f}_{(n,n+1)}^\mathrm{i,c}(t_{n+1}) = f^\mathrm{i,c}(t_{n+1})
\]
\begin{remark}
	Note that this is a different approach than used in~\cite{Shay_2007,Chen2011,Taitano2013}, where the distribution function is simply shifted by the wave velocity at the initialization stage. In our numerical experiments, this simple shifting did not yield satisfactory results due to the phase-space (as opposed to moments) representation, see also the discussion in Section~\ref{sec:DiscussionIonWave}.
\end{remark}
In the case of a strongly hyperbolic problem such as the Vlasov-Poisson problem, traveling with an approximate flow results in big improvements in terms of noise and accuracy since the derivative which is estimated is smaller in the new frame.
In addition, our approach allows to naturally take into account the \textit{transport component} of the equation.

\paragraph{Extrapolation}
After completion of the fine-scale steps by the PIC integrator, the coarse variables obtained by the restriction operator described above are projected forward using a forward Euler scheme as in Eq.~\eqref{eq:extrapolation}:
\begin{equation}\label{eq:extrapolation2}
	f^\mathrm{i,c}(t_{n+1}) = \widetilde{f}_{(n,n+1)}^\mathrm{i,c}(t_{n+1}) = \widetilde{f}_{(n,n+1)}^\mathrm{i,c}(t_n + N_e \delta t) + (\Delta t - N_e \delta t) \left ( \frac{\mathrm{d}}{\mathrm{d}t} \widetilde{f}_{(n,n+1)}^\mathrm{i,c}\right ).
\end{equation}
In this expression, the derivative is estimated by a least-squares estimator, which writes
\begin{equation}\label{eq:leastsquares}
	\frac{\mathrm{d}}{\mathrm{d}t} \widetilde{f}_{(n,n+1)}^\mathrm{i,c} = \sum_{p = 0}^{\lfloor N_e / 2 \rfloor} \frac{6 }{N_e (N_e + 1) (N_e + 2) } \frac{(N_e - 2i)}{\delta t} \left [ \widetilde{f}_{(n,n+1)}^\mathrm{i,c}(t_n + (N_e - p) \delta t) - \widetilde{f}_{(n,n+1)}^\mathrm{i,c}(t_n + p \delta t) \right ].
\end{equation}	
\begin{remark}\label{rem:chargeconservation}
	This projection scheme conserves the total charge. Indeed, the total charge is proportional to the sum of all coefficients of $f^\mathrm{i}$ or all the weights of the particles at any given step. This number is conserved by the PIC time-stepper and the restriction operator, and thus also equal to the sum of all coefficients of $\widetilde{f}_{(n,n+1)}^\mathrm{i,c}$ for $p = 0 \dots N_e$. The sum of all entries of $\frac{\mathrm{d}}{\mathrm{d}t} \widetilde{f}_{(n,n+1)}^\mathrm{i,c}$ is then zero by~\eqref{eq:leastsquares}, thus the total charge is conserved by the projection step~\eqref{eq:extrapolation2}.
\end{remark}

\begin{remark}
	The forward Euler scheme~\eqref{eq:extrapolation2} is only first order accurate. It is possible to use higher order timestepping for the projection step in EFPI, e.g. leap-frog or Runge-Kutta methods~\cite{Kevrekidis_2003,Gear_2003}.
	However in numerical experiments for the wavelet-based EFPI scheme, implementing a higher order integrator did not yield any improvements in accuracy or stability.
\end{remark}

\paragraph{Lifting and initialization of the next macro-step}
The lifting step consists of three important stages: 
\begin{itemize}
\item Increasing the phase space resolution of the data using the linear interpolation~\eqref{def:interpolation} and the inverse DWT in~\eqref{eq:waveletapproximation}; 
\item Loading the particles using the quiet start procedure.
\end{itemize}
Note that the loading scheme is particularly important as it controls the initial level of the noise. We detail this procedure in Appendix~\ref{sec:QuietStart}.

\section{Test cases}\label{sec:numresults}
We now illustrate the performance of the proposed approach for solving the Vlasov-Poisson system by modeling
two different test cases, and seek to compare the results obtained with the classical PIC scheme and with the proposed coarse time-stepping algorithm. 
Simulations with the original EFPI algorithm proposed by Shay {\it et. al.}~\cite{Shay_2007} are also considered in the first test case, comprising an ion acoustic wave, for a comparison. The second test case is a Gaussian plasma expansion in a vacuum.

\subsection{Ion acoustic wave test-case}

The first test case is the propagation and nonlinear steepening of an ion acoustic wave in a Maxwellian two-species plasma. In previous studies~\cite{Shay_2007, Maluckov_2008}, other versions of EFPI-accelerated PIC-based codes were applied to this test case and we refer in particular to~\cite{Shay_2007} for a detailed account of the theory and initial setup of the wave. 
We initialize the Vlasov-Poisson system with
\begin{equation}\label{eq:errormeasures}
f^\mathrm{i,e}(x,v, t=0) = C^\mathrm{i,e} \left (1 + \delta cos \left ( 2\pi x/L\right ) \right )\mathrm{exp}\left (- \frac{1}{2} \left ( \frac{ v - \delta cos \left ( 2\pi k \cdot x/L \right ) }{v_\mathrm{th}^{i,e}}\right )^2\right ),
\end{equation}
where $C^{i,e} $ are renormalization constants, $\delta = 0.2$ is the initial amplitude of the perturbation. The initial thermal velocities are chosen as $v_\mathrm{th}^\mathrm{i} = 0.22$ and $v_\mathrm{th}^\mathrm{e} = 42.5$, since the condition $T_0^\mathrm{e} \gg T_0^\mathrm{i}$ is necessary to minimize Landau damping of the wave.

The set of parameters normalized for the XES1-based PIC code~\cite{Birdsall_Langdon_1985} are $\epsilon_0 = 1$, $\omega_p^\mathrm{e} = 1527.35$, $\omega_p^\mathrm{i} = 36$, $q_\mathrm{e}/m_\mathrm{e} = -1800$, $q_\mathrm{i}/m_\mathrm{i} = 1$. The time step is chosen as $ \delta t = 0.0001667 \approx 0.25 / \omega_p^\mathrm{e} $. The number of mesh cells $nx = n_f$ is adjusted depending on the system length $L$ for a constant grid step $h = 1/128 \approx 0.28 \lambda$. These parameters correspond to the same case as in~\cite{Shay_2007}. Under these conditions, both the fast space and time scales (the Debye length and the electron plasma period) are resolved, and the PIC method is stable.

Parameters for the simulation runs used for this study are shown in Table~\ref{tab:iawparameters}.
\begin{table}[ht]
\centering
\begin{tabular}{|c|c|c|c|c|c|c|c|c|}
 \hline
Run & Type & $L$ & $N_\mathrm{ppc}$ & $n_{x,\mathrm{f}}$ & $n_{x,\mathrm{c}}$  & $n_{v,\mathrm{f}}$ & $n_{v,\mathrm{c}}$ & $\Delta t / \delta t$ \\ 
\hline
\hline
1 & PIC & $4$ & $8192$ & $512$ & $-$ & $-$ & $-$ & $-$\\
\hline 
2 & EFPI & $4$ & $4096$ & $512$ & $512$ & $2048$ & $64$ & $150$ \\
\hline
3 & EFPI & $4$ & $2048$ & $512$ & $128$ & $2048$ & $64$ & $250$ \\
\hline
4 & EFPI & $4$ & $1024$ & $512$ & $32$ & $2048$ & $64$ & $600$ \\
\hline
\hline
5 & PIC & $8$ & $8192$ & $1024$ & $-$ & $-$ & $-$ & $-$\\
\hline 
6 & EFPI & $8$ & $4096$ & $1024$ & $1024$ & $2048$ & $64$ & $150$ \\
\hline
7 & EFPI & $8$ & $2048$ & $1024$ & $256$ & $2048$ & $64$ & $250$ \\
\hline
8 & EFPI & $8$ & $1024$ & $1024$ & $64$ & $2048$ & $64$ & $600$ \\
\hline
\end{tabular}
\caption{Simulation runs presented for this test case. $L$ is the domain length, $N_\mathrm{ppc}$ the number of particles per species per cell for the underlying PIC code, $n_{x,\mathrm{f}}$ and $n_{x,\mathrm{c}}$ ($n_{v,\mathrm{f}}$ and $n_{v,\mathrm{c}}$) respectively the number of fine and coarse space (velocity) grid points, $\delta t$ and $\Delta t$ respectively the fine and coarse time steps.} \label{tab:iawparameters}
\end{table}

\subsubsection{Choice of EFPI numerical parameters}
\begin{figure}[ht]
\centering
\begin{minipage}{.49\textwidth}
\centering
\includegraphics[width=\textwidth]{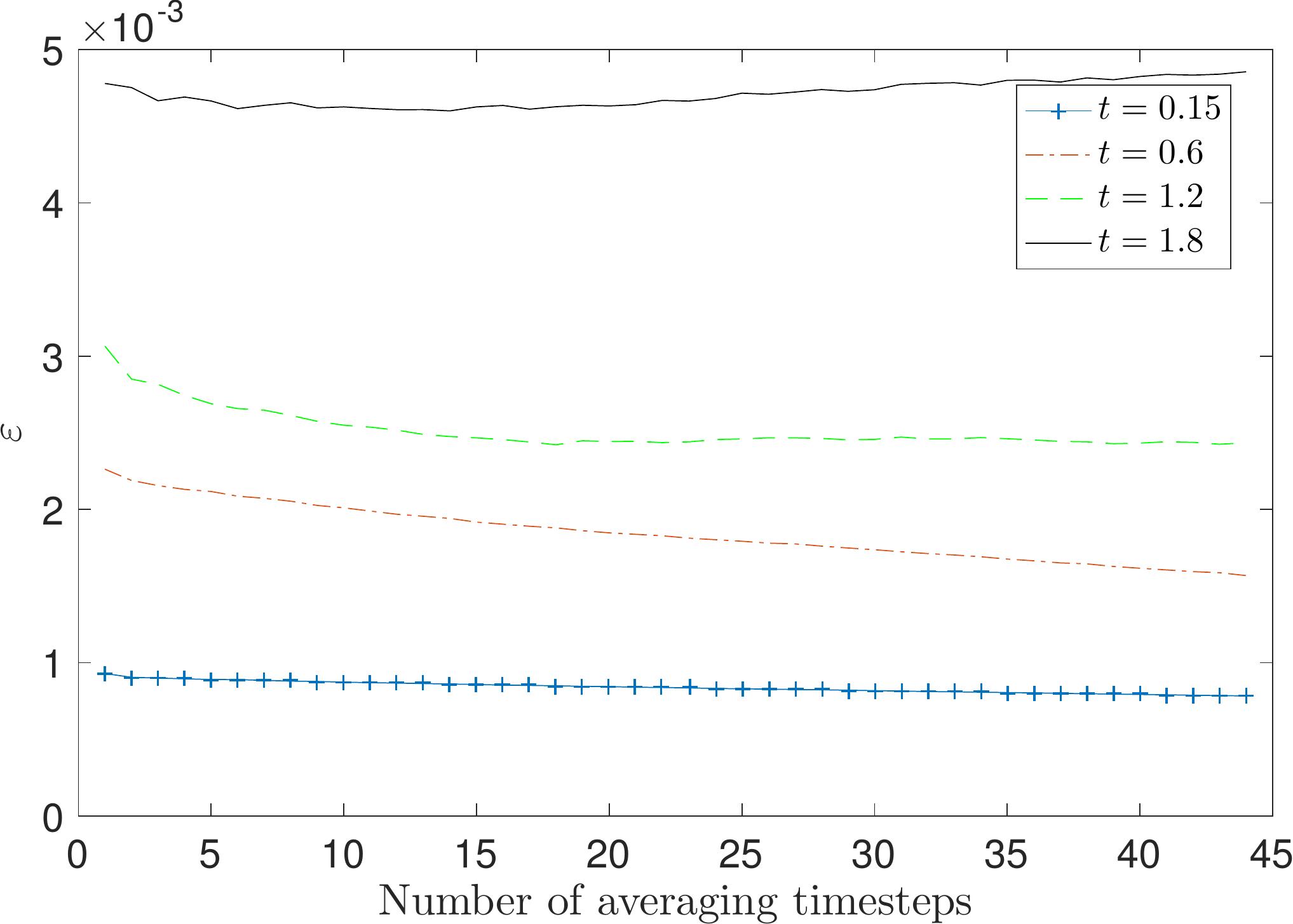}
\subcaption{Error in the density, $\varepsilon$.} \label{fig:optimtimestep:a}
\end{minipage}
\begin{minipage}{.49\textwidth}
\includegraphics[width=\textwidth]{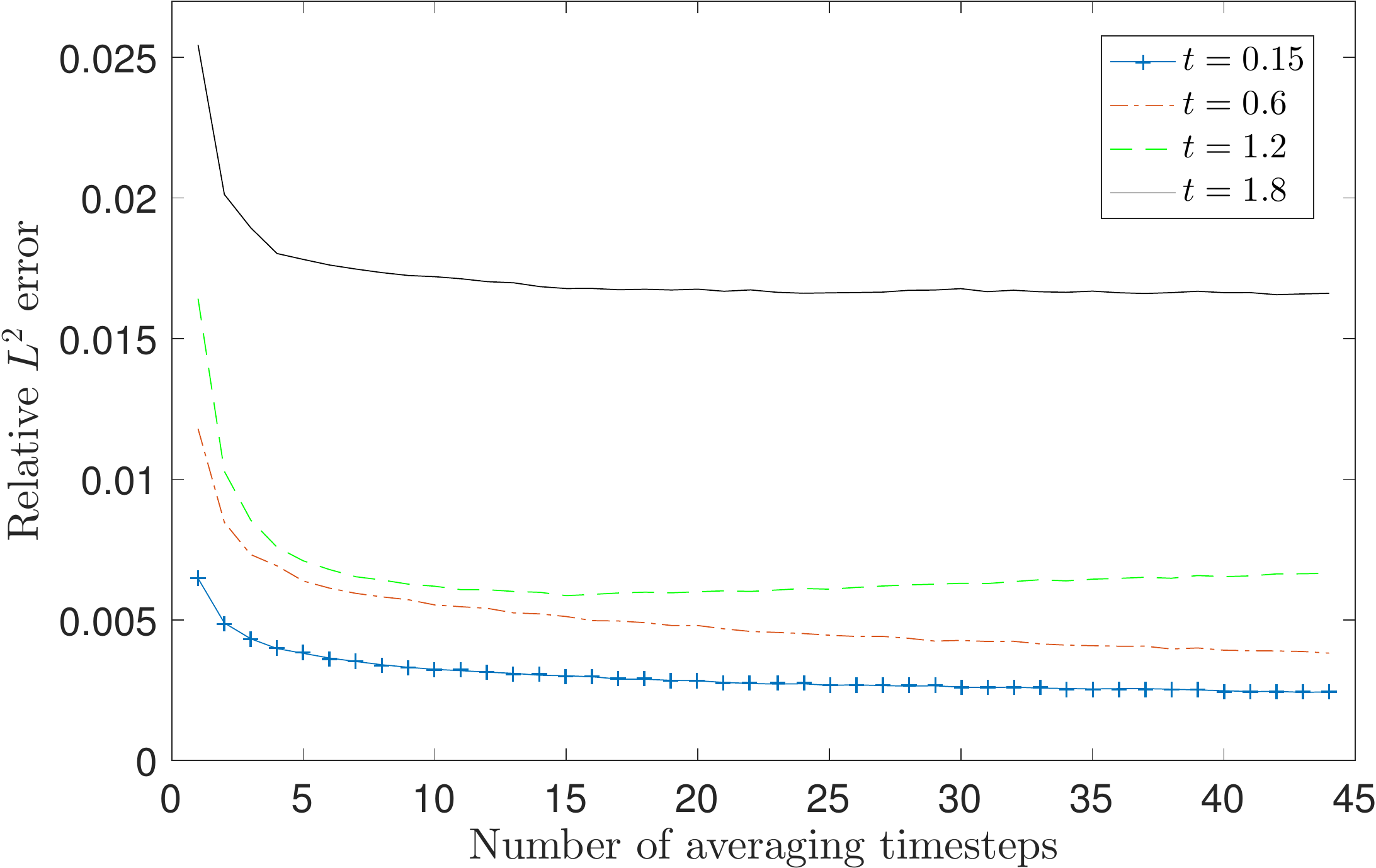}
\subcaption{Error in the discretized distribution function.}\label{fig:optimtimestep:b}
\end{minipage}
\caption{Extrapolation error as a function of the number of fine-scale time steps for each macro step projection cycle, Run $2$.}\label{fig:optimtimestep}
\end{figure}
The EFPI algorithm presents two main parameters, both of which require some tuning to obtain optimal results: the coarse time step $\Delta t$ and the number of fine-scale time steps needed during each projective integration cycle $N_e$. 
As noted in~\cite{Shay_2007}, in the first case $\Delta t$ is controlled by a Courant-type condition which is well understood. 

However, the choice of $N_e$ is more delicate. 
In~\cite{Shay_2007}, various choices ranging from $N_e = 10$ to $N_e = 20$ are analyzed numerically showing a strong effect on the error, with the simulation diverging quickly for $N_e \leq 14$ and almost indistinguishable results for $ N_e \geq 16$. 
%In addition, the authors observed that the minimum value of $N_e$ for which the EFPI run remains stable does not depend on the size of the system, $L$, nor on the number of particles.
We investigate here the influcence of $N_e$ on both the relative $L^2$ error in the ion density function, noted $\varepsilon$ (see also Eq. (33) in~\cite{Shay_2007}), and the relative $L^2$ error in the coefficients of the discretized distribution function over the whole phase space, which are defined as:
\begin{equation}
	\varepsilon = \sqrt{\frac{\sum_{j = 1}^{n_{x,\mathrm{c}}}(n_j^{EFPI} - n_j^{PIC})^2}{\sum_{j = 1}^{n_{x,\mathrm{c}}}(n_j^{EFPI})^2 + (n_j^{PIC})^2}}, \qquad \left \Vert f^{EFPI} - f^{PIC} \right \Vert_{L^2} = \sqrt{\frac{\sum_{k_x, k_v}(f^{\mathrm{i,c},EFPI}_{k_x, k_v} - f^{\mathrm{i,c},PIC}_{k_x, k_v})^2}{\sum_{k_x,k_v} (f^{\mathrm{i,c},EFPI}_{k_x, k_v})^2 + (f^{\mathrm{i,c},PIC}_{k_x, k_v})^2}}.
\end{equation}
Figure~\ref{fig:optimtimestep} shows both error functions for Run $2$ for various times up to $t = 1.8$ as a function of $N_e$.
We observe that the method is stable for all choices of $N_e \geq 1$. Furthermore, the error stagnates for $N_e \geq 6$, a value which was stable with respect to the number of particles used. This value of $N_e = 6$ is used in all results presented in this section (Runs $1$--$8$). 

\subsubsection{Numerical results}

Figures~\ref{fig:comparison1},~\ref{fig:comparison2} and~\ref{fig:comparison3} show results for runs $1$--$4$ (case $L = 4$) at times $t=1.5$, $3.5$ and $t=4.5$, respectively. We show the unnormalized ion distribution function for the PIC scheme (run~$1$) and our proposed wavelet-based projective method (run~$2$--$4$). The highly resolved PIC simulation is taken as the reference. We can see that as the wave propagates, it deforms and steepens as a shock develops. 

\begin{figure}[b]
\centering
\begin{minipage}{.49\textwidth}
	\centering
	\includegraphics[width=.8\textwidth]{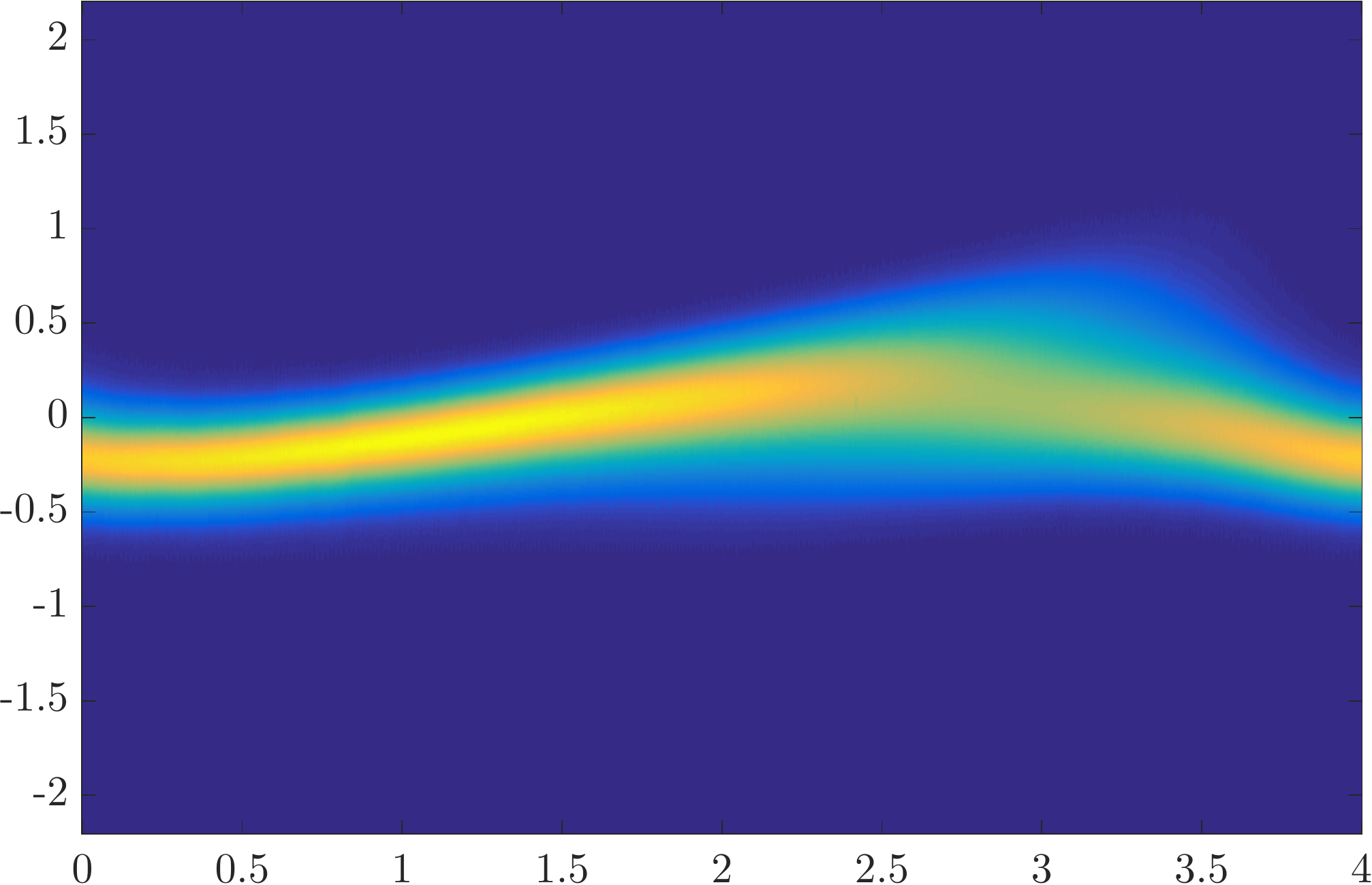}
	\subcaption{Fully resolved PIC solution (Run $1$)}\label{fig:comparison1:a}
\end{minipage}
\begin{minipage}{.49\textwidth}
	\centering
	\includegraphics[width=.8\textwidth]{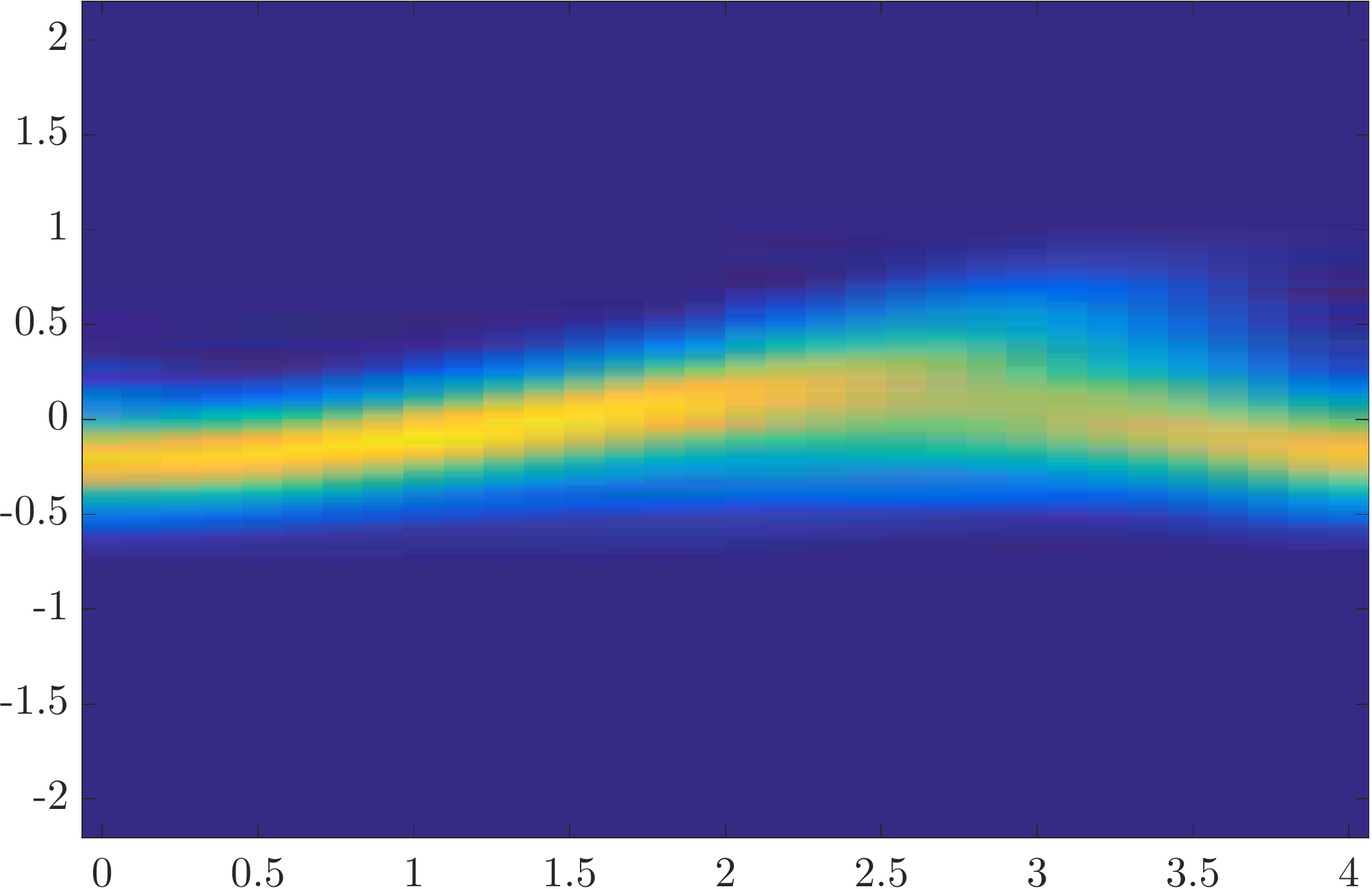}
	\subcaption{EFPI solution (Run $4$, $n_c = 32$)}\label{fig:comparison1:b}
\end{minipage}
\caption{Ion distribution function plots, Runs $1$ and $4$, $t = 1.5$.}\label{fig:comparison1}
\end{figure}
\begin{figure}[!t]
\centering
\begin{minipage}{.49\textwidth}
	\centering
	\includegraphics[width=.75\textwidth]{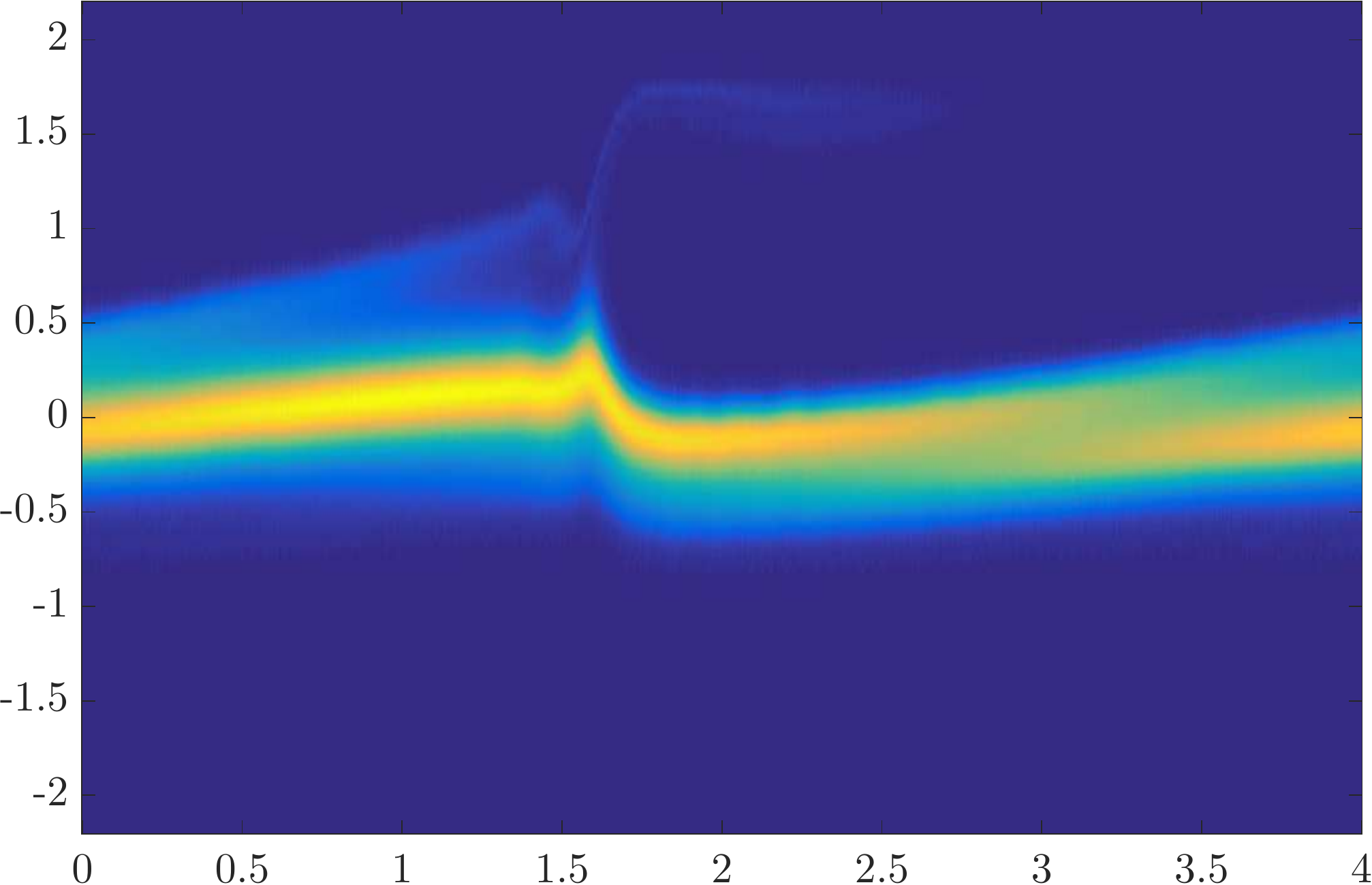}
	\subcaption{Fully resolved PIC solution (Run $1$)}\label{fig:comparison2:a}
\end{minipage}
\begin{minipage}{.49\textwidth}
	\centering
	\includegraphics[width=.75\textwidth]{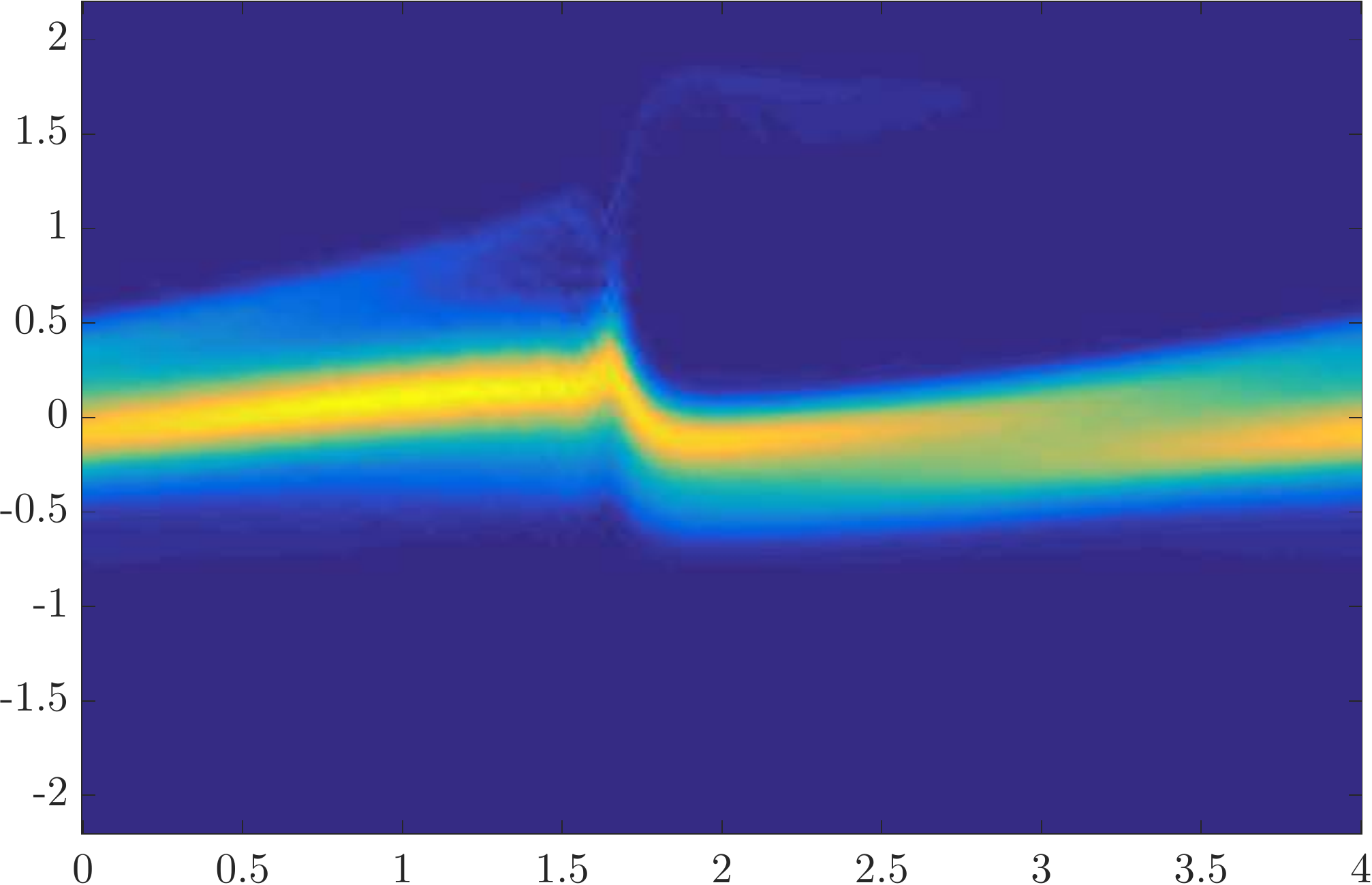}
	\subcaption{EFPI solution (Run $2$, $n_c = 512$)}\label{fig:comparison2:b}
\end{minipage}
\\
\begin{minipage}{.49\textwidth}
	\centering
	\includegraphics[width=.75\textwidth]{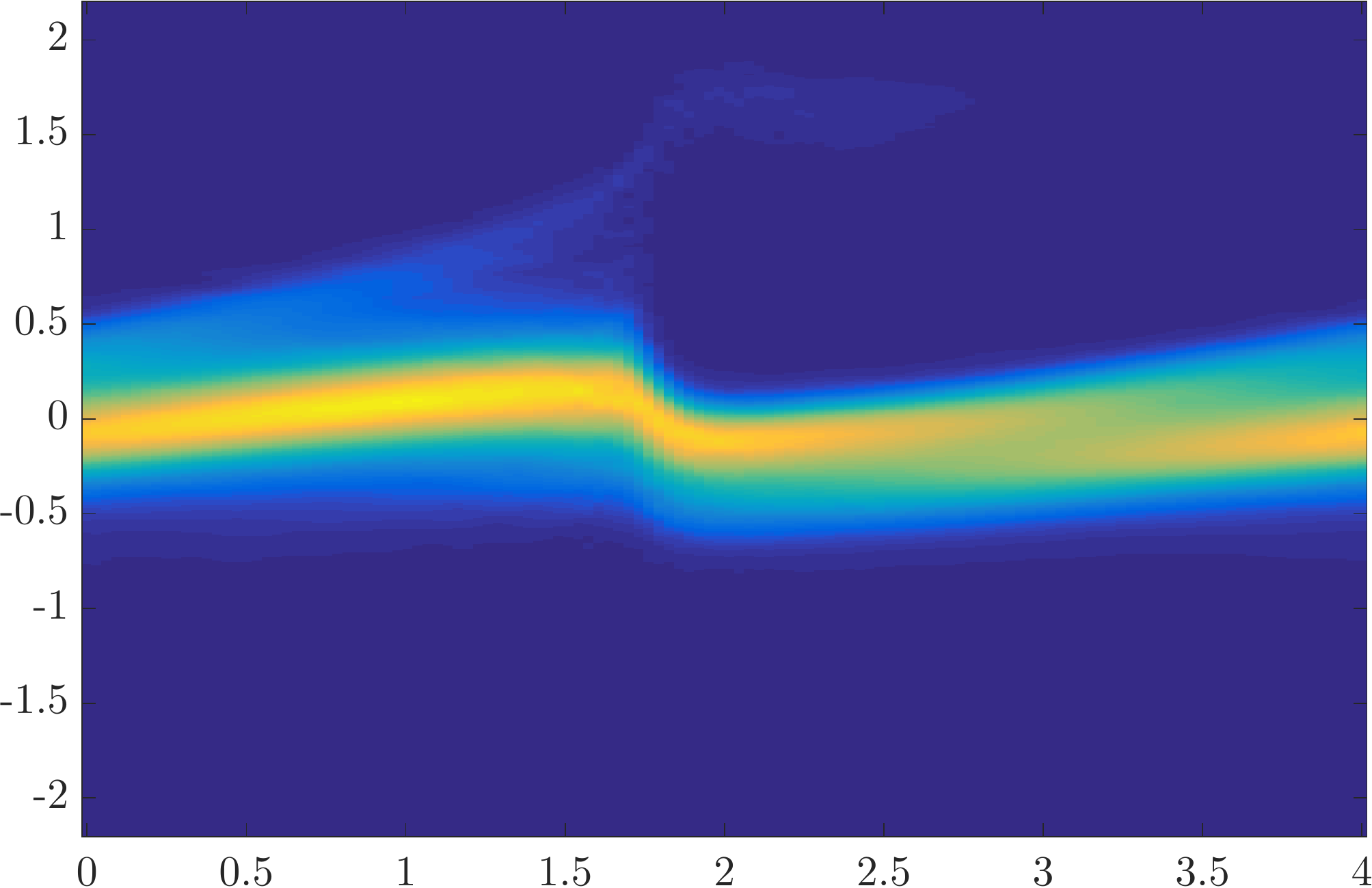}
	\subcaption{EFPI solution (Run $3$, $n_c = 128$)}\label{fig:comparison2:c}
\end{minipage}
\begin{minipage}{.49\textwidth}
	\centering
	\includegraphics[width=.75\textwidth]{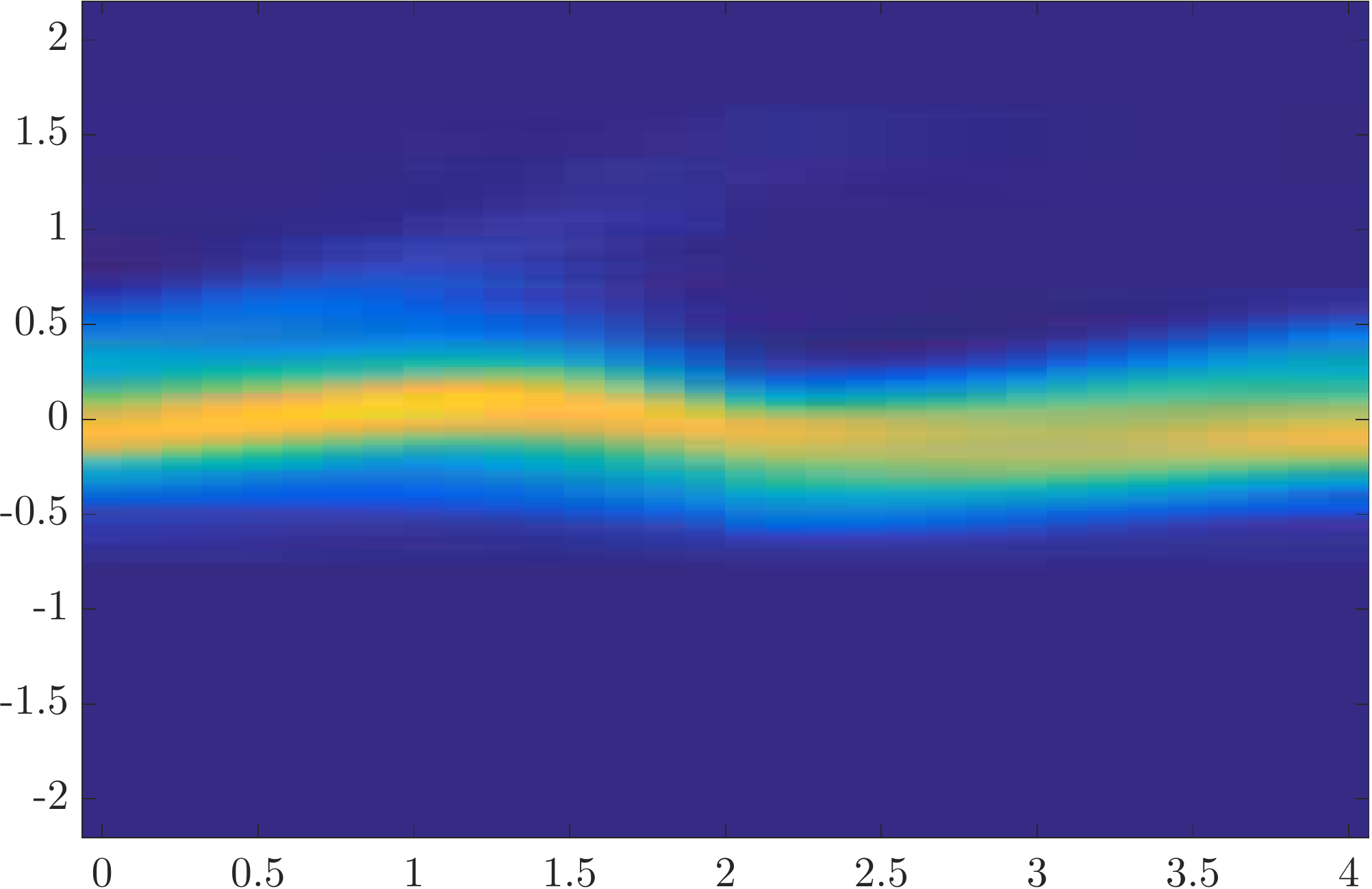}
	\subcaption{EFPI solution (Run $4$, $n_c = 32$)}\label{fig:comparison2:d}
\end{minipage}
\caption{Ion distribution function plots, Runs $1$--$4$, $t = 3.5$.}\label{fig:comparison2}
\end{figure}
At time $t=1.5$, the ion distribution function is very smooth. 
The different schemes are in good agreement with each other and all main features are well represented. 
The propagation speed of the ion wave is very well matched by the proposed method. 
Note that the upper velocity tail starts to deviate from the Maxwellian distribution. 
This result shows that the (linear) thresholding of the wavelet coefficients is very efficient and already very well adapted to use in the EFPI algorithm.

At time $t=3.5$, a ion shock has fully developed. As can be seen on Figure~\ref{fig:comparison2}, the accuracy of the wavelet-based EFPI solution is very dependent on the level of coarsening in space. Phase-space features of the shock are very well reproduced when $512$ grid points are used for the projective integration (no spatial coarsening), but are completely smoothed out when the solution is coarsened using $32$ points. 
 \begin{figure}[!b]
\centering
\begin{minipage}{.49\textwidth}
	\centering
	\includegraphics[width=.8\textwidth]{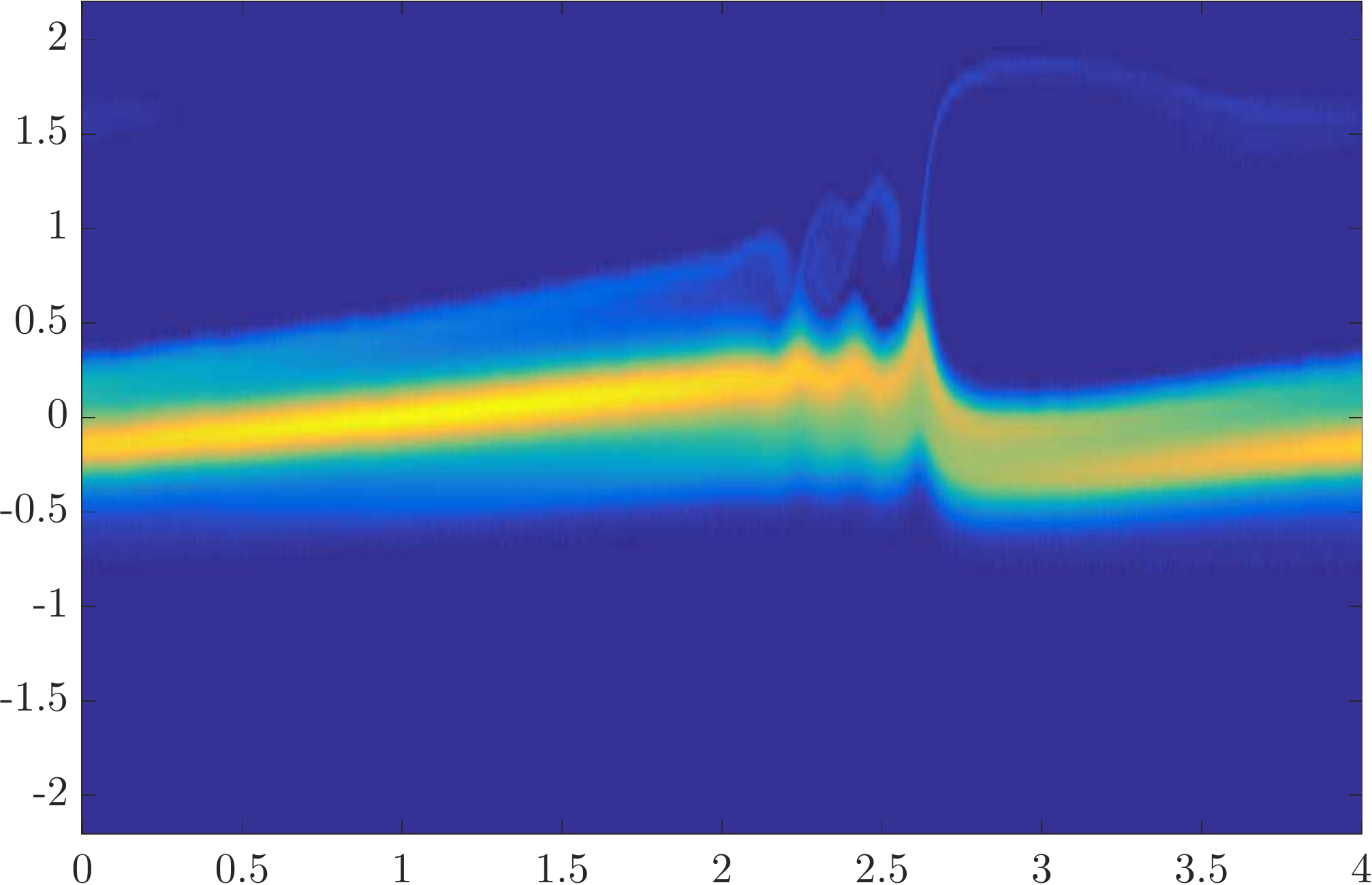}
	\subcaption{Fully resolved PIC solution (Run $1$)}\label{fig:comparison3:a}
\end{minipage}
\begin{minipage}{.49\textwidth}
	\centering
	\includegraphics[width=.8\textwidth]{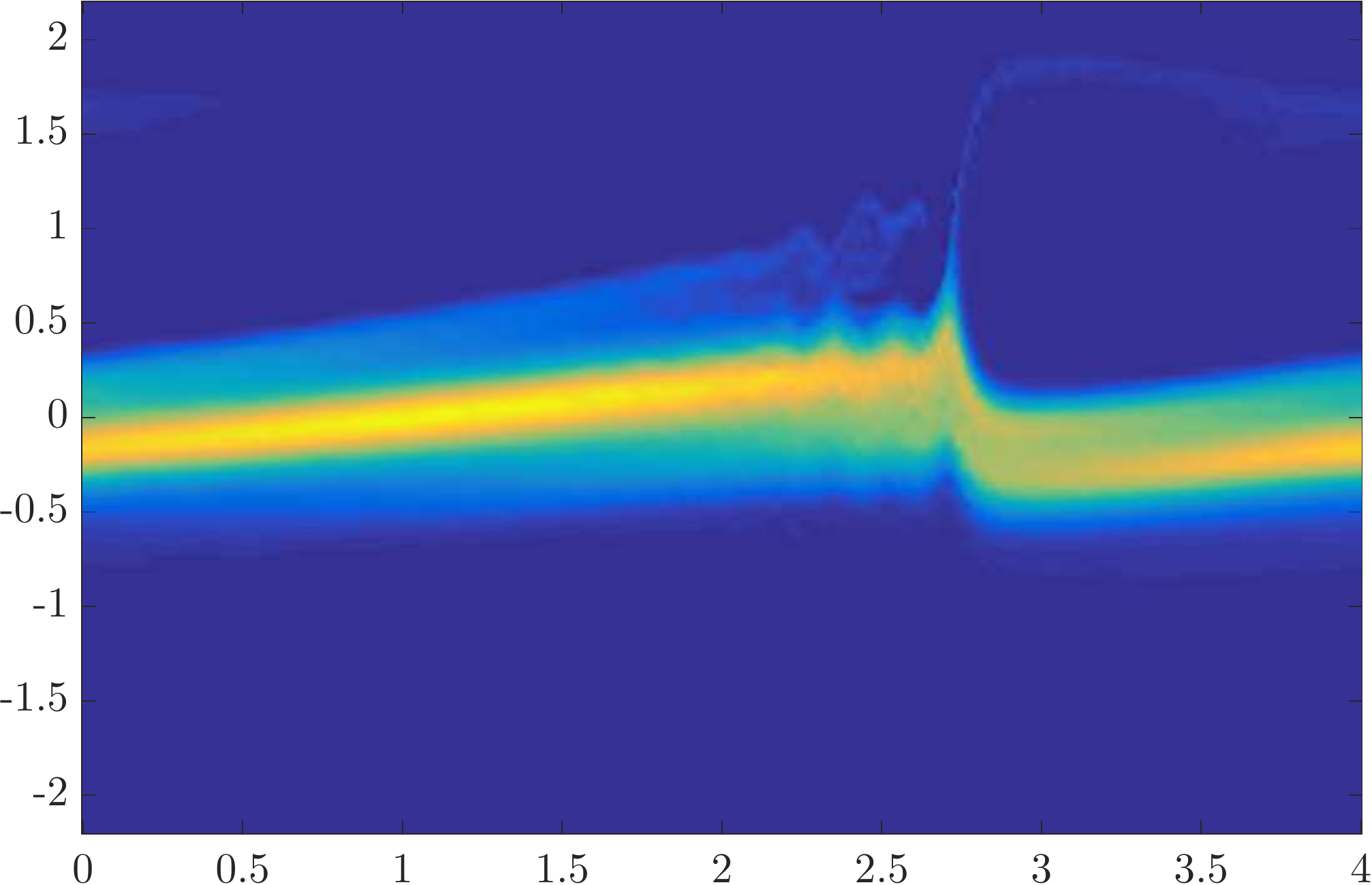}
	\subcaption{EFPI solution (Run $2$, $n_c = 512$)}\label{fig:comparison3:b}
\end{minipage}
\caption{Ion distribution function plots, Runs $1$--$2$, $t = 4.5$.}\label{fig:comparison3}
\end{figure}

 At the final time $t = 4.5$, the reflected ions as well as smaller trailing shocks form distinctive features, which are again well reproduced in the finer EFPI solution.
\begin{remark}
 	This level of diffusion was not present in the EFREE study~\cite{Shay_2007}, where knowledge of the ion sound speed is used to track the ion sound wave in a uniformly co-moving frame. However this transformation did not yield more accurate results with the present method which tracks the full phase-space.
 \end{remark}

\begin{figure}[!t]
\centering
\begin{minipage}{.49\textwidth}
	\centering
	\includegraphics[width=.8\textwidth]{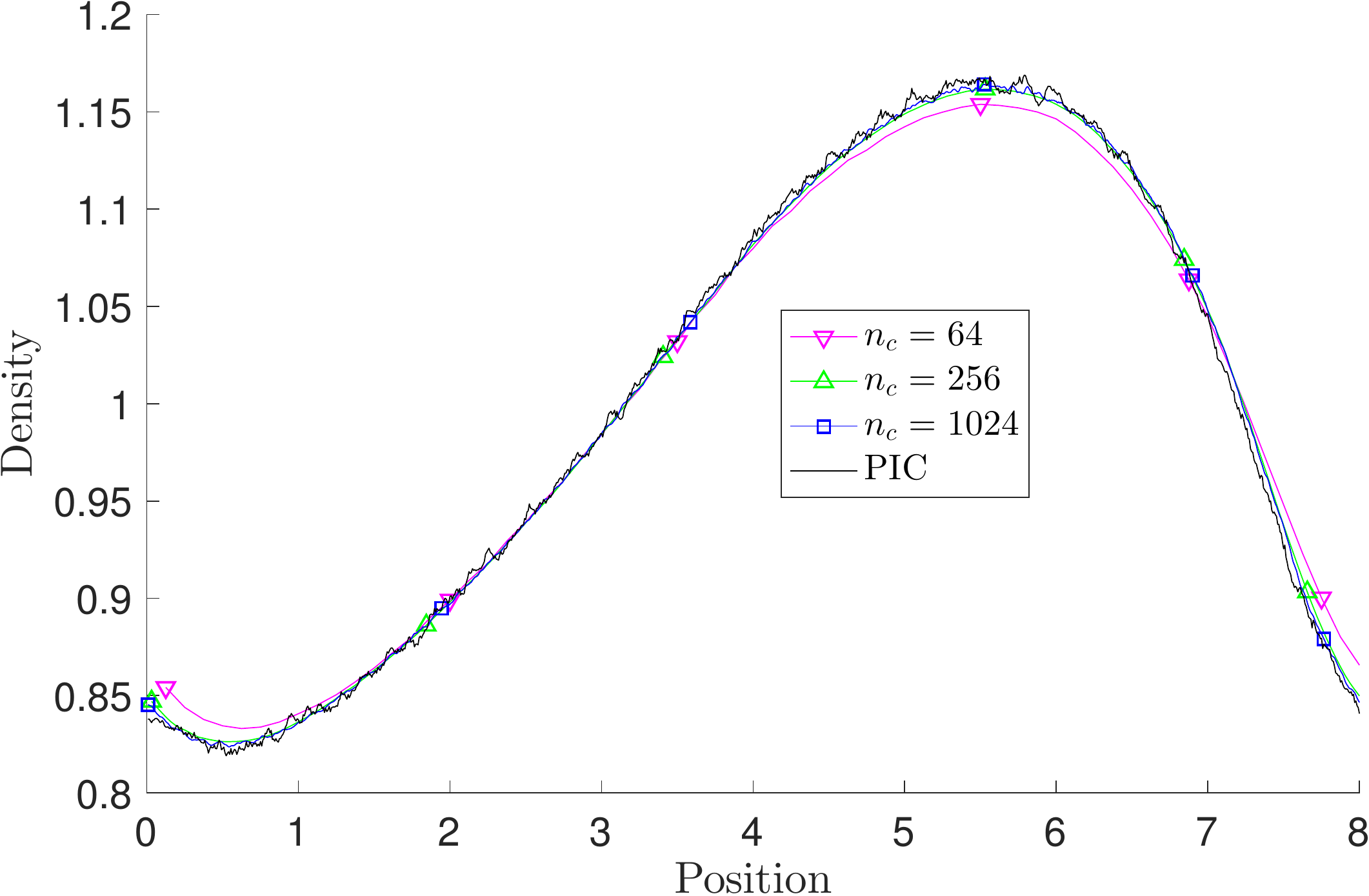}
	\subcaption{Ion density versus $x$}\label{fig:comparison4:a}
\end{minipage}
\begin{minipage}{.49\textwidth}
	\centering
	\includegraphics[width=.8\textwidth]{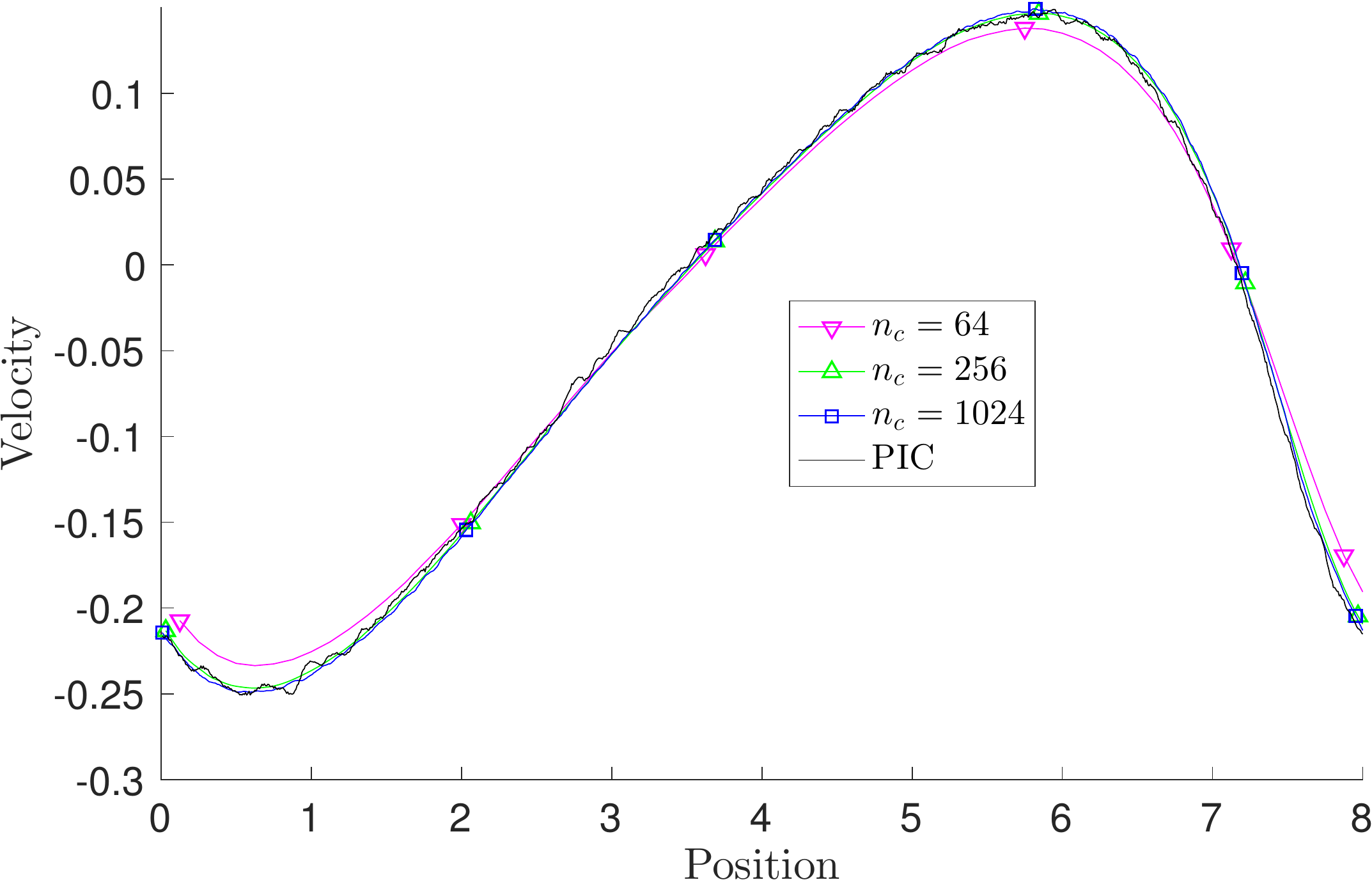}
	\subcaption{Ion velocity versus $x$}\label{fig:comparison4:b}
\end{minipage}
\\
\begin{minipage}{.49\textwidth}
	\centering
	\includegraphics[width=.8\textwidth]{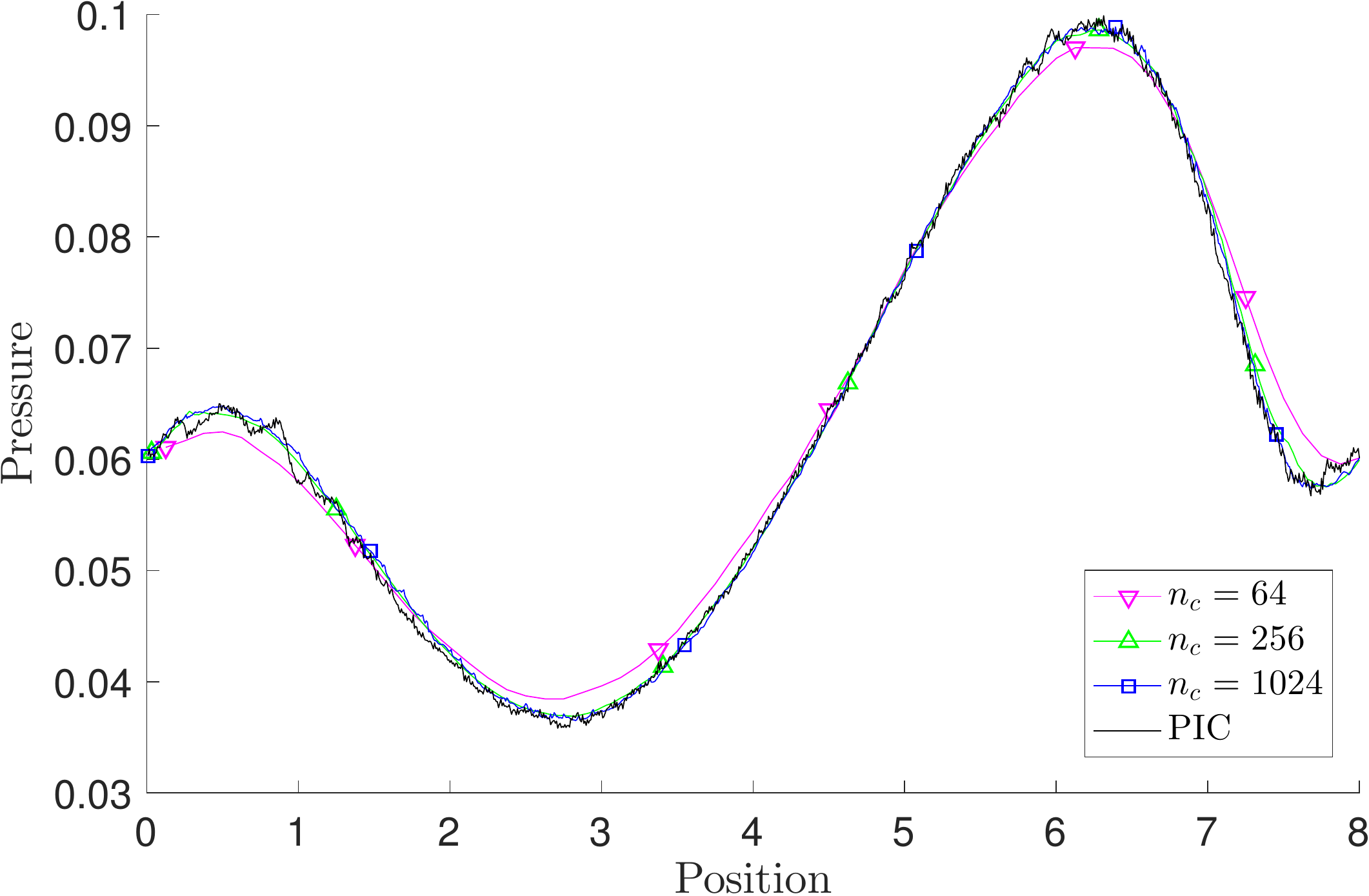}
	\subcaption{Ion pressure versus $x$}\label{fig:comparison4:c}
\end{minipage}
\begin{minipage}{.49\textwidth}
	\centering
	\includegraphics[width=.8\textwidth]{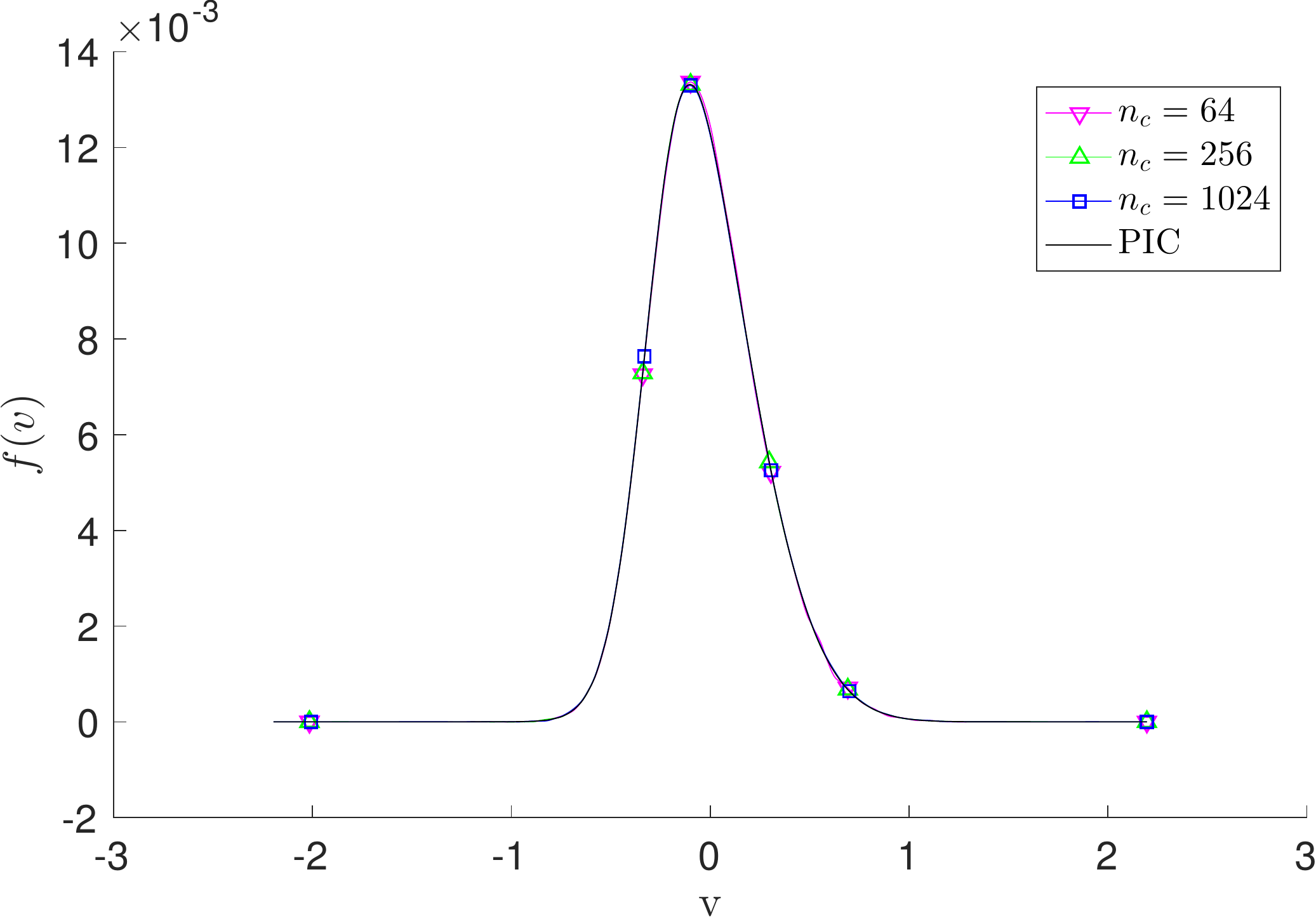}
	\subcaption{Integrated ion velocity distribution}\label{fig:comparison4:d}
\end{minipage}
\caption{Comparison between moments obtained with PIC (Run $5$) and EFPI (Runs $6$--$8$) at $t = 3$.}\label{fig:comparison4}
\end{figure}

Next, we compare the moments and the velocity distribution obtained for the case of a larger wave number ($L = 8$) at times $t=3$ and $t=7$.
At time $t = 3$, Figure~\ref{fig:comparison3}, there is reasonable agreement between all the schemes, although the coarsest solution with $n_\mathrm{x,c} = 64$ shows some diffusion. In particular, the anharmonic shape for the ion pressure and the asymmetric velocity distribution are well reproduced.

At time $t = 7$, Figure~\ref{fig:comparison4}, the accuracy of the EFPI solution is highly dependent on the resolution which is chosen. In particular, the narrow density and velocity peaks at the shock are well reproduced by the fine EFPI solution ($1024$ grid points), but absent from the coarser solutions. The overall shape is still well reproduced using only $256$ grid points, but clearly overdamped by numerical diffusion when using only $64$ points. The ion velocity distribution shows the same behavior, with the second peak corresponding to ions reflected by the shock well reproduced for the finer EFPI solutions.

\begin{figure}[!t]
\begin{minipage}{.49\textwidth}
	\centering
	\includegraphics[width=.8\textwidth]{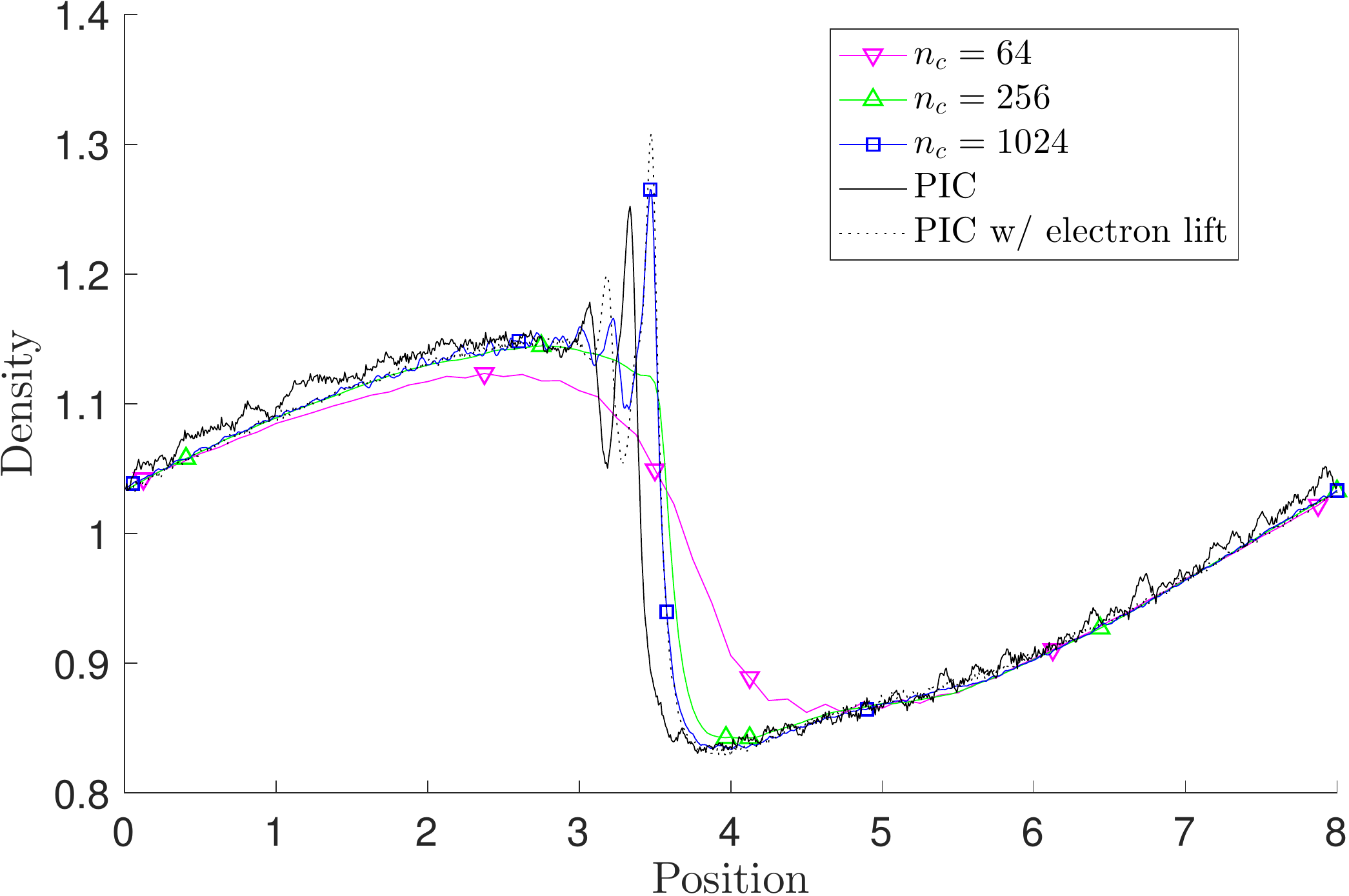}
	\subcaption{Ion density versus $x$}\label{fig:comparison5:a}
\end{minipage}
\begin{minipage}{.49\textwidth}
	\centering
	\includegraphics[width=.8\textwidth]{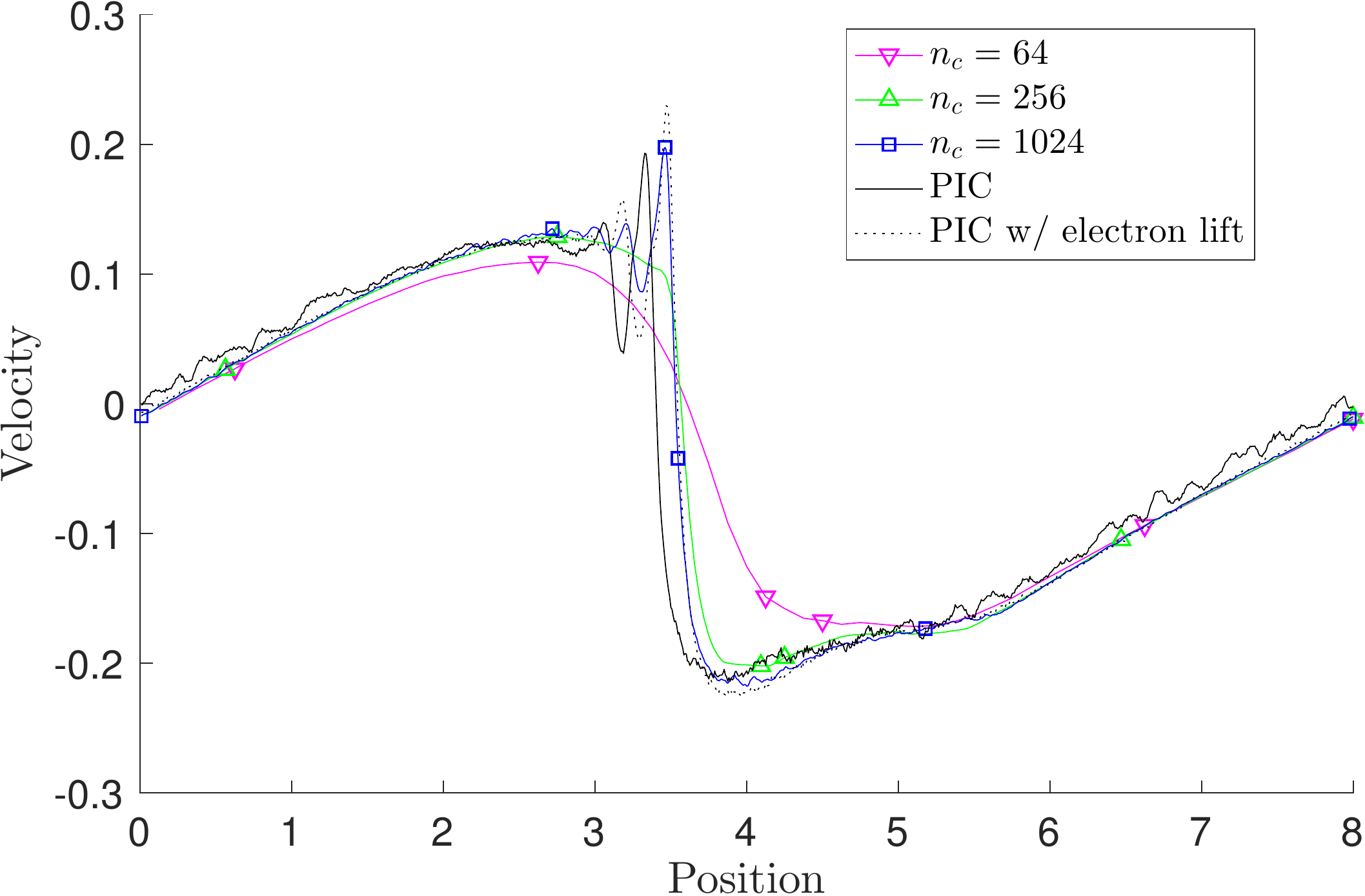}
	\subcaption{Ion velocity versus $x$}\label{fig:comparison5:b}
\end{minipage}
\\
\begin{minipage}{.49\textwidth}
	\centering
	\includegraphics[width=.8\textwidth]{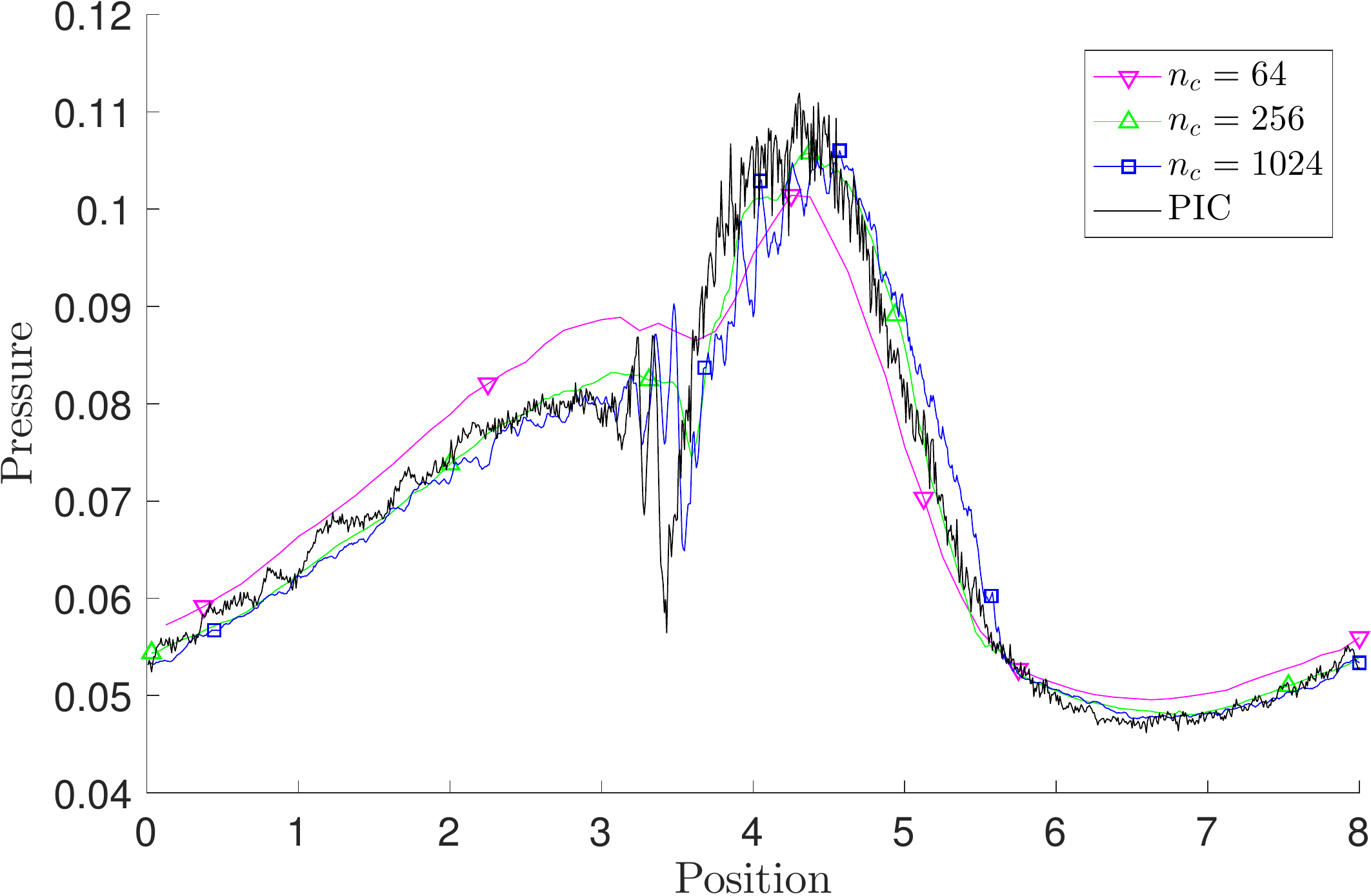}
	\subcaption{Ion pressure versus $x$}\label{fig:comparison5:c}
\end{minipage}
\begin{minipage}{.49\textwidth}
	\centering
	\includegraphics[width=.8\textwidth]{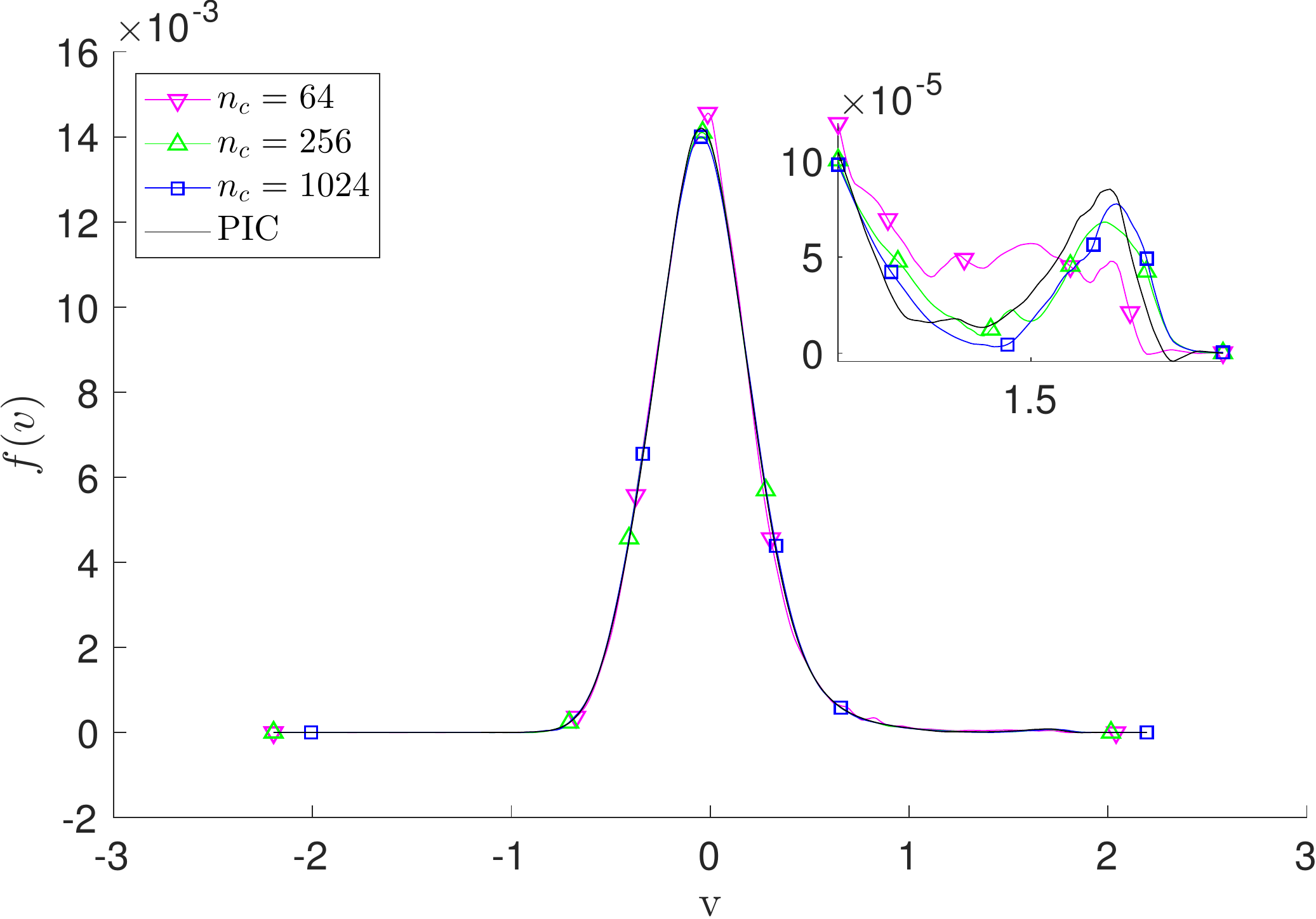}
	\subcaption{Integrated ion velocity distribution, with an inset showing the reflected ion tail bump}\label{fig:comparison5:d}
\end{minipage}
\caption{Comparison between moments obtained with PIC (Run $5$) and EFPI (Runs $6$--$8$) at $t = 7$.}\label{fig:comparison5}
\end{figure}

\subsubsection{Discussion} \label{sec:DiscussionIonWave}
Strongly non-Maxwellian effects appear progressively during the ion acoustic wave propagation. 
In particular one observes reflected particles and a double-peaked ion distribution function at the shock. 
The simple system representation by shifted Maxwellians in the original EFREE implementation~\cite{Shay_2007} lead to noticeable differences with the reference simulation. 
In particular, the ion pressure diverges quickly and does not develop the correct strongly anharmonic shape. We have illustrated here how the wavelet-based equation free method substantially improves on these results. 
Our projective integration scheme gives accurate results even in the strongly nonlinear regime, after the creation of the shock, while using a large macroscopic time step. This shows that the EFPI framework can be effective in a kinetic setting, provided a non-parametric representation of the general non-maxwellian ion distribution function is used.

% \begin{figure}[t]
% 	\centering
% \begin{minipage}{.49\textwidth}
% 	\centering
% 	\includegraphics[width=.8\textwidth]{Wave_Reference_distribution_L8_t=7.pdf}
% 	\subcaption{Fully resolved PIC solution (Run $5$)}\label{fig:comparison6:a}
% \end{minipage}
% \begin{minipage}{.49\textwidth}
% 	\centering
% 	\includegraphics[width=.8\textwidth]{Wave_EFPI_x1024_distribution_L8_t=7.pdf}
% 	\subcaption{EFPI solution (Run $6$, $n_c = 1024$)}\label{fig:comparison6:b}
% \end{minipage}
% 	\caption{Comparison between distribution functions obtained with PIC (Run $5$) and EFPI (Runs $6$--$8$) results at $t = 7$.}
% 	\label{fig:comparison6}
% \end{figure}

\begin{figure}[t]
	\centering
	\includegraphics[width=.5\textwidth]{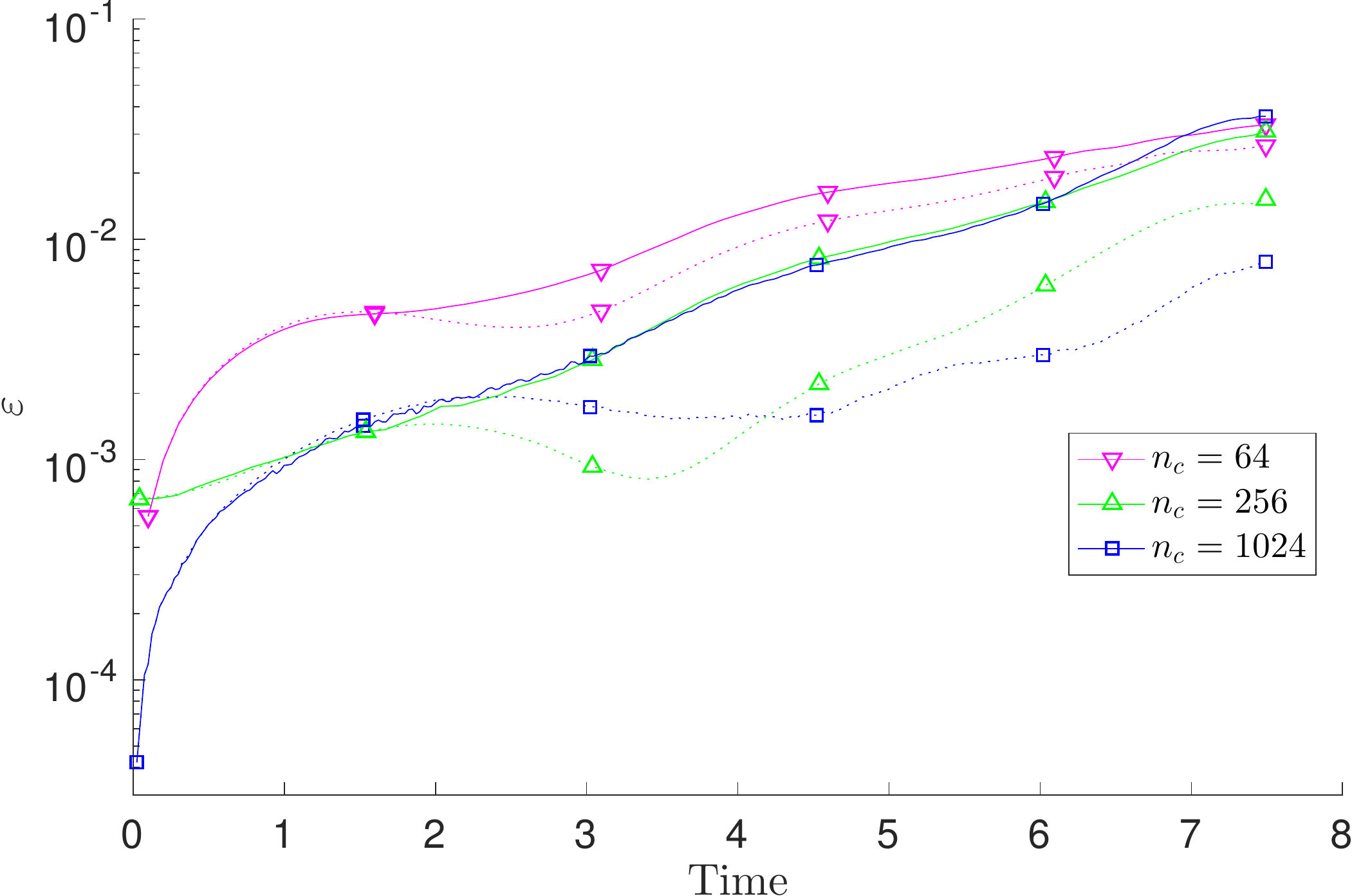}
	\caption{Error $\varepsilon$ in the ion density compared to the fully resolved PIC solutions without (solid lines) and with (dotted lines) periodic reset of the electrons to the self-consistent adiabatic Maxwell-Boltzmann equilibrium.}\label{fig:error1}
	\vspace{10pt}
	\includegraphics[width=.5\textwidth]{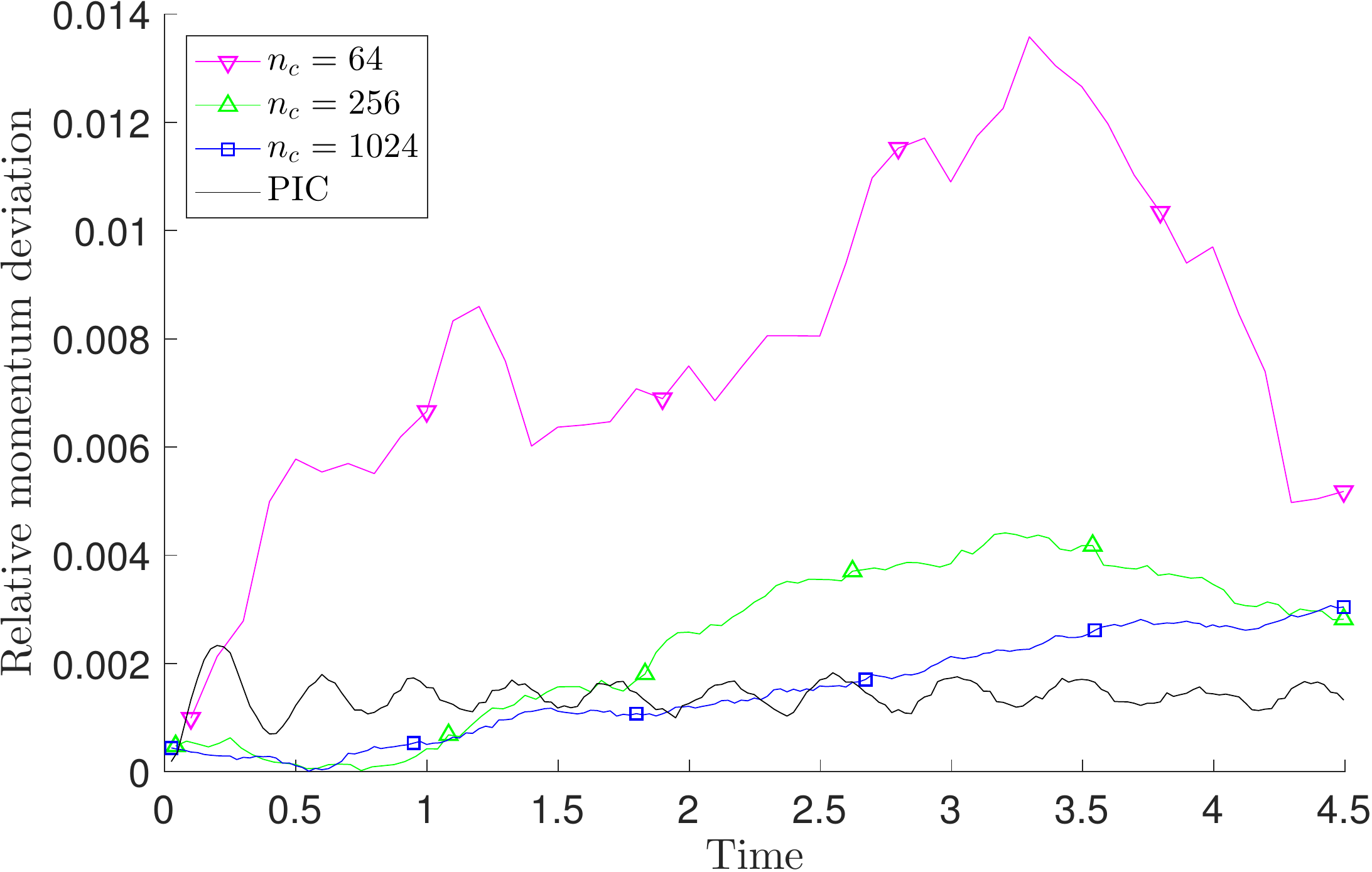}
	\caption{Relative total ion momentum deviation (Runs $5$--$8$)}\label{fig:error2}
\end{figure}
\paragraph{Convergence analysis and modeling error}
We present on Figure~\ref{fig:error1} the density error $\varepsilon$ as a function of time for Runs $6$--$8$, as defined in~\eqref{eq:errormeasures}. Solid lines indicate the error relative to the reference PIC solution (Run $5$). The error increases slowly with time and is quite well controlled, after the initial period where noise dominates. Note that increasing the resolution from $n_{x,\mathrm{c}} = 256$ to $1024$ grid points has little effect on the error apart from the very beginning, but the highly underresolved solution with $n_{x,\mathrm{c}} = 64$ is as expected much less accurate overall.

An important component of the error is that the wave speed obtained with EFPI is slightly too high. Comparing the density peaks at the shock front for the case where $L=8$, we observe a position difference of approximately $+0.12$ with respect to the reference PIC solution at time $t = 7$ (see Fig.~\ref{fig:comparison5:a}). It is interesting to note that this difference can be completely explained by the assumption of Boltzmann-distributed Maxwellian electrons which is implicit in the EFPI method. Indeed, resonant electrons are trapped by the shock potential, inducing a slight deviation from the Maxwellian distribution~\cite{Schamel1973}. This effect is neglected by the EFPI representation.
 To show this, we perform a full PIC simulation with the same parameters as Run~$5$, where in addition the electrons particles are periodically reset to the self-consistent Boltzmann equilibrium. More precisely, every $150$ PIC timesteps, the electron phase space distribution is recalculated from the ion density using equations~\eqref{eq:poisson2rescaled} and~\eqref{eq:electronpassive}. The result is plotted in dotted lines on Figures~\ref{fig:comparison5:a} and~\ref{fig:comparison5:b}, and shows an very good agreement with the finest EFPI solution. 

 In addition, we plot on Figure~\ref{fig:error1} the density error with respect to this modified PIC solution, in dotted lines. We see that the error is much improved, especially at later times where Debye-scale structures appear at the shock front. At the final time $t = 8$, with a fully developed shock, the EFPI solution converges to this new modified PIC solution as the macro grid step and the macro time step are reduced. 

\paragraph{Stability and the need for a particle deposition scheme with integration along approximate characteristics}

\begin{table}[t]
\centering
\begin{tabular}{|c|c|c|}
 	\hline
	 \multirow{2}{*}{Number of grid points}	& 	\multicolumn{2}{|c|}{Maximum stable timestep for EFPI} 			\\
	\cline{2-3}
	 					 	& Integration along characteristics &		No integration along characteristics   	\\
	\hline \hline
			$512$			& 		$0.028$ ($170 \delta t$)	&			$0.005$   ($30 \delta t$) 			\\
	\hline
			$128$			& 		$0.058$ ($350 \delta t$)	&			$0.012$   ($75 \delta t$) 			\\
	\hline
			$32$			& 		$0.104$ ($625 \delta t$)	&			$0.050$ ($300 \delta t$) 			\\
	\hline
\end{tabular}
\caption{Maximum projective time steps ensuring stability over the time interval $0 \leq t \leq 5$ for the case $L = 4$.}
\label{tab:stability}
\end{table}

Maximum stable time steps, listed in Table~\ref{tab:stability}, show that the method can be used to directly exploit the large difference between electron and ion timescales, without the need for coarsening.
When a larger macro grid step can be chosen, this can be exploited to use even bigger time steps.

An interesting test is to project $f^\mathrm{i,c}$ directly instead of the integration along the characteristics~\eqref{def:explicitflow}--\eqref{eq:extrapolation2} employed here, using a straightforward deposition of the particles on the fine grid. As seen in Table~\ref{tab:stability}, this approach requires the use of a much smaller time step, especially at higher resolution where the electric field becomes very strong at the shock front.
\begin{remark}
	Due to the use of a phase space grid, the stability constraint is influenced by the local electric field which grows nonlinearly as the shock front develops.
	Since numerical diffusion tends to smooth the resulting electric field peak, depending on the coarsening level, the resulting stability condition is not as easy to interpret as in the EFREE case. 
	Indeed, the use of a spatial-only grid results on a Courant condition based solely on the macroscopic ion sound wave speed~\cite{Shay_2007}.
\end{remark}
% Note that numerically, the method is in general insensitive for its stability with respect to $N_e$, the number of micro time steps per projection cycle, or to $N_p$, the number of particles, showing that the particle noise does not drive the stability constraints in the range of parameters we consider here. 

\paragraph{Particle noise}
Numerical experiments show that random-like fluctuations due to the PIC codes necessitate a careful treatment. 
As we have shown, the level of fluctuations in the active (ion) variables has been successfully reduced by our proposed wavelet-based approach due to (a) linear smoothing from fine to coarse grid, (b) denoising by linear or nonlinear thresholding of wavelet coefficients, (c) projection along approximate characteristics, (d) quiet start. 
This last factor appears particularly efficient since the fine-scale integration is only realized for a small number of timesteps, enabling the PIC error to scale as $\mathrm{log}(N_p)/N_p$ instead of $1/\sqrt{N_p}$~\cite{Wollman_1996}.

\paragraph{Conservation of moments}
As noted in Remark~\ref{rem:chargeconservation}, charge is conserved by the EFPI approach. This is not however the case for the momentum or energy. Indeed, the implicit representation of electrons makes this a difficult task since ion momentum alone is not a conserved quantity. We plot in Figure~\ref{fig:error2} the relative deviation of the total ion momentum with respect to its initial value. The reference curve oscillates around $0.02$, as momentum is transfered to trapped electrons. This effect cannot be reproduced by the EFPI solutions which show larger deviations, but stay within reasonable bounds.

\subsection{One-dimensional plasma expansion test-case}
As a second test case, we consider a one-dimensional expansion problem, the expansion of a Gaussian plasma in vacuum~\cite{Dorozhkina1998,Baitin1998,Kovalev2002,Kovalev2003,Mora2005}. Note that this is similar to the case of a thin uniform slab described in~\cite{Grismayer_2006, Grismayer_2008}, used as a test case e.g. in~\cite{Degond_2010}. The ions initially uniformly occupy a Gaussian slab of characteristic radius $r_0 \gg L$, while the electrons are initialized as the adiabatic Maxwell-Boltzmann equilibrium in a self-consistent potential. The test problem requires the observation of the expansion of the ion slab, initialized with:
\begin{equation}
f^{i}(x,v, t=0) = C^\mathrm{i} \mathrm{exp}\left ( - \frac{1}{2} \left ( \frac{v}{v_\mathrm{th}^\mathrm{i}}  \right )^2  - \left (\frac{x}{r_0} \right )^2 \right ),
\end{equation}
where $C^\mathrm{i}$ is a renormalization constant, and the electron distribution function is deduced by solving the nonlinear Poisson-Boltzmann equation~\eqref{eq:poisson2rescaled} and hypothesis~\eqref{eq:electronpassive}. Note that this is the same procedure used when lifting the coarse ion data in the approach developed in this paper. In the quasineutral limit, an analytic self-similar solution exists~\cite{Dorozhkina1998,Kovalev2002,Kovalev2003} and the electron distribution function keeps a Maxwellian form with a time-dependent homogeneous temperature, an important fact for the applicability of the method presented in this paper. 
\begin{table}[t]
\centering
\begin{tabular}{|c|c|c|c|c|c|c|c|c|c|}
 \hline
Run & Type & $L$ & $r_0$ & $N_\mathrm{ppc}$ & $n_{x,\mathrm{f}}$ & $n_{x,\mathrm{c}}$ & $n_{v,\mathrm{f}}$ & $n_{v,\mathrm{c}}$ & $\Delta t / \delta t$ \\ 
\hline
\hline
9 & PIC & $1000$ & $40$ & $16384$ & $4096$ & $-$ & $-$ & $-$ & $-$\\
\hline
10 & EFPI & $1000$ & $40$ & $4096$ & $4096$ & $4096$ & $2048$ & $512$ & $250$\\
\hline
11 & EFPI & $1000$ & $40$ & $4096$ & $4096$ & $2048$ & $2048$ & $512$ & $400$\\
\hline
\end{tabular}

\caption{Simulation runs presented for this test case. $L$ is the domain length, $r_0$ the initial slab width, $N_\mathrm{ppc}$ the number of particles per species per cell, $n_{x,\mathrm{f}}$ and $n_{x,\mathrm{c}}$ ($n_{v,\mathrm{f}}$ and $n_{v,\mathrm{c}}$) respectively the number of fine and coarse space (velocity) grid points, $\delta t$ and $\Delta t$ respectively the fine and coarse time steps.} \label{tab:expparameters}
\end{table}

\begin{figure}[!b]
\centering
\begin{minipage}{.49\textwidth}
\centering
\includegraphics[width=\textwidth]{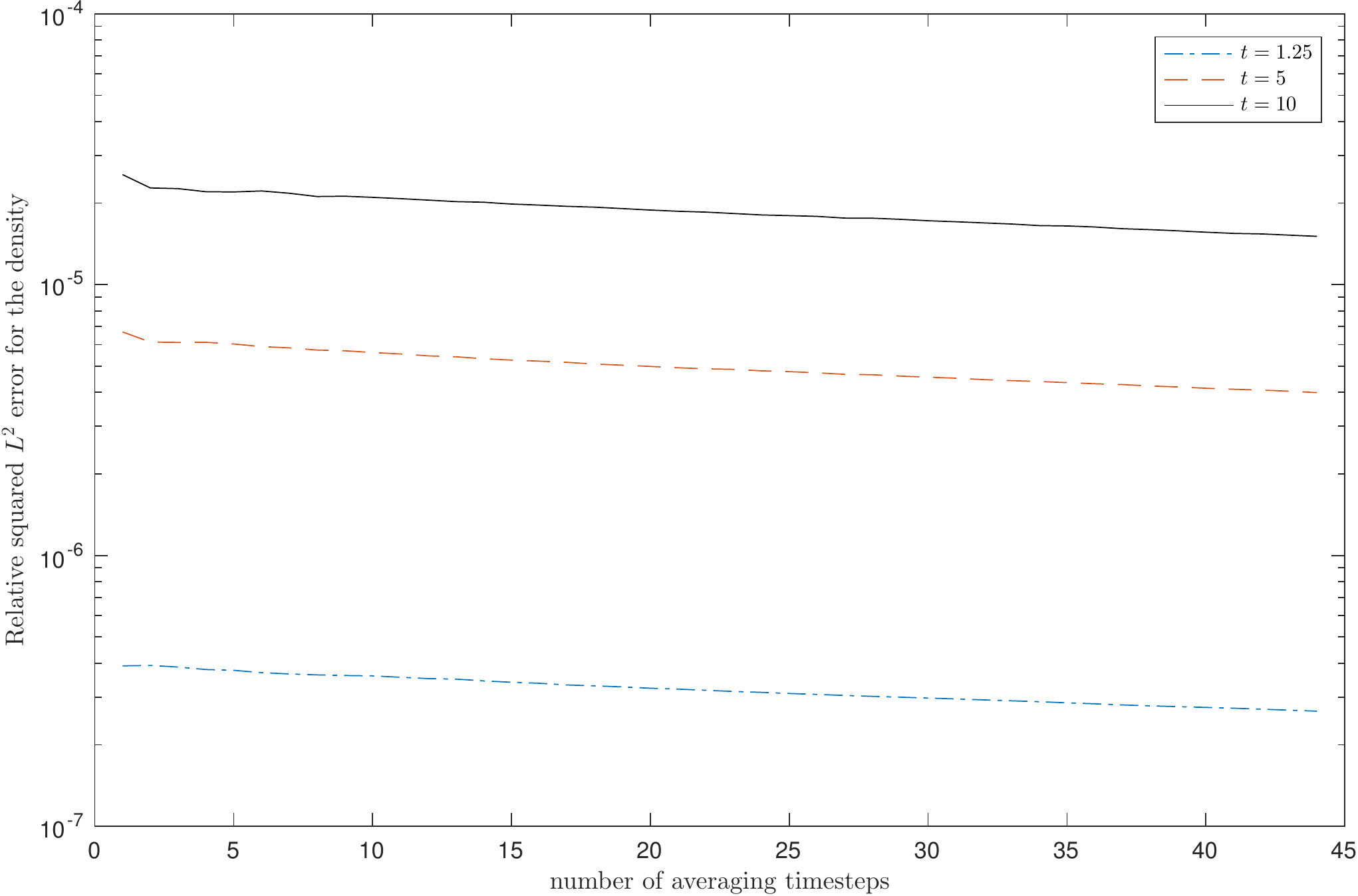}
\subcaption{Error in the density, $\varepsilon$.} \label{fig:optimtimestep2:a}
\end{minipage}
\begin{minipage}{.49\textwidth}
\includegraphics[width=\textwidth]{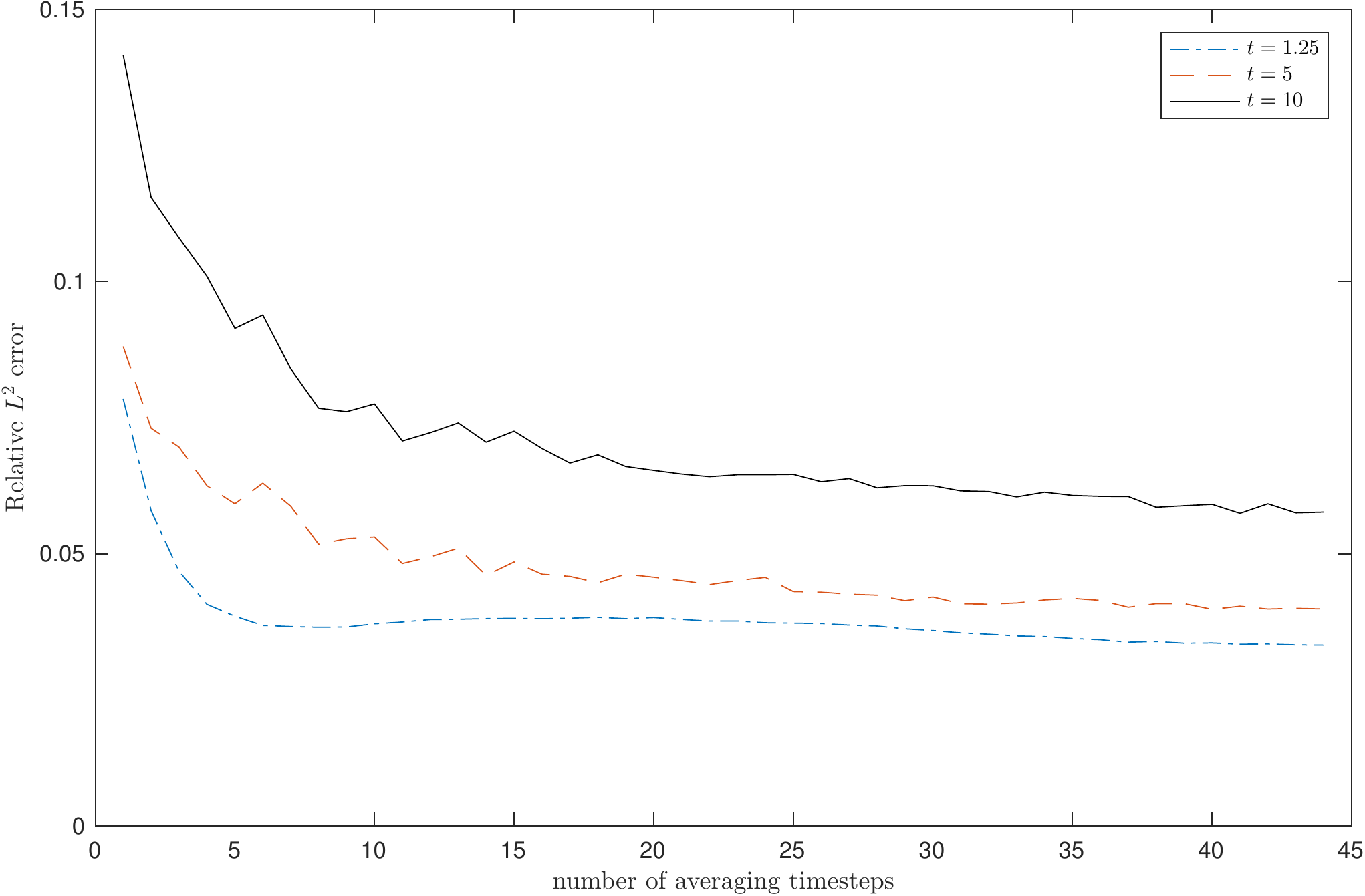}
\subcaption{Error in the discretized distribution function.}\label{fig:optimtimestep2:b}
\end{minipage}
\caption{Extrapolation error as a function of the number of fine-scale timesteps per projective step. Parameters are otherwise chosen as for run $7$ in Table~\ref{tab:expparameters}.}\label{fig:optimtimestep2}
\end{figure}

The initial electron temperature is 1000 times higher than the initial ion temperature: the thermal velocities are set as $v_\mathrm{th}^{i} = 0.0316$, $v_\mathrm{th}^\mathrm{e} = 42.84857$. 
The set of numerical parameters normalized for the PIC code are, in code units, $\epsilon_0 = 1$, $\omega_p^\mathrm{e} = 42.8486$, $q_\mathrm{e}/m_\mathrm{e} = -1.0$, $\omega_p^\mathrm{i} = 1$, $q_\mathrm{i}/m_\mathrm{i} = 0.000554466$. The system length is $L = 1000$, and the initial value for the characteristic radius is $r_0 = 40$. Simulation parameters are summarized in Table~\ref{tab:expparameters}. As in the previous test case, the electron motion is much faster than the ion motion and we can exploit this separation of scales with EFPI, since the expansion occurs roughly at the speed of the ion acoustic wave, $r(t) \approx \sqrt{r_0^2 + 2 c_\mathrm{s0}^2 t^2}$~\cite{Mora2005}.

The Debye length is initially $\lambda = v_\mathrm{th}^\mathrm{e}/\omega_p^\mathrm{e} = 1$ at the center of the plasma slab. The PIC time step is chosen as $ \delta t = 0.005 \approx 0.21 / \omega_p^\mathrm{e}$ and the grid step $h = L / 4096 \approx 0.244 \lambda$, allowing to resolve the fast space and time scales.

\subsubsection{Differences in the setup}
The boundary conditions in this example are different from the previous example: we simulate only a half-domain with a purely absorbing right boundary condition and an axis of symmetry on the left of the domain, as for example in~\cite{Degond_2010}.
In addition, the global kinetic energy of the electrons is now projectively integrated as an active variable to account for the progressive decrease in temperature as the expansion expands and freezes, as well as the total charge of each species to account for the possible escape of particles through the right boundary.

\subsubsection{Choice of EFPI numerical parameters and numerical results}

As in the ion acoustic wave test case, numerical experiments show that the coarse time step $\Delta t$ is controlled by a Courant-type condition. To choose the number $N_e$ of fine-scale time steps computed for each projective integration cycle, we repeat the analysis from the previous test case by measuring empirically the relative $L^2$ error in the ion density and discretized distribution functions depending on $N_e$, using the parameters from run~$10$. 
Results are shown in Fig.~\ref{fig:optimtimestep2}. The method is stable for all values of $N_e$, and the choice of $N_e = 10$ seems to give good results, and is used in all simulations presented in this section.

\begin{figure}[p]
\centering
\begin{minipage}{.49\textwidth}
	\centering
	\includegraphics[width=.8\textwidth]{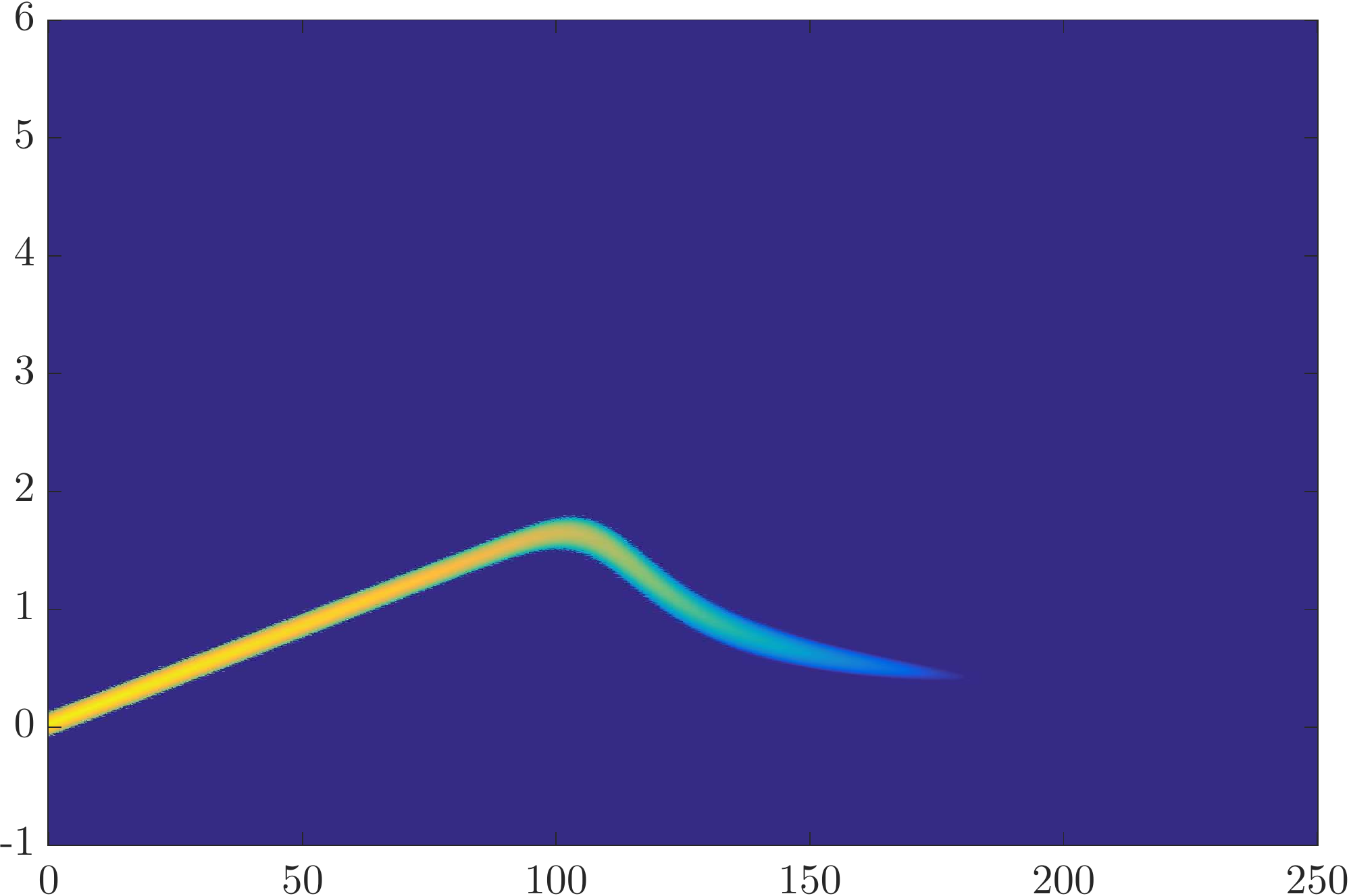}
	\subcaption{Fully resolved PIC solution (Run $9$), logarithmic color scale}\label{fig:expansion1:a}
\end{minipage}
\begin{minipage}{.49\textwidth}
	\centering
	\includegraphics[width=.8\textwidth]{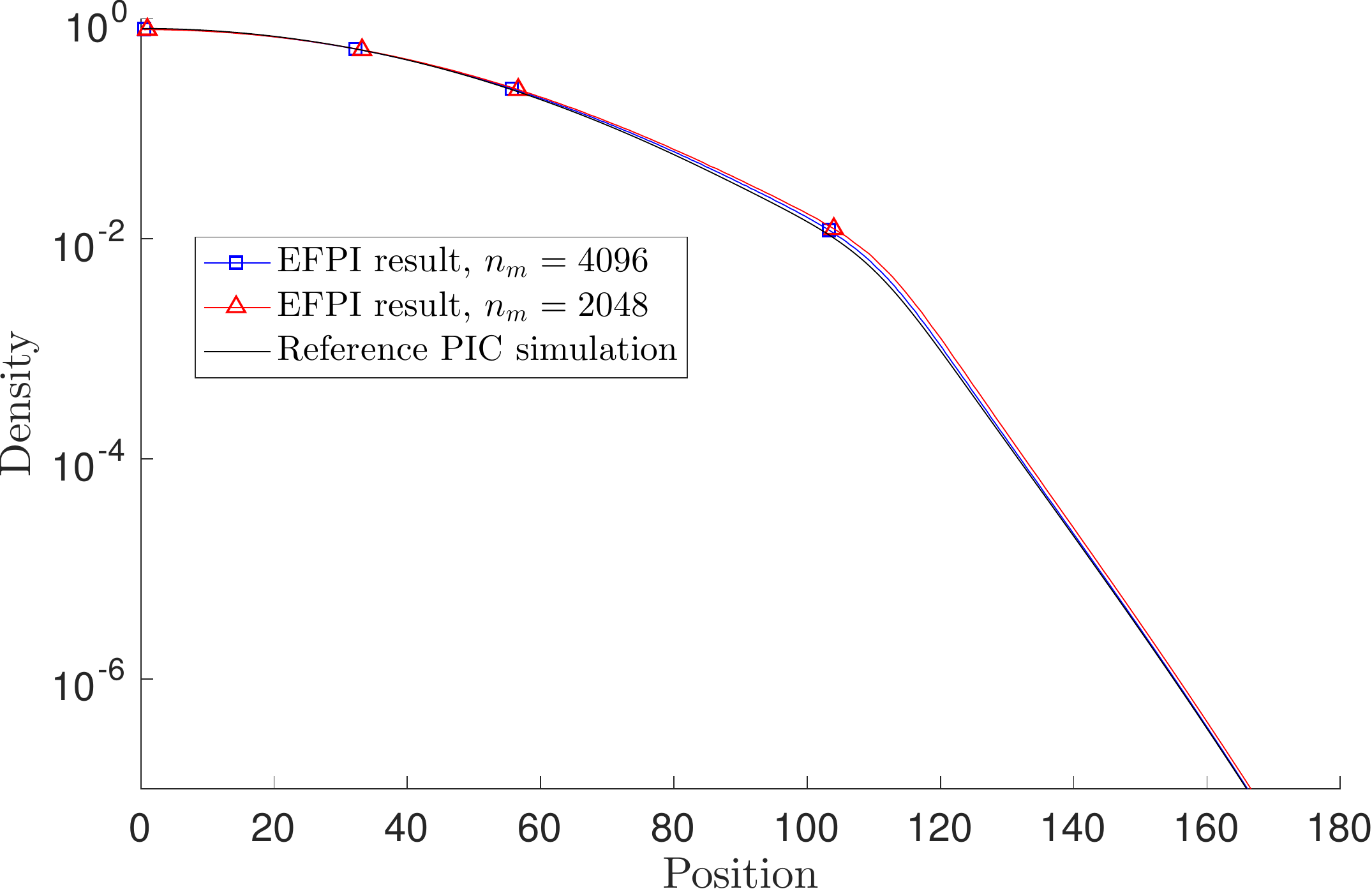}
	\subcaption{Ion density (logarithmic scale) versus $x$}\label{fig:expansion1:b}
\end{minipage}
\\
\begin{minipage}{.49\textwidth}
	\centering
	\includegraphics[width=.8\textwidth]{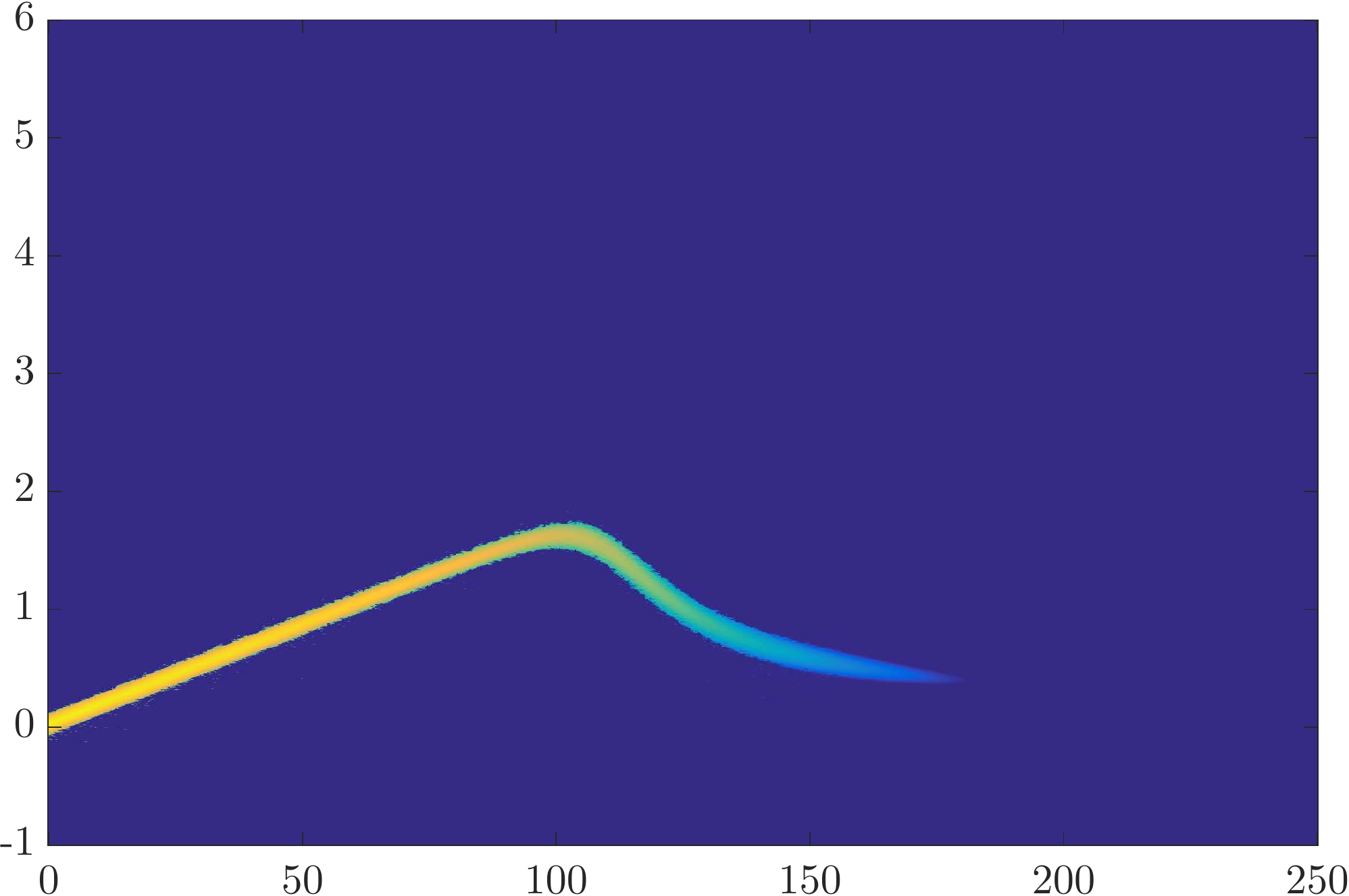}
	\subcaption{EFPI solution (Run $10$, $n_c = 4096$), logarithmic color scale}\label{fig:expansion1:c}
\end{minipage}
\begin{minipage}{.49\textwidth}
	\centering
	\includegraphics[width=.8\textwidth]{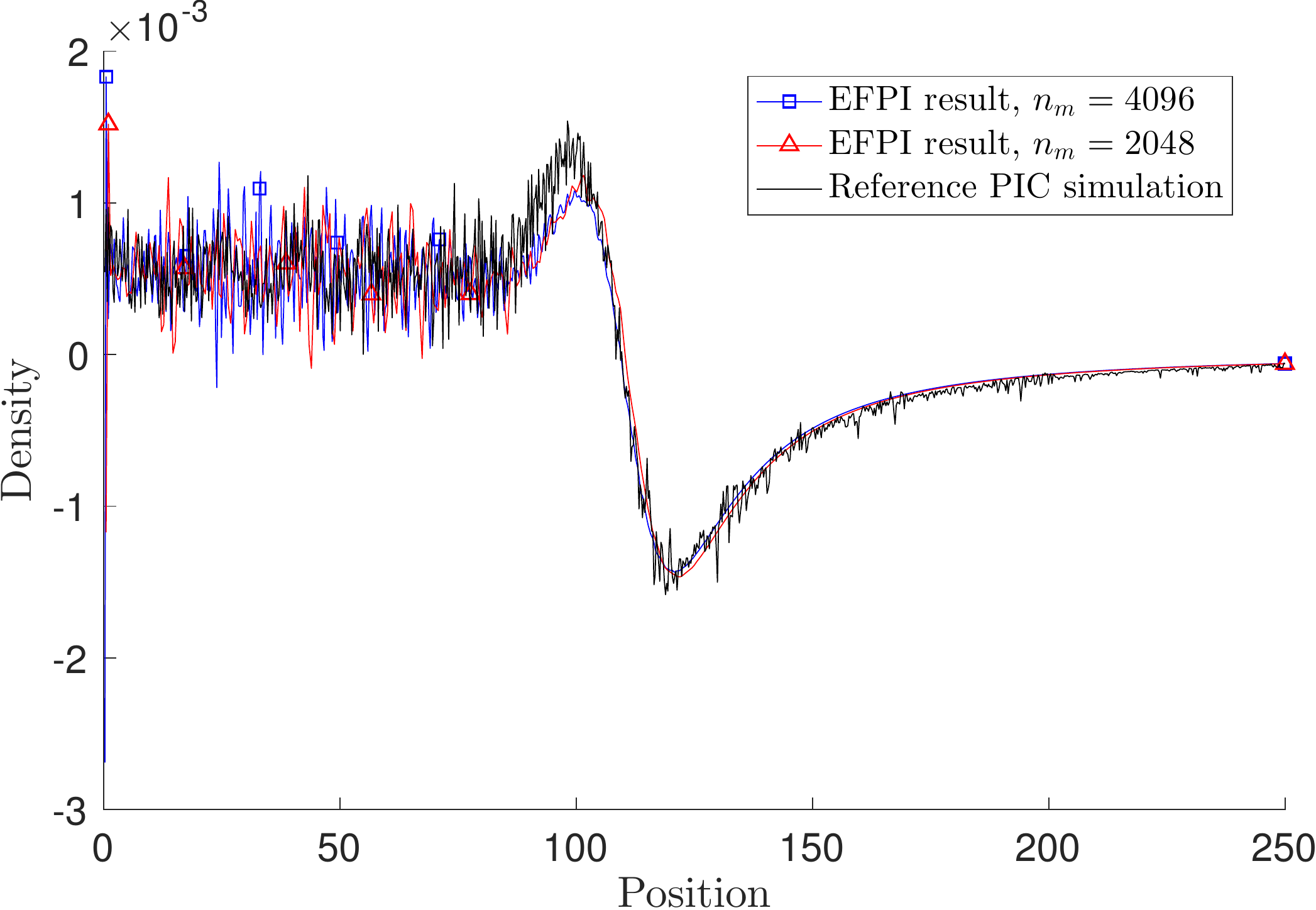}
	\subcaption{Local charge density versus $x$}\label{fig:expansion1:d}
\end{minipage}
\caption{Plasma expansion at $t = 20$: PIC (Run $9$) and EFPI (Runs $10$ and $11$, $n_c = 4096$ and $2048$).}\label{fig:expansion1}

\vspace{10pt}
\centering
\begin{minipage}{.49\textwidth}
	\centering
	\includegraphics[width=.8\textwidth]{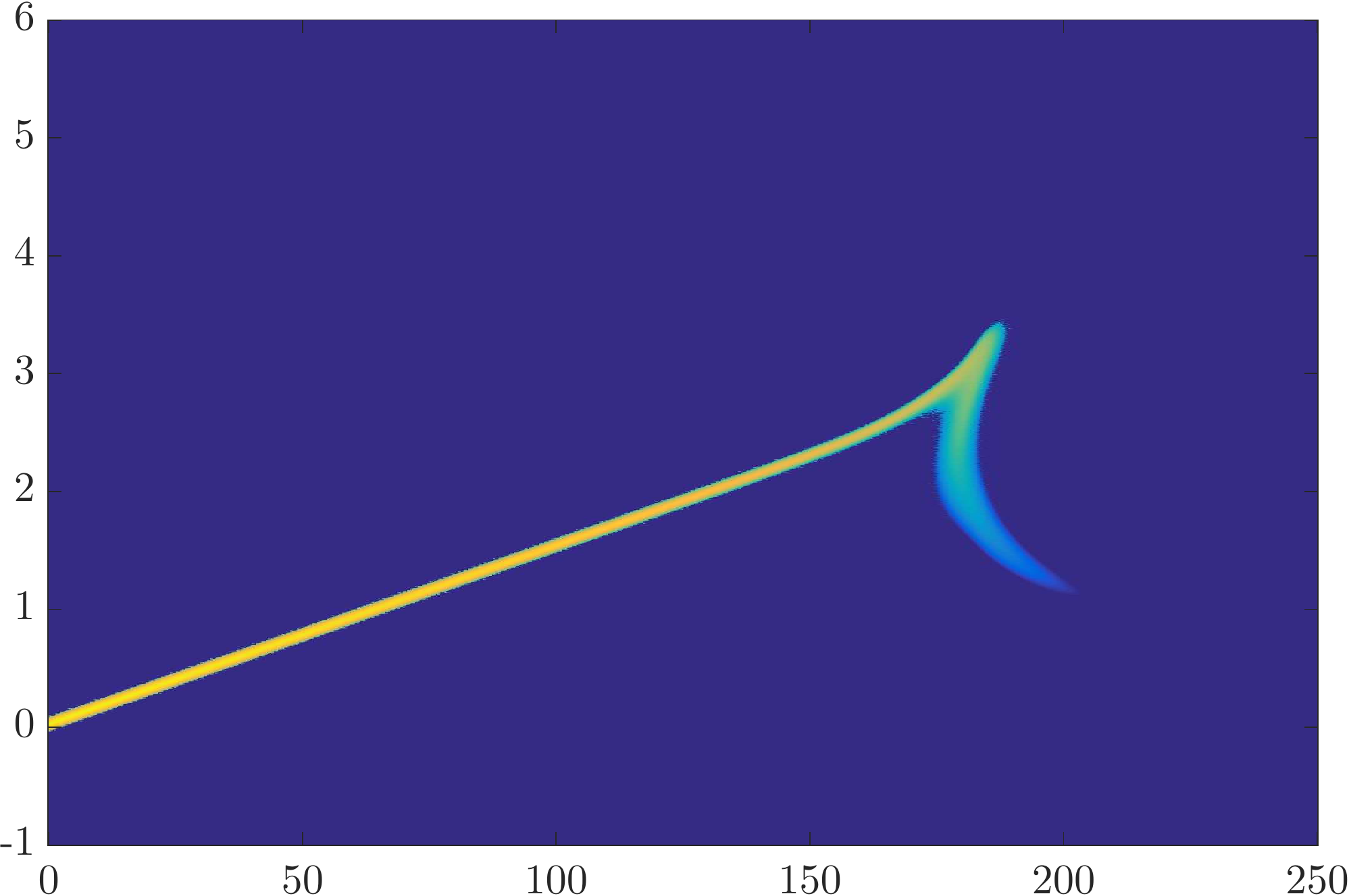}
	\subcaption{Fully resolved PIC solution (Run $9$), logarithmic color scale}\label{fig:expansion2:a}
\end{minipage}
\begin{minipage}{.49\textwidth}
	\centering
	\includegraphics[width=.8\textwidth]{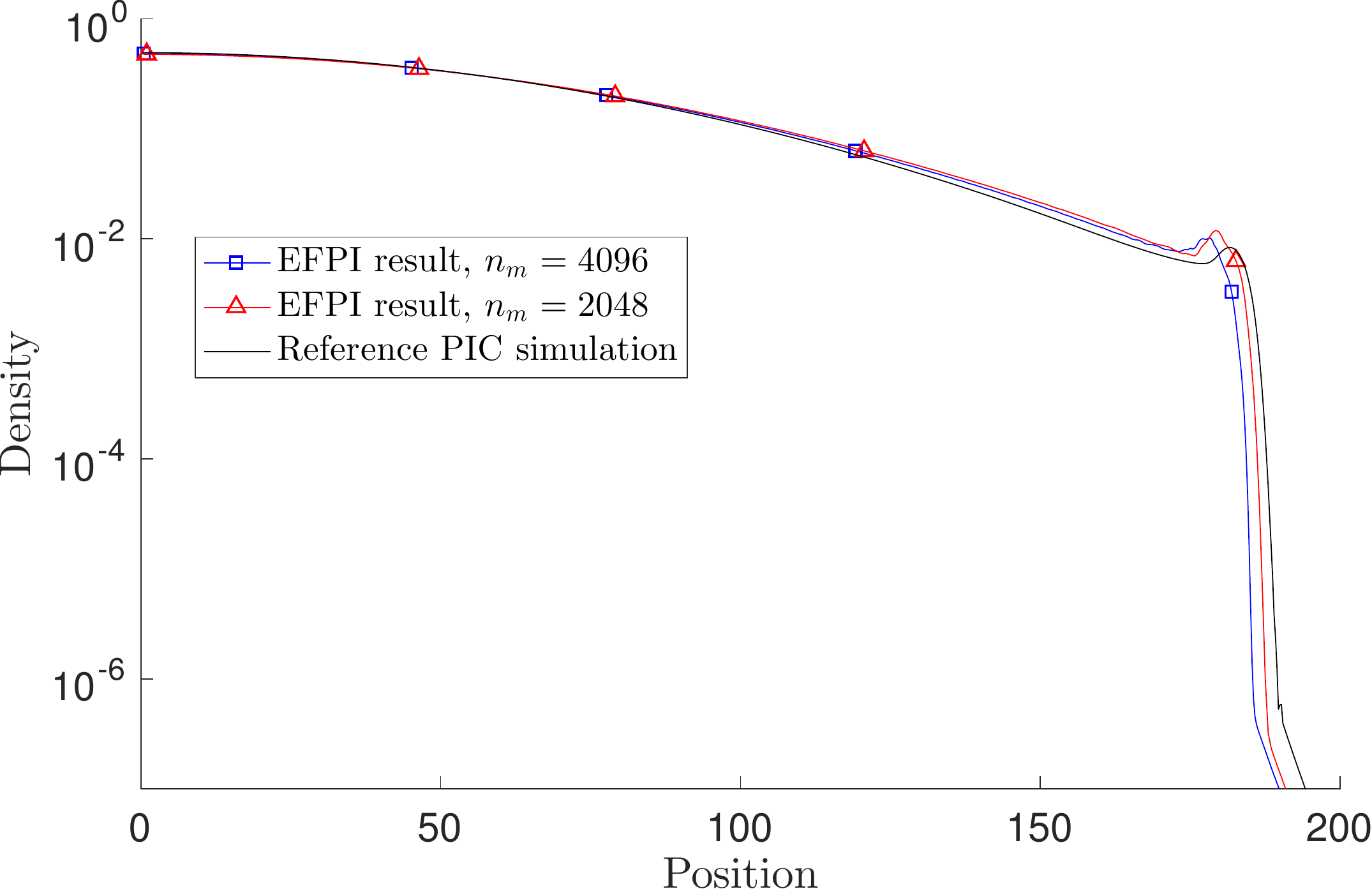}
	\subcaption{Ion density (logarithmic scale)}\label{fig:expansion2:b}
\end{minipage}
\\
\begin{minipage}{.49\textwidth}
	\centering
	\includegraphics[width=.8\textwidth]{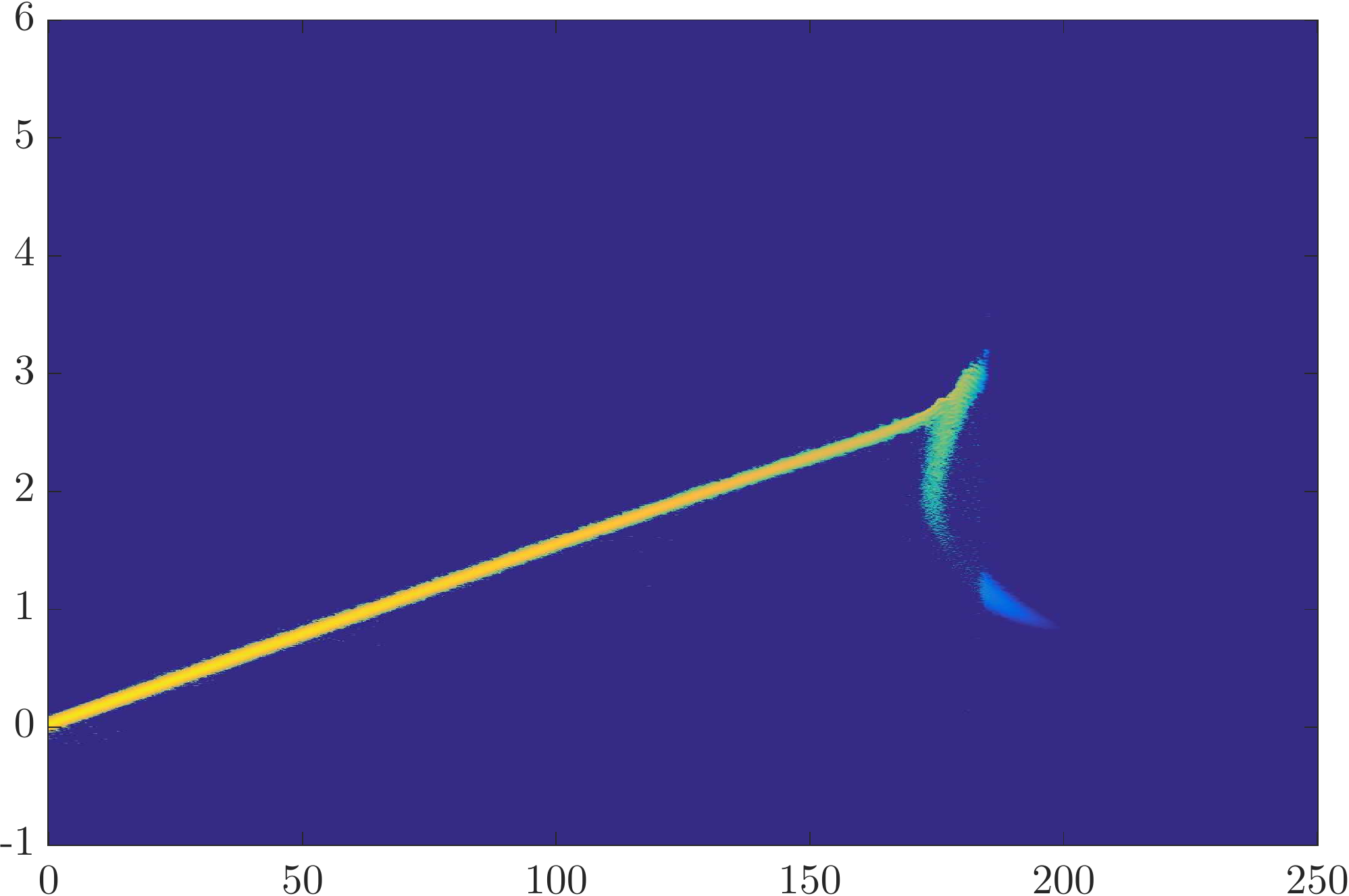}
	\subcaption{EFPI solution (Run $10$, $n_c = 4096$), logarithmic color scale}\label{fig:expansion2:c}
\end{minipage}
\begin{minipage}{.49\textwidth}
	\centering
	\includegraphics[width=.8\textwidth]{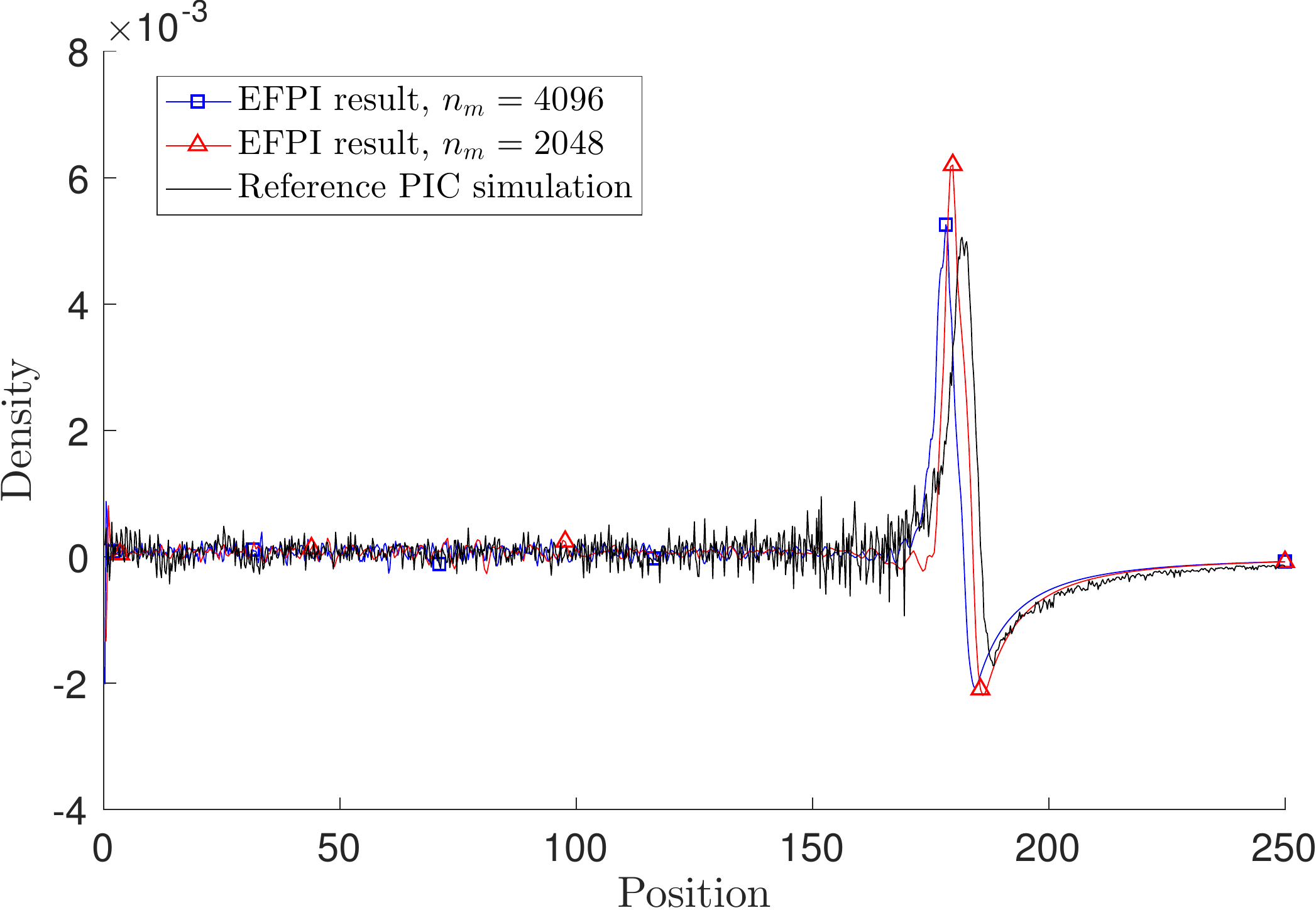}
	\subcaption{Local charge density}\label{fig:expansion2:d}
\end{minipage}
\caption{Plasma expansion at $t = 50$: PIC (Run $9$) and EFPI (Runs $10$ and $11$, $n_c = 4096$ and $2048$).}\label{fig:expansion2}
\end{figure}

Figures~\ref{fig:expansion1} and~\ref{fig:expansion2} show the results for runs~$6$ and~$7$, respectively, at times $t = 20$ and $t = 50$. We plot the unnormalized ion distribution function, the ion density and the local charge density due to both ion and electrons. As time advances, the ion slab expands roughly at the speed of the ion acoustic wave~\cite{Mora_2005}.
 
First, we observe that the proposed method gives a correct account of the speed of expansion of the ion acoustic wave, thus verifying that it correctly captures this wave as also seen earlier. 
%Note that the expansion is noticeably slowed down with other methods (Direct-Implicit or Asymptotic-Preserving) as shown in~\cite{Degond_2010}, in contrary to the case here. 
As it expands, the plasma develops a double layer at the edges of the expansion which is accurately reproduced by the EFPI solution. At later times (see Fig.~\ref{fig:expansion2}, $t = 50$), the solution exhibits wave breaking and the ion distribution function becomes clearly non-maxwellian, an effect which is also reproduced by the EFPI method.
We note also that data has the same or lower level of noise compared to the reference PIC simulation, even as it uses $4$ times less particles (see Table~\ref{tab:expparameters}).

\section{Conclusions}
In this paper we have presented a novel equation-free projective integration method for the Vlasov-Poisson equation. At the coarse level, the ion distribution function is represented on a non-parametric wavelet basis. At the fine-scale level, a particle-in-cell description is implemented for both ions and electrons. A lifting operator and a restriction operator are used to go back and forth consistently between these descriptions, assuming that the electrons are adiabatic. During each macro time-step, an explicit PIC code is stepped forward for a small number of fine-scale time steps to determine the rate of change in the macroscopic variables, which are then projected forward. The effects of statistical shot noise are controlled using both wavelet-based techniques, such as linear thresholding of coefficients, and more usual techniques in the EFPI framework such as coarsening in space and integration with the flow. 

To validate the method and investigate its potential for simulating plasma systems, we applied it to two test cases. 
The first, a well known problem of the nonlinear propagation and steepening of an ion acoustic wave, was proposed in~\cite{Shay_2007} to test the EFREE method as a first attempt at implementing the projective integration framework to model kinetic plasma phenomena. 
We also proposed a second test case, new for projective integration methods: the expansion of an ion slab in a vacuum~\cite{Degond_2010}. These tests have confirmed that the method is stable and allows the use of time and space steps which are much larger than for the standard PIC method, while lifting numerous restrictions of the original EFREE method, most notably the assumption of a Maxwellian ion distribution function.
We observed computational speedups between one to two orders of magnitude in our test runs. While these numbers depend highly on implementation and code optimization, such results are certainly encouraging.

To further increase the accuracy and speedup of the method, the next step is to implement a non-linear thresholding algorithm to adaptively select the relevant wavelet coefficients for the coarse description of the ion distribution function. This will allow the exploitation of the sparsity of the discretization of smooth ion distribution functions in a wavelet basis and reduce the size of the coarse discretization. This will be essential in higher dimensions. This approach should also be combined with a framework for adapting the coarse resolution in time and space. To realize this, new a posteriori error estimators should be developed to measure deviations between the coarse model and the underlying multiscale plasma system. The numerous algorithms which have been developed for applying wavelets to signal analysis and adaptive denoising offer many interesting possibilities for this new application. These various possible research directions show that the method still has a lot of potential for improvement. 

Ultimately, the goal will be to apply the projective integration paradigm to more challenging plasma problems for which direct simulations by ordinary explicit, implicit or hybrid methods is out of reach due to the need to resolve meso-physics with a complex influence on the macroscopic behavior, such as transport driven by meso-scale turbulence.
\section*{Acknowledgment}
This work was partially supported by the Department of Energy, Grant No. DE-SC0008553.
\section*{References}
\bibliography{biblio}

\begin{thebibliography}{47}
\providecommand{\natexlab}[1]{#1}
\providecommand{\url}[1]{\texttt{#1}}
\providecommand{\urlprefix}{URL }
\expandafter\ifx\csname urlstyle\endcsname\relax
  \providecommand{\doi}[1]{doi:\discretionary{}{}{}#1}\else
  \providecommand{\doi}[1]{doi:\discretionary{}{}{}\begingroup
  \urlstyle{rm}\url{#1}\endgroup}\fi
\providecommand{\bibinfo}[2]{#2}

\bibitem[{Gear et~al.(2002)Gear, Kevrekidis, and Theodoropoulos}]{Gear_2002}
\bibinfo{author}{C.~Gear}, \bibinfo{author}{I.~Kevrekidis},
  \bibinfo{author}{C.~Theodoropoulos}, \bibinfo{title}{Coarse
  integration/bifurcation analysis via microscopic simulators:
  micro--{G}alerkin methods}, \bibinfo{journal}{Computational Chemical
  Engineering} \bibinfo{volume}{26} (\bibinfo{year}{2002})
  \bibinfo{pages}{941--}.

\bibitem[{Kevrekidis et~al.(2003)Kevrekidis, Gear, Hyman, Kevrekidis, Runborg,
  Theodoropoulos et~al.}]{Kevrekidis_2003}
\bibinfo{author}{I.~G. Kevrekidis}, \bibinfo{author}{C.~W. Gear},
  \bibinfo{author}{J.~M. Hyman}, \bibinfo{author}{P.~G. Kevrekidis},
  \bibinfo{author}{O.~Runborg}, \bibinfo{author}{C.~Theodoropoulos}, et~al.,
  \bibinfo{title}{Equation-free, coarse-grained multiscale computation:
  Enabling mocroscopic simulators to perform system-level analysis},
  \bibinfo{journal}{Communications in Mathematical Sciences}
  \bibinfo{volume}{1}~(\bibinfo{number}{4}) (\bibinfo{year}{2003})
  \bibinfo{pages}{715--762}.

\bibitem[{Gear and Kevrekidis(2002)}]{Gear_2003}
\bibinfo{author}{C.~Gear}, \bibinfo{author}{I.~Kevrekidis},
  \bibinfo{title}{Projective methods for stiff differential equations: problems
  with gaps in their eigenvalue spectrum}, \bibinfo{journal}{SIAM Journal of
  Scientific Computing} \bibinfo{volume}{24} (\bibinfo{year}{2002})
  \bibinfo{pages}{1091--}.

\bibitem[{Kavousanakis et~al.(2007)Kavousanakis, Erban, Boudouvis, Gear, and
  Kevrekidis}]{Kavousanakis_2007}
\bibinfo{author}{M.~E. Kavousanakis}, \bibinfo{author}{R.~Erban},
  \bibinfo{author}{A.~G. Boudouvis}, \bibinfo{author}{C.~W. Gear},
  \bibinfo{author}{I.~G. Kevrekidis}, \bibinfo{title}{Projective and coarse
  projective integration for problems with continuous symmetries},
  \bibinfo{journal}{Journal of Computational Physics}
  \bibinfo{volume}{225}~(\bibinfo{number}{1}) (\bibinfo{year}{2007})
  \bibinfo{pages}{382--407}.

\bibitem[{Papanicolau et~al.(1978)Papanicolau, Bensoussan, and
  Lions}]{Papanicolau_1978}
\bibinfo{author}{G.~Papanicolau}, \bibinfo{author}{A.~Bensoussan},
  \bibinfo{author}{J.~Lions}, \bibinfo{title}{Asymptotic analysis for periodic
  structures}, \bibinfo{publisher}{Elsevier}, \bibinfo{year}{1978}.

\bibitem[{Birdsall and Langdon(2004)}]{Birdsall_Langdon_1985}
\bibinfo{author}{C.~K. Birdsall}, \bibinfo{author}{A.~B. Langdon},
  \bibinfo{title}{Plasma physics via computer simulation},
  \bibinfo{publisher}{Taylor \& Francis}, \bibinfo{year}{2004}.

\bibitem[{Denavit(1981)}]{Denavit_1981}
\bibinfo{author}{J.~Denavit}, \bibinfo{title}{Time-filtering particle
  simulations with {$\omega_{pe}\Delta t \gg 1$}}, \bibinfo{journal}{Journal of
  Computational Physics} \bibinfo{volume}{42}~(\bibinfo{number}{2})
  (\bibinfo{year}{1981}) \bibinfo{pages}{337--366}.

\bibitem[{Cohen et~al.(1982)Cohen, Langdon, and Friedman}]{Cohen_1982}
\bibinfo{author}{B.~I. Cohen}, \bibinfo{author}{A.~B. Langdon},
  \bibinfo{author}{A.~Friedman}, \bibinfo{title}{Implicit time integration for
  plasma simulation}, \bibinfo{journal}{Journal of Computational Physics}
  \bibinfo{volume}{46}~(\bibinfo{number}{1}) (\bibinfo{year}{1982})
  \bibinfo{pages}{15--38}.

\bibitem[{Langdon et~al.(1983)Langdon, Cohen, and Friedman}]{Langdon_1983}
\bibinfo{author}{A.~B. Langdon}, \bibinfo{author}{B.~I. Cohen},
  \bibinfo{author}{A.~Friedman}, \bibinfo{title}{Direct implicit large
  time-step particle simulation of plasmas}, \bibinfo{journal}{Journal of
  Computational Physics} \bibinfo{volume}{51}~(\bibinfo{number}{1})
  (\bibinfo{year}{1983}) \bibinfo{pages}{107--138}.

\bibitem[{Ricci et~al.(2002)Ricci, Lapenta, and Brackbill}]{Ricci_2002}
\bibinfo{author}{P.~Ricci}, \bibinfo{author}{G.~Lapenta},
  \bibinfo{author}{J.~Brackbill}, \bibinfo{title}{GEM reconnection challenge:
  {I}mplicit kinetic simulations with the physical mass ratio},
  \bibinfo{journal}{Geophysical research letters}
  \bibinfo{volume}{29}~(\bibinfo{number}{23}).

\bibitem[{Lapenta et~al.(2006)Lapenta, Brackbill, and Ricci}]{Lapenta_2006}
\bibinfo{author}{G.~Lapenta}, \bibinfo{author}{J.~Brackbill},
  \bibinfo{author}{P.~Ricci}, \bibinfo{title}{Kinetic approach to
  microscopic-macroscopic coupling in space and laboratory plasmasa)},
  \bibinfo{journal}{Physics of Plasmas (1994-present)}
  \bibinfo{volume}{13}~(\bibinfo{number}{5}) (\bibinfo{year}{2006})
  \bibinfo{pages}{055904}.

\bibitem[{Markidis et~al.(2010)Markidis, Lapenta et~al.}]{Markidis_2010}
\bibinfo{author}{S.~Markidis}, \bibinfo{author}{G.~Lapenta}, et~al.,
  \bibinfo{title}{Multi-scale simulations of plasma with iPIC3D},
  \bibinfo{journal}{Mathematics and Computers in Simulation}
  \bibinfo{volume}{80}~(\bibinfo{number}{7}) (\bibinfo{year}{2010})
  \bibinfo{pages}{1509--1519}.

\bibitem[{Mason(1981)}]{Mason_1981}
\bibinfo{author}{R.~J. Mason}, \bibinfo{title}{Implicit moment particle
  simulation of plasmas}, \bibinfo{journal}{Journal of Computational Physics}
  \bibinfo{volume}{41}~(\bibinfo{number}{2}) (\bibinfo{year}{1981})
  \bibinfo{pages}{233--244}.

\bibitem[{Mason(1983)}]{Mason_1983}
\bibinfo{author}{R.~J. Mason}, \bibinfo{title}{Implicit moment PIC-hybrid
  simulation of collisional plasmas}, \bibinfo{journal}{Journal of
  Computational Physics} \bibinfo{volume}{51}~(\bibinfo{number}{3})
  (\bibinfo{year}{1983}) \bibinfo{pages}{484--501}.

\bibitem[{Lipatov(2002)}]{Lipatov_2002}
\bibinfo{author}{A.~S. Lipatov}, \bibinfo{title}{The hybrid multiscale
  simulation technology: an introduction with application to astrophysical and
  laboratory plasmas}, \bibinfo{publisher}{Springer}, \bibinfo{year}{2002}.

\bibitem[{Friedman et~al.(1991)Friedman, Parker, Ray, and
  Birdsall}]{Friedman_1991}
\bibinfo{author}{A.~Friedman}, \bibinfo{author}{S.~Parker},
  \bibinfo{author}{S.~Ray}, \bibinfo{author}{C.~Birdsall},
  \bibinfo{title}{Multi-scale particle-in-cell plasma simulation},
  \bibinfo{journal}{Journal of Computational Physics}
  \bibinfo{volume}{96}~(\bibinfo{number}{1}) (\bibinfo{year}{1991})
  \bibinfo{pages}{54--70}.

\bibitem[{Parker et~al.(1993)Parker, Friedman, Ray, and Birdsall}]{Parker_1993}
\bibinfo{author}{S.~Parker}, \bibinfo{author}{A.~Friedman},
  \bibinfo{author}{S.~Ray}, \bibinfo{author}{C.~Birdsall},
  \bibinfo{title}{Bounded multi-scale plasma simulation: Application to sheath
  problems}, \bibinfo{journal}{Journal of Computational Physics}
  \bibinfo{volume}{107}~(\bibinfo{number}{2}) (\bibinfo{year}{1993})
  \bibinfo{pages}{388--402}.

\bibitem[{Chen et~al.(2011)Chen, Chac{\'o}n, and Barnes}]{Chen2011}
\bibinfo{author}{G.~Chen}, \bibinfo{author}{L.~Chac{\'o}n},
  \bibinfo{author}{D.~C. Barnes}, \bibinfo{title}{An energy-and
  charge-conserving, implicit, electrostatic particle-in-cell algorithm},
  \bibinfo{journal}{Journal of Computational Physics}
  \bibinfo{volume}{230}~(\bibinfo{number}{18}) (\bibinfo{year}{2011})
  \bibinfo{pages}{7018--7036}.

\bibitem[{Markidis and Lapenta(2011)}]{Markidis2011}
\bibinfo{author}{S.~Markidis}, \bibinfo{author}{G.~Lapenta},
  \bibinfo{title}{The energy conserving particle-in-cell method},
  \bibinfo{journal}{Journal of Computational Physics}
  \bibinfo{volume}{230}~(\bibinfo{number}{18}) (\bibinfo{year}{2011})
  \bibinfo{pages}{7037--7052}.

\bibitem[{Chen and Chacon(2014)}]{Chen2014}
\bibinfo{author}{G.~Chen}, \bibinfo{author}{L.~Chacon}, \bibinfo{title}{An
  energy-and charge-conserving, nonlinearly implicit, electromagnetic 1D-3V
  Vlasov--Darwin particle-in-cell algorithm}, \bibinfo{journal}{Computer
  Physics Communications} \bibinfo{volume}{185}~(\bibinfo{number}{10})
  (\bibinfo{year}{2014}) \bibinfo{pages}{2391--2402}.

\bibitem[{Chen and Chacon(2015)}]{Chen2015}
\bibinfo{author}{G.~Chen}, \bibinfo{author}{L.~Chacon}, \bibinfo{title}{A
  multi-dimensional, energy-and charge-conserving, nonlinearly implicit,
  electromagnetic Vlasov--Darwin particle-in-cell algorithm},
  \bibinfo{journal}{Computer Physics Communications} \bibinfo{volume}{197}
  (\bibinfo{year}{2015}) \bibinfo{pages}{73--87}.

\bibitem[{Chac{\'o}n and Chen(2016)}]{Chacon2016}
\bibinfo{author}{L.~Chac{\'o}n}, \bibinfo{author}{G.~Chen}, \bibinfo{title}{A
  curvilinear, fully implicit, conservative electromagnetic PIC algorithm in
  multiple dimensions}, \bibinfo{journal}{Journal of Computational Physics}
  \bibinfo{volume}{316} (\bibinfo{year}{2016}) \bibinfo{pages}{578--597}.

\bibitem[{Degond et~al.(2006)Degond, Deluzet, and Navoret}]{Degond_2006}
\bibinfo{author}{P.~Degond}, \bibinfo{author}{F.~Deluzet},
  \bibinfo{author}{L.~Navoret}, \bibinfo{title}{An asymptotically stable
  Particle-in-Cell ({PIC}) scheme for collisionless plasma simulations near
  quasineutrality}, \bibinfo{journal}{Comptes Rendus Math{\'e}matique}
  \bibinfo{volume}{343}~(\bibinfo{number}{9}) (\bibinfo{year}{2006})
  \bibinfo{pages}{613--618}.

\bibitem[{Degond et~al.(2010)Degond, Deluzet, Navoret, Sun, and
  Vignal}]{Degond_2010}
\bibinfo{author}{P.~Degond}, \bibinfo{author}{F.~Deluzet},
  \bibinfo{author}{L.~Navoret}, \bibinfo{author}{A.-B. Sun},
  \bibinfo{author}{M.-H. Vignal}, \bibinfo{title}{Asymptotic-preserving
  particle-in-cell method for the {V}lasov--{P}oisson system near
  quasineutrality}, \bibinfo{journal}{Journal of Computational Physics}
  \bibinfo{volume}{229}~(\bibinfo{number}{16}) (\bibinfo{year}{2010})
  \bibinfo{pages}{5630--5652}.

\bibitem[{Shay et~al.(2007)Shay, Drake, and Dorland}]{Shay_2007}
\bibinfo{author}{M.~A. Shay}, \bibinfo{author}{J.~F. Drake},
  \bibinfo{author}{B.~Dorland}, \bibinfo{title}{Equation free projective
  integration: A multiscale method applied to a plasma ion acoustic wave},
  \bibinfo{journal}{Journal of Computational Physics}
  \bibinfo{volume}{226}~(\bibinfo{number}{1}) (\bibinfo{year}{2007})
  \bibinfo{pages}{571--585}.

\bibitem[{Maluckov et~al.(2008)Maluckov, Ishiguro, and
  {\v{S}}koric}]{Maluckov_2008}
\bibinfo{author}{A.~M. Maluckov}, \bibinfo{author}{S.~Ishiguro},
  \bibinfo{author}{M.~M. {\v{S}}koric}, \bibinfo{title}{A Macro Projective
  Integration Method in 2D Microscopic System Applied to Nonlinear Ion Acoustic
  Waves in a Plasma}, \bibinfo{journal}{Communications in Computational
  Physics} \bibinfo{volume}{4}~(\bibinfo{number}{3}) (\bibinfo{year}{2008})
  \bibinfo{pages}{556--574}.

\bibitem[{Mallat(1999)}]{Mallat_1999}
\bibinfo{author}{S.~Mallat}, \bibinfo{title}{A wavelet tour of signal
  processing}, \bibinfo{publisher}{Academic press}, \bibinfo{year}{1999}.

\bibitem[{Donoho et~al.(1996)Donoho, Johnstone, Kerkyacharian, and
  Picard}]{Donoho_1996}
\bibinfo{author}{D.~L. Donoho}, \bibinfo{author}{I.~M. Johnstone},
  \bibinfo{author}{G.~Kerkyacharian}, \bibinfo{author}{D.~Picard},
  \bibinfo{title}{Density estimation by wavelet thresholding},
  \bibinfo{journal}{The Annals of Statistics}  (\bibinfo{year}{1996})
  \bibinfo{pages}{508--539}.

\bibitem[{Nguyen Van~Yen et~al.(2010)Nguyen Van~Yen, del Castillo-Negrete,
  Schneider, Farge, and Chen}]{Nguyen_2010}
\bibinfo{author}{R.~Nguyen Van~Yen}, \bibinfo{author}{D.~del Castillo-Negrete},
  \bibinfo{author}{K.~Schneider}, \bibinfo{author}{M.~Farge},
  \bibinfo{author}{G.~Chen}, \bibinfo{title}{Wavelet-based density estimation
  for noise reduction in plasma simulations using particles},
  \bibinfo{journal}{Journal of Computational Physics}
  \bibinfo{volume}{229}~(\bibinfo{number}{8}) (\bibinfo{year}{2010})
  \bibinfo{pages}{2821--2839}.

\bibitem[{Luisier et~al.(2010)Luisier, Vonesch, Blu, and Unser}]{Luisier_2010}
\bibinfo{author}{F.~Luisier}, \bibinfo{author}{C.~Vonesch},
  \bibinfo{author}{T.~Blu}, \bibinfo{author}{M.~Unser}, \bibinfo{title}{Fast
  interscale wavelet denoising of Poisson-corrupted images},
  \bibinfo{journal}{Signal Processing}
  \bibinfo{volume}{90}~(\bibinfo{number}{2}) (\bibinfo{year}{2010})
  \bibinfo{pages}{415--427}.

\bibitem[{Mora(2005{\natexlab{a}})}]{Mora2005}
\bibinfo{author}{P.~Mora}, \bibinfo{title}{Collisionless expansion of a
  Gaussian plasma into a vacuum}, \bibinfo{journal}{Physics of Plasmas}
  \bibinfo{volume}{12}~(\bibinfo{number}{11}) \bibinfo{eid}{112102}.

\bibitem[{Cottet and Raviart(1984)}]{Cottet_1984}
\bibinfo{author}{G.-H. Cottet}, \bibinfo{author}{P.-A. Raviart},
  \bibinfo{title}{Particle methods for the one-dimensional {V}lasov-{P}oisson
  equations}, \bibinfo{journal}{SIAM journal on numerical analysis}
  \bibinfo{volume}{21}~(\bibinfo{number}{1}) (\bibinfo{year}{1984})
  \bibinfo{pages}{52--76}.

\bibitem[{Ganguly and Victory(1989)}]{Ganguly_1989}
\bibinfo{author}{K.~Ganguly}, \bibinfo{author}{H.~Victory, Jr},
  \bibinfo{title}{On the convergence of particle methods for multidimensional
  Vlasov-Poisson systems}, \bibinfo{journal}{SIAM journal on numerical
  analysis} \bibinfo{volume}{26}~(\bibinfo{number}{2}) (\bibinfo{year}{1989})
  \bibinfo{pages}{249--288}.

\bibitem[{Wollman and Ozizmir(1996)}]{Wollman_1996}
\bibinfo{author}{S.~Wollman}, \bibinfo{author}{E.~Ozizmir},
  \bibinfo{title}{Numerical approximation of the one-dimensional
  {V}lasov-{P}oisson system with periodic boundary conditions},
  \bibinfo{journal}{SIAM journal on numerical analysis}
  \bibinfo{volume}{33}~(\bibinfo{number}{4}) (\bibinfo{year}{1996})
  \bibinfo{pages}{1377--1409}.

\bibitem[{Nevins et~al.(2005)Nevins, Hammett, Dimits, Dorland, and
  Shumaker}]{Nevins_2005}
\bibinfo{author}{W.~Nevins}, \bibinfo{author}{G.~Hammett},
  \bibinfo{author}{A.~Dimits}, \bibinfo{author}{W.~Dorland},
  \bibinfo{author}{D.~Shumaker}, \bibinfo{title}{Discrete particle noise in
  particle-in-cell simulations of plasma microturbulence},
  \bibinfo{journal}{Physics of Plasmas}
  \bibinfo{volume}{12}~(\bibinfo{number}{12}) (\bibinfo{year}{2005})
  \bibinfo{pages}{122305}.

\bibitem[{Krall and Trivelpiece(1986)}]{Krall_1986}
\bibinfo{author}{N.~A. Krall}, \bibinfo{author}{A.~W. Trivelpiece},
  \bibinfo{title}{Principles of Plasma Physics}, \bibinfo{publisher}{San
  Francisco Press, Inc., San Francisco, CA}, \bibinfo{year}{1986}.

\bibitem[{Daubechies(1993)}]{Daubechies1993}
\bibinfo{author}{I.~Daubechies}, \bibinfo{title}{Orthonormal bases of compactly
  supported wavelets II. Variations on a theme}, \bibinfo{journal}{SIAM Journal
  on Mathematical Analysis} \bibinfo{volume}{24}~(\bibinfo{number}{2})
  (\bibinfo{year}{1993}) \bibinfo{pages}{499--519}.

\bibitem[{Taitano et~al.(2013)Taitano, Knoll, Chac{\'o}n, and
  Chen}]{Taitano2013}
\bibinfo{author}{W.~T. Taitano}, \bibinfo{author}{D.~A. Knoll},
  \bibinfo{author}{L.~Chac{\'o}n}, \bibinfo{author}{G.~Chen},
  \bibinfo{title}{Development of a Consistent and Stable Fully Implicit Moment
  Method for Vlasov--Amp{\`e}re Particle in Cell (PIC) System},
  \bibinfo{journal}{SIAM Journal on Scientific Computing}
  \bibinfo{volume}{35}~(\bibinfo{number}{5}) (\bibinfo{year}{2013})
  \bibinfo{pages}{S126--S149}.

\bibitem[{Schamel(1973)}]{Schamel1973}
\bibinfo{author}{H.~Schamel}, \bibinfo{title}{A modified Korteweg-de Vries
  equation for ion acoustic wavess due to resonant electrons},
  \bibinfo{journal}{Journal of Plasma Physics}
  \bibinfo{volume}{9}~(\bibinfo{number}{03}) (\bibinfo{year}{1973})
  \bibinfo{pages}{377--387}.

\bibitem[{Dorozhkina and Semenov(1998)}]{Dorozhkina1998}
\bibinfo{author}{D.~Dorozhkina}, \bibinfo{author}{V.~Semenov},
  \bibinfo{title}{Exact solution of Vlasov equations for quasineutral expansion
  of plasma bunch into vacuum}, \bibinfo{journal}{Physical Review Letters}
  \bibinfo{volume}{81}~(\bibinfo{number}{13}) (\bibinfo{year}{1998})
  \bibinfo{pages}{2691}.

\bibitem[{Baitin and Kuzanyan(1998)}]{Baitin1998}
\bibinfo{author}{A.~Baitin}, \bibinfo{author}{K.~Kuzanyan}, \bibinfo{title}{A
  self-similar solution for expansion into a vacuum of a collisionless plasma
  bunch}, \bibinfo{journal}{Journal of Plasma physics}
  \bibinfo{volume}{59}~(\bibinfo{number}{01}) (\bibinfo{year}{1998})
  \bibinfo{pages}{83--90}.

\bibitem[{Kovalev et~al.(2002)Kovalev, Bychenkov, and Tikhonchuk}]{Kovalev2002}
\bibinfo{author}{V.~Kovalev}, \bibinfo{author}{V.~Bychenkov},
  \bibinfo{author}{V.~Tikhonchuk}, \bibinfo{title}{Particle dynamics during
  adiabatic expansion of a plasma bunch}, \bibinfo{journal}{Journal of
  Experimental and Theoretical Physics}
  \bibinfo{volume}{95}~(\bibinfo{number}{2}) (\bibinfo{year}{2002})
  \bibinfo{pages}{226--241}.

\bibitem[{Kovalev and Bychenkov(2003)}]{Kovalev2003}
\bibinfo{author}{V.~Kovalev}, \bibinfo{author}{V.~Y. Bychenkov},
  \bibinfo{title}{Analytic solutions to the Vlasov equations for expanding
  plasmas}, \bibinfo{journal}{Physical Review Letters}
  \bibinfo{volume}{90}~(\bibinfo{number}{18}) (\bibinfo{year}{2003})
  \bibinfo{pages}{185004}.

\bibitem[{Grismayer(2006)}]{Grismayer_2006}
\bibinfo{author}{T.~Grismayer}, \bibinfo{title}{Etude th{\'e}orique et
  num{\'e}rique de l'expansion d'un plasma cr{\'e\'e} par laser :
  acc{\'e}l{\'e}ration d'ions haute {\'e}nergie}, Ph.D. thesis,
  \bibinfo{school}{Ecole Polytechnique, Palaiseau, France},
  \bibinfo{year}{2006}.

\bibitem[{Grismayer et~al.(2008)Grismayer, Mora, Adam, and
  H{\'e}ron}]{Grismayer_2008}
\bibinfo{author}{T.~Grismayer}, \bibinfo{author}{P.~Mora},
  \bibinfo{author}{J.~Adam}, \bibinfo{author}{A.~H{\'e}ron},
  \bibinfo{title}{Electron kinetic effects in plasma expansion and ion
  acceleration}, \bibinfo{journal}{Physical Review E}
  \bibinfo{volume}{77}~(\bibinfo{number}{6}) (\bibinfo{year}{2008})
  \bibinfo{pages}{066407}.

\bibitem[{Mora(2005{\natexlab{b}})}]{Mora_2005}
\bibinfo{author}{P.~Mora}, \bibinfo{title}{Thin-foil expansion into a vacuum},
  \bibinfo{journal}{Physical Review E}
  \bibinfo{volume}{72}~(\bibinfo{number}{5})
  (\bibinfo{year}{2005}{\natexlab{b}}) \bibinfo{pages}{056401}.

\bibitem[{Hammersley(1960)}]{Hammersley1960}
\bibinfo{author}{J.~M. Hammersley}, \bibinfo{title}{Monte Carlo methods for
  solving multivariable problems}, \bibinfo{journal}{Annals of the New York
  Academy of Sciences} \bibinfo{volume}{86}~(\bibinfo{number}{3})
  (\bibinfo{year}{1960}) \bibinfo{pages}{844--874}.

\end{thebibliography}
\section*{Appendix}
\appendix
\renewcommand*{\thesection}{\Alph{section}}

\section{A primer on wavelets}\label{sec:PrimerWavelets}
For completeness, let us recall some notions on wavelets, see e.g.~\cite{Mallat_1999} for further details. 

\paragraph{Continuous decomposition} For any function $f\in L^2(\R)$ and a starting resolution level $J_\mathrm{f}$, representation in the wavelet basis is given by the reconstruction formula:
\begin{equation}\label{eq:reconstructionformula}
f(y) = \sum_{k = -\infty}^\infty \overline{f}_{J_\mathrm{f},k} \phi_{J_\mathrm{f},k}(y) + \sum_{j = J_\mathrm{f}}^\infty \sum_{k = -\infty}^\infty \widetilde{f}_{j,k} \psi_{j,k}(y)
\end{equation} 
where
\begin{equation}
\begin{aligned}
\phi_{j,k}(y) &= 2^{j/2} \phi \left (2^j y - k \right ),\\
\psi_{j,k}(y) &= 2^{j/2} \psi \left (2^j y - k \right ),
\end{aligned}
\end{equation}
are scaled and translated versions of the father wavelet (or scaling function) $\phi(y)$ and the mother wavelet $\psi(y)$. 
We will assume that $\psi$ has zero average and compact support, and that the wavelet family $\left (\psi_{j,k}(y) \right )_{j \in \N, k \in \Z}$ is an orthonormal family. One can show existence of the scaling function $\phi$, defined such that 
the orthogonal complement in $L^2(\R)$ of the linear space spanned by the wavelets is itself orthogonally spanned by the translates of the function $\phi$.

The coefficients $\widetilde{f}_{j,k} = \langle f \lvert\psi_{j,k} \rangle = \int_{\R} f \psi_{j,k}$ of a function $f$ in this wavelet family, called detail coefficients, describe fluctuations of the function $f$ at scale $2^{-j}$ around the position $k / 2^j$. Large values of $j$ correspond to fine scales, small values to coarse scales. The coefficients $\overline{f}_{j,k} = \langle f \lvert\phi_{j,k} \rangle$ are called the scaling coefficients.
The first sum on the right-hand side of formula~\eqref{eq:reconstructionformula} forms a smooth approximation of $f$ at the coarse scale $2^{-J_\mathrm{f}}$, and the second sum represents a correction formed by the addition of details at successively finer scales.

In this work, we employ the R-Coiflet wavelet family of order 4, chosen for its regularity and symmetry properties~\cite{Daubechies1993}:
\[
\int_\R y^m \psi(y) = 0, \qquad 0 \leq m < 4, \qquad \int_\R y^m \phi(y) = 0, \qquad 1 \leq m < 4.
\]
This implies that the first four moments of the distribution function depend only on its scaling coefficients, and not on the detail coefficients.

\paragraph{Periodization}
Let us note that the reconstruction formula~\eqref{eq:reconstructionformula} involves an infinite sum over the position index $k$. In practice we are interested only in functions defined over a bounded domain, and the commonly used alternative is to periodize the wavelet transform over this domain~\cite{Mallat_1999}. Assuming that the coordinates are rescaled so that all particles lie in $[0,1]$, wavelets and scaling functions are replaced by their periodic counterparts:
\begin{equation}
	\begin{aligned}
		\phi_{j,k}(y) & \longrightarrow \sum_{i \in \Z} \phi_{j,k}(y + l), \\
		\psi_{j,k}(y) & \longrightarrow \sum_{i \in \Z} \psi_{j,k}(y + l).
	\end{aligned}
\end{equation}
Throughout this paper we use such periodic wavelets.

\paragraph{Discrete wavelet decomposition}
In practice, a distribution function $f$ is typically represented numerically by its histogram $f^H$ on a regular grid with $2^{J_\mathrm{f}}$ points $x_k = 2^{-J_\mathrm{f}} k$, $0 \leq k \leq 2^{J_\mathrm{f}}-1$. Wavelet-based density estimation algorithms~\cite{Donoho_1996,Nguyen_2010} proceed by approximating the scaling coefficients at the finest scale $J_\mathrm{f}$ by:
\begin{equation}\label{def:ScalingCoeffHistogram}
	\overline{f}_{J_\mathrm{f},k} \approx 2^{-J_\mathrm{f}/2} f^H(2^{-J_\mathrm{f}} k).	
\end{equation}
The fast Discrete Wavelet Transform~\cite{Mallat_1999} can then be used to deduce the wavelet coefficients at coarser scales ($j \leq J_\mathrm{f}$). Filtering algorithms can then choose coefficients to be discarded according to some linear or nonlinear criteria~\cite{Donoho_1996}.

\section{Particle sampling and co-evolving frame integration}
\label{sec:MovingFrame}
We describe here the co-evolving frame projection which is used in the restriction step of our proposed method, see Section~\ref{sec:wefree}. Let $\Delta t$ the macroscopic time step, $t_n = n \Delta t$. Suppose that at time $s \in [t_n, t_{n+1})$, we know the ion distribution function $f^\mathrm{i}(x,v,s)$. The exact solution $f^\mathrm{i}$ solution of the system~(\ref{eq:vlasov}--\ref{eq:poisson1}) is constant along the characteristics: for a given point $(X_0,V_0)$ in phase space, let $\left ( X(t),V(t) \right )$ be given by Newton's law,
\begin{equation}\label{eq:characteristics}
	\begin{aligned}
		\frac{\mathrm{d}X}{\mathrm{d}t} &= V, &\frac{\mathrm{d}V}{\mathrm{d}t} &= - \frac{e}{m_\textrm{i}} \nabla_x \phi (X,t), \\
		X(s) &= X_0, &V(s) &= V_0.
	\end{aligned}
\end{equation}
We have that $f^\mathrm{i}(X(t), V(t), t) = f^\mathrm{i}(X_0, V_0, s)$. 
Now the self-consistent potential $\phi$ is typically slow-moving and quite smooth.
Assuming that this is the case and $\phi(x,t_n) \equiv \phi_n(x)$ is a good approximation for $\phi(x,s)$ for $s \in [t_n , t_{n+1}]$, a simple and explicit approximation to the characteristics is obtained by
\begin{equation}\label{eq:approxcharacteristics}
	\widetilde{X}_{(n,n+1)}(s,t) = X_0 + \left (t - s \right ) V_0, \qquad \widetilde{V}_{(n,n+1)}(s,t) = V_0 -  \frac{e}{m_\textrm{i}} \left (t - s \right )  \nabla_x \widetilde{\phi}_n (X_0) .
\end{equation}
Let us introduce finally the approximately integrated solution $\widetilde{f}_{(n,n+1)}^\mathrm{i}(x,v,s)$ defined implicitely by
\begin{equation}
	\widetilde{f}_{(n,n+1)}^\mathrm{i} \left (\widetilde{X}_{(n,n+1)}(s,t_{n+1}), \widetilde{V}_{(n,n+1)}(s,t_{n+1}), s \right ) = f^\mathrm{i}(X_0, V_0, s).
\end{equation}
Note that if $t \mapsto (\widetilde{X}_{(n,n+1)}(s,t),\widetilde{V}_{(n,n+1)}(s,t))$ were the true characteristics then $\widetilde{f}_{(n,n+1)}^\mathrm{i}(s)$ would be constant to $f^\mathrm{i}(t_{n+1})$. In any case, it holds $\widetilde{f}_{(n,n+1)}^\mathrm{i}(s = t_{n+1}) = f^\mathrm{i}(t_{n+1})$.
Thus we will reduce the projection error by computing and projecting forward a coarse discretization $\widetilde{f}_{(n,n+1)}^\mathrm{i,c}$ of $\widetilde{f}_{(n,n+1)}^\mathrm{i}$, instead of projecting directly the coefficients $f^\mathrm{i,c}$.
To achieve this, we propose the following restriction scheme.

\begin{enumerate}
	\item At initialization of the fine-scale PIC time stepper, e.g. at time $t_n = n \Delta t$, we store the electric field $E_n = -\nabla_x \phi_n$.
	\item Let $\delta t$ be the fine-scale time step. For each fine scale step, $p = 0,\dots,N_e$, the coarse description of the ion distribution function is computed as follows.
	Iterating through all particles representing ions, the phase space position of each is integrated on the approximate characteristics~\eqref{eq:approxcharacteristics}, with $s = t_n + p \delta t$:
	\begin{equation}\label{def:explicitflow}
\left \{\begin{aligned}
\widetilde{ X^\mathrm{i}_j} = X^\mathrm{i}_j +\left (\Delta t - p \delta t \right ) \cdot V^\mathrm{i}_j,\\
\widetilde{ V^\mathrm{i}_j} = V^\mathrm{i}_j + \frac{e}{m_\mathrm{i}}\left (\Delta t - p \delta t \right ) \cdot E_n \left ( X^\mathrm{i}_j \right ).
\end{aligned}\right.
\end{equation}
The particle is deposited on the grid by updating the values of the histogram on the fine grid:
\begin{equation} \left \{
\begin{aligned}
k_x & = \left \lfloor \widetilde{ X^\mathrm{i}_j} / h \right \rfloor, \qquad \qquad \theta_x = \widetilde{ X^\mathrm{i}_j} / h - k_x, \qquad
k_v &= \left \lfloor (\widetilde{V^\mathrm{i}_j}-V_{min})/\delta v \right \rfloor,\\
 [ \widetilde{f}_{(n,n+1)}^\mathrm{i} ]_{k_x, k_v} &\leftarrow [ \widetilde{f}_{(n,n+1)}^\mathrm{i} ]_{k_x, k_v} + (1 - \theta_x) \cdot \omega^\mathrm{i}_j, \\
 [ \widetilde{f}_{(n,n+1)}^\mathrm{i} ]_{k_x+1, k_v} &\leftarrow  [ \widetilde{f}_{(n,n+1)}^\mathrm{i}  ]_{k_x+1, k_v} + \theta_x \cdot \omega^\mathrm{i}_j.
\end{aligned} \right.
\end{equation}
\item Finally, the linear smoothing scheme~\eqref{def:restriction} and the discrete wavelet transform are used to coarsen the phase space resolution of the data as described above.
\end{enumerate}

\section{Quiet start particle loading and momentum-conserving histogram reconstruction}
\label{sec:QuietStart}
At the particle loading stage of the projection method presented in Section~\ref{sec:wefree}, the ion distribution function $f^\mathrm{i}$ is known by its scaling coefficients on the fine grid. Using the approximation~\eqref{def:ScalingCoeffHistogram}, we recover the approximated histogram of the ion distribution function:
\[
	f^\mathrm{i,f} \equiv \left \{ f^\mathrm{i,f}_{k_x, k_v} \right \}_{\stackrel{1 \leq k_x \leq n_{x,\mathrm{f}}}{ 1 \leq k_v \leq n_{v,\mathrm{f}}}}.
\]
Due to the extrapolation process or the properties of the wavelets, some values of the approximate histogram may be negative. 
This results in a non-monotone CCDF, which is nonphysical and must be corrected. 
We propose here a simple reconstruction method which conserves momentum (except in pathological cases) and is found to give good results in practice.
\begin{remark}
	Note that the reconstruction algorithm described hereafter can also be applied after lifting the histogram to the fine grid in velocity, but before lifting it to the fine grid in space. Since it is applied to the velocity histogram for each grid point in space, this can save some computational time if this process becomes a major factor in the total computational time (this was not the case in our numerical experiments).
\end{remark}
Let us compute for a given index of the spatial grid $1 \leq k_x \leq n_{x,\mathrm{f}}$ the corresponding value of the ion density,
\begin{equation}
	n^\mathrm{i,f}_{k_x} = \sum_{k_v = 1} ^{n_{v, \mathrm{f}}} f^\mathrm{i,f}_{k_x,k_v},
\end{equation}
and the values of the normalized conditional cumulative distribution function $F_v$ at points $0 \leq k_v \leq n_{v,\mathrm{f}}$: 
\begin{equation}
F_v[ k_v \ \lvert\  k_x] = \frac{1}{n^\mathrm{i}_{k_x}} \sum_{ l = 1 }^{k_v} f^\mathrm{i,f}_{k_x,k_v}.
\end{equation}
For simplicity, we drop the dependence of $n^\mathrm{i,f}$ and $F_v$ on $k_x$ in the rest of the paragraph.
Note that $F_v[0] = 0$ and $F_v[n_{v,\mathrm{f}}] = 1$. In a first step, we modify the entries of $F_v$ to ensure that $0 \leq F_v \leq 1$ for all indexes; in a second step, we ensure that the resulting $F_v$ is monotone. Let us describe these two stages of the algorithm in details.

\paragraph{Edge reconstruction}
It is possible that some entries of the array $F_v$ fall outside of the range $[0,1]$ because of overshoot. Let us describe here a procedure to correct the case where some values are negative, which can be easily tranposed to the case of values larger than $1$. 
\begin{Steps}
	\item Find $k_0$, the largest index such that $F_v[k_0] < 0$. Set $k^- = k_0 - 1$, $k^+$ = $k_0 + 1$.
	\item Compute the momentum-conserving slope (see Lemma~\ref{lem:MomentumConservation}),
	\begin{equation}\label{eq:Slope1}
		\mathrm{d}F = \frac{ 2 \sum_{k=1}^{k^+-1} F_v[k] }{ (k^+ - k^-)(k^+ - k^- - 1 ) }.
	\end{equation}
	\begin{itemize}
		\item If $0 \leq (k^+ - k^- - 1 )\mathrm{d}F \leq F_v[k^+]$, then go to the next step;
		\item Otherwise, extend the range of indexes between $k^-$ and $k^+$ and recompute $\mathrm{d}F$ accordingly: \\
			If $k^- > 0$, then $k^- \leftarrow k^- - 1$. Otherwise let $k^+ \leftarrow k^+ + 1$ and $k^- \leftarrow k_0 - 1$.
	\end{itemize}
	\item Finally, modify the array $F_v$ by
	\begin{equation}\label{eq:Reconstruction1}
		F_v[k] \leftarrow 0 \text{ for } 0 \leq k < k^-, \qquad \qquad F_v[k] \leftarrow (k - k^-)\mathrm{d}F \text{ for } k^- \leq k < k^+.
	\end{equation}
\end{Steps}

\paragraph{Interior reconstruction}
Let us suppose now that $0 \leq F_v \leq 1$. Suppose that $k_v \mapsto F_v [k_v]$ is not monotone. 
\begin{Steps}
	\item Find $k_0$, the smallest index such that $F_v[k_0] > F_v[k_0 + 1]$. 
		Let $k_1$ be the smallest index such that $F_v[k_0] \leq F_v[k_1]$ and 
		\begin{equation}
				k^+  = \underset{k_0+1 \leq k \leq k_1 }{\mathrm{argmin}}\ F[k] \qquad \text{ and } \qquad
				k_-  = \max \left \{ 0 \leq k \leq k_0 - 1 \ \vert \ F_v[k] \leq F[k^+] \right\}.
		\end{equation}
		\item Compute the momentum-conserving slope:
	\begin{equation}\label{eq:Slope2}
		\mathrm{d}F = \frac{2}{k^+ - k^- - 1} \left ( \frac{\sum_{k=k^-}^{k^+-1} F_v[k]}{k^+ - k^-} - F_v[k_-] \right ).
	\end{equation}
	\begin{itemize}
		\item If $0 \leq (k^+ - k^- - 1 )\mathrm{d}F \leq F_v[k^+]$, then go to the next step;
		\item Otherwise, set $k^+ \leftarrow k^+ + 1$ and recompute $\mathrm{d}F$ accordingly.
	\end{itemize}
	\item Finally, modify the array $F_v$ by
	\begin{equation}\label{eq:Reconstruction2}
		F_v[k] \leftarrow F_v[k_-] + (k - k^-)\mathrm{d}F \qquad \text{ for } k^- < k < k^+.
	\end{equation}
\end{Steps}
Now $F_v$ is monotone up to $k^+ > k_0$. Proceeding recursively, we obtain in all cases a corrected monotone array $F_v^+$ where the $+$ denotes that the correction is computed from the lower to the upper indexes. 
To ensure symmetry of the reconstruction, it is necessary to compute another corrected array $F_v^-$, computed in similar fashion but starting from the higher indexes.
Finally, we define the reconstructed array $\mathcal{R}F_v$ as:
\begin{equation}\label{eq:reconstruction}
\mathcal{R}F_v = ( F_v^+ + F_v^- ) / 2.
\end{equation}
\begin{lemma}\label{lem:MomentumConservation}
The slope defined by Eq.~\eqref{eq:Slope2} (resp. Eq.~\eqref{eq:Slope1}) is designed so that the reconstruction in Eq.~\eqref{eq:Reconstruction2} (resp. Eq.~\eqref{eq:Reconstruction1}) conserves the first two moments of the histogram $f_k = F_v[k] - F_v[k-1]$,
\[
	M_0 = \sum_{k = 1}^{n_{v,\mathrm{f}}} (F_v[k] - F_v[k-1]) \qquad \text{and} \qquad M_1 = \sum_{k = 1}^{n_{v,\mathrm{f}}} k (F_v[k] - F_v[k-1]).
\]
\end{lemma}
\begin{proof}
The moment of order zero is always conserved, since $M_0 = F_v[n_{v,\mathrm{f}}] - F_v[0] = 1$.
For the moment of order one, we investigate only Eq.~\eqref{eq:Reconstruction2} since the proof is the same. The change in $M_1$ can be computed by
\begin{align*}
	M_1^{before} - M_1^{after} &= \sum_{k^- < k  < k^+} (k - k^+) \left ( F_v[k] - F_v[k-1] -\mathrm{d}F \right ) \\
					&= \sum_{k^- < k  < k^+} (k - k^+) F_v[k] - \sum_{k^- \leq k < k^+} (k +1 - k^+) F_v[k] -\left ( \sum_{k' = 1}^{k^+ - k^- - 1} -k' \right )  \mathrm{d}F \\
					&=  \frac{(k^+ - k^-)(k^+ - k^- -1)}{2} \mathrm{d}F  - (k^- - k^+) F_v[k^-] - \sum_{k = k^-}^{k^+ - 1} F_v[k] = 0,
\end{align*}
where we have used the fact that $(k +1 - k^+) F_v[k] = 0$ for $k = k^+ - 1$ and also Eq.~\ref{eq:Slope2}.
\end{proof}
\paragraph{Particle loading} Once all the reconstructed arrays $\left (\mathcal{R} F_v [\ \cdot\ \lvert k_{x} ] \right )_{1 \leq k_x \leq n_{x,\mathrm{f}}}$ have been computed, the quiet loading of the ion particles proceeds as follows.
Recall that each particle is parameterized by its position $X^\mathrm{i}_j$, velocity $V^\mathrm{i}_j$ and weight $\omega^\mathrm{i}_j$.
Let $(a_j, b_j)_{j \in \N} \subset [0,1]^2$ be a low discrepancy sequence (e.g. the Hammersley sequence~\cite{Hammersley1960}), which yields numbers that are more evenly distributed than a randomly generated sequence. Then, we set the initial particle parameters as 
\begin{equation}
\begin{aligned}
X^\mathrm{i}_i &= L \cdot a_i, \\
V^\mathrm{i}_i &= (\mathcal{R} F_v)^{-1}(b_i \ \lvert\ x = X^\mathrm{i}_i) ,\\
\omega^\mathrm{i}_i &= n^\mathrm{i}(X^\mathrm{i}_i), 
\end{aligned}
\end{equation}
where all values are linearly interpolated between the two closest points of the fine grid. The inverse mapping $ (\mathcal{R} F_v)^{-1}(\ \cdot \ \lvert x = X^\mathrm{i}_i)$ is evaluated by lookup and linear interpolation.
To finish initializing the PIC code, the electrons are loaded using the normal quiet start procedure for a shifted Maxwellian function.

\end{document}